\documentclass[aps,pra,twocolumn,superscriptaddress,10pt,longbibliography,floatfix]{revtex4-2}

\usepackage{amsmath,amssymb,amsthm,mathtools}
\usepackage{bm}
\usepackage{graphicx}
\usepackage{booktabs}
\usepackage{xcolor}
\usepackage{enumitem}
\usepackage[caption=false]{subfig}
\usepackage{tikz}
\usepackage{quantikz}
\usepackage[ruled,vlined,linesnumbered]{algorithm2e}

\usepackage{orcidlink}

\usetikzlibrary{arrows.meta,positioning,calc,decorations.pathreplacing,fit,backgrounds,shapes.geometric}

\setlist{nosep,leftmargin=*}

\hypersetup{
	colorlinks=true,
	linkcolor=blue!70!black,
	citecolor=green!50!black,
	urlcolor=blue!60!black
}

\newcommand{\R}{\mathbb{R}}
\newcommand{\E}{\mathbb{E}}
\newcommand{\Var}{\mathrm{Var}}
\newcommand{\calO}{\mathcal{O}}
\newcommand{\calL}{\mathcal{L}}

\newcommand{\calC}{\mathcal{C}}
\newcommand{\btheta}{\bm{\theta}}
\newcommand{\by}{\bm{y}}
\newcommand{\ba}{\bm{a}}
\newcommand{\bW}{\bm{W}}
\newcommand{\bA}{\bm{A}}
\newcommand{\bM}{\bm{M}}
\newcommand{\bF}{\bm{F}}

\newcommand{\bI}{\bm{I}}
\newcommand{\bs}{\bm{s}}

\newcommand{\GEANT}{\textsc{Geant4}}

\newtheorem{theorem}{Theorem}

\begin{document}
	
	\title{Quantum Feature Amplification Network (QFAN) as An Autoregressive Quantum Generative Model}
	
	\author{Jamal Slim\orcidlink{0000-0002-9418-8459}}
\email{jamal.slim@desy.de}
\affiliation{Deutsches Elektronen-Synchrotron DESY, 22603 Hamburg, Germany}
\author{Saverio Monaco\orcidlink{0000-0001-8784-5011}}
\affiliation{Deutsches Elektronen-Synchrotron DESY, 22603 Hamburg, Germany}
\affiliation{RWTH Aachen University, 52062 Aachen, Germany}
\author{Florian Rehm\orcidlink{0000-0002-8337-0239}}
\affiliation{European Organization for Nuclear Research (CERN), 1211 Geneva, Switzerland}
\author{Dirk Kr\"ucker\orcidlink{0000-0003-1610-8844}}
\affiliation{Deutsches Elektronen-Synchrotron DESY, 22603 Hamburg, Germany}
\author{Kerstin Borras\orcidlink{0000-0003-1111-249X}}
\affiliation{Deutsches Elektronen-Synchrotron DESY, 22603 Hamburg, Germany}
\affiliation{RWTH Aachen University, 52062 Aachen, Germany}
\date{\today}

	\begin{abstract}
		Gate-based quantum generative models have been limited to small images, because their quantum resources grow with the size of the data. The Quantum Feature Amplification Network (QFAN) removes this limit. It reuses a single few-qubit circuit to generate the data block by block, each block conditioned on a compact summary of those before it, with the circuit's measurement outcomes as the source of randomness. The number of qubits is set by the block rather than by the data, and the number of circuits per training step does not grow with the data either. In simulations with only four qubits, QFAN generates full-resolution calorimeter showers from the CaloChallenge Dataset~1 benchmark, 368 voxels each, and learns their correlations and shower-to-shower fluctuations. A three-qubit version generates 12- and 25-pixel images on IBM quantum hardware. This brings high-dimensional scientific data within reach of near-term quantum generative models.
	\end{abstract}
	\maketitle
	
	% ═══════════════════════════════════════════════════════
	\section{Introduction}\label{sec:intro}
	% ═══════════════════════════════════════════════════════
	
	Most particles entering a calorimeter initiate a stochastic cascade of secondary interactions, depositing energy across hundreds to thousands of detector cells. Simulating these showers from first principles with \GEANT{}~\cite{geant4} is one of the largest consumers of computing in LHC experiments, and projections for the High-Luminosity LHC era place the annual budget at more than $10^{11}$ events~\cite{hllhc_computing}, a demand that cannot be met with current resources.
	Classical surrogate models based on deep generative networks~\cite{calogan,caloflow,calodiffusion,atlasfast3} have achieved speedups of $10^3$ to $10^6$, but they use a larger number of trainable parameters, typically $10^6$ to $10^7$, and rely often on GPU inference.
	
	Quantum generative models offer a different design point. Compact circuits whose correlation structure is carried by entanglement rather than by depth of a classical network. In principle a few-qubit circuit can encode nontrivial dependencies with far fewer parameters than a classical generative network.
	In practice, existing gate-model calorimeter proof-of-concepts and quantum-assisted generators~\cite{chang2021_dualpqcgan,rehm2023_fullqgan,caloqvae2024,toledomarin2025} are limited either to strongly downscaled images or to hybrid latent-variable constructions.
	The reason is a resource bottleneck. If a circuit emits the image pixels directly from its register, then the register must be as large as the image. Angle-encoded models need $n_q=\Omega(d)$ qubits for a $d$-pixel image. Amplitude-encoded models need only $\lceil\log_2 d\rceil$ qubits, but they output a probability distribution rather than real-valued intensities, and preparing a generic state costs circuit depth growing with $d$. Either way the cost grows with the image.
	
	The idea behind QFAN is that the circuit does not need to hold the whole image at once.
	We split the $d$ pixels into $B$ consecutive blocks of at most $b$ pixels, and generate the blocks one after another with the \emph{same} $n_q$-qubit circuit (Fig.~\ref{fig:paradigm}).
	The register then has to be large enough for one block, not for the image, so $n_q$ is set by $b$ and not by $d$.
	Between blocks, a streaming count-sketch~\cite{charikar2002} compresses all previously generated pixels into a fixed-length vector that conditions the next circuit call, and the circuit's Pauli expectation values are mapped to pixel intensities by a ridge regressor with a closed-form solution.
	Because the decoder is closed-form, the loss remains an explicit function of the circuit parameters, and we optimize it with exact analytic gradients rather than a stochastic-perturbation estimator. No bilevel optimization is required.
	Figure~\ref{fig:outline} unrolls one full generation pass and shows explicitly how the same circuit is reused across blocks.
	
	This buys register size at the cost of sequential depth. A larger image is handled by more autoregressive steps rather than by more qubits. Two consequences frame the rest of the paper. First, for the Pauli family we measure, the number of circuits executed per training step does not depend on $d$ (Sec.~\ref{sec:training}). Second, the accuracy of the blockwise decoder limits how many pixels one block may contain, which in turn sets how many blocks a given image needs. We quantify this with a measured feature-to-target ratio and use it, cautiously, for extrapolation (Sec.~\ref{sec:outlook}).
	
	We validate QFAN at $d{=}12$ pixels with a 3-qubit circuit, on a noiseless simulator and on IBM \texttt{ibm\_fez} hardware, and Appendix~\ref{app:25px} extends both to 25-pixel images. The architecture reproduces per-pixel marginals, inter-pixel correlations and the total deposited energy on both backends. Freezing the circuit at its untrained initialization and refitting every classical stage leaves a model that reproduces neither the marginals nor the correlations, which establishes that the trained circuit produces both (Sec.~\ref{sec:certificates}). Section~\ref{sec:ds1} then applies the same architecture, without enlarging the register beyond four qubits, to all 368 voxels of the CaloChallenge Dataset~1 photon showers.
	
	Reducing the register requirement this far also determines the scale at which a physical device becomes necessary. Section~\ref{sec:simulability} quantifies that scale. It shows that within the present architecture large registers are demanded only when the autoregressive chain must be kept short, quantifies that trade, and argues that width alone does not buy classical hardness at this depth and observable weight. It then identifies the modification that would change the picture, namely using the circuit as a sampler rather than as an estimator of few-body expectation values, and states what that would cost. The two are linked by the growth of exposure bias with chain length, which decides not only whether the method scales but whether it has a quantum motivation at all. Section~\ref{sec:ds1} measures generation along a chain of 185 blocks.
	
	Training is teacher-forced, meaning each block is conditioned on ground-truth prefixes through a pre-computed sketch cache, whereas generation is free-running, with each block conditioned on the model's own previous outputs. At the depth of the 12-pixel demonstration ($B{=}6$) the difference is mild, but for the much longer chains implied by benchmark geometries the accumulation of model bias along the chain, known as exposure bias, is a plausible dominant failure mode. Our shot-noise analysis bounds measurement noise only. Section~\ref{sec:ds1} tests free-running generation along a chain of $B{=}185$ blocks at the full resolution of the CaloChallenge Dataset~1 benchmark.
	
	\begin{figure}[tb]
		\centering
		\resizebox{\columnwidth}{!}{%
			\begin{tikzpicture}[scale=0.65, font=\small,
				circ/.style={draw, circle, minimum size=0.35cm, fill=#1, inner sep=0pt},
				lbl/.style={font=\tiny, text=gray!70!black}]
				\node[font=\scriptsize\bfseries] at (2.2, 6.0) {Prior direct-register QML: $n_q{=}d$};
				\foreach \row in {0,1,2} {\foreach \col in {0,1,2,3} {
						\pgfmathtruncatemacro{\idx}{4*\row+\col}
						\node[circ=red!25] (qL\idx) at ({0.5+\col*0.75}, {4.6-\row*0.75}) {};}}
				\node[draw, dashed, red!50, rounded corners, fit=(qL0)(qL11), inner sep=5pt] (regL) {};
				\node[lbl, below=8pt of regL,font=\small] {12 qubits for $d{=}12$};
				\draw[->, thick, red!60] (regL.south) -- ++(0,-0.55) node[below=10pt, lbl,font=\small, align=center] {$\tilde{\mathbf{y}}\in\mathbb{R}^{12}$};
				\draw[gray!40, thick, dashed] (4.5,1) -- (4.5,5.5);
				\node[draw, fill=white, rounded corners, font=\scriptsize\bfseries] at (4.5,3.2) {vs};
				\node[font=\scriptsize\bfseries] at (8.2, 6.0) {QFAN: $n_q{\ll}d$};
				\node[circ=blue!40] (qR0) at (8.3,5.1) {};
				\node[circ=blue!40] (qR1) at (9.1,5.1) {};
				\node[circ=blue!40] (qR2) at (9.9,5.1) {};
				\node[draw, blue!60, rounded corners, fit=(qR0)(qR2), inner sep=4pt] (regR) {};
				\draw[->, thick, blue!60] (regR.south) -- ++(0,-0.3) -| (6.8,3.8) node[pos=0.55, below=7pt, lbl,font=\small] {block 1,};
				\draw[->, thick, blue!60] (regR.south) -- ++(0,-0.7) node[pos=1.05, below, lbl,font=\small] {block 2, \ldots,};
				\draw[->, thick, blue!60, dashed] (regR.south) -- ++(0,-0.3) -| (11.4,3.8) node[pos=0.55, below=7pt, lbl,font=\small] {block $B$};
				\foreach \i in {0,...,11} {\fill[blue!20] ({6.8+\i*0.3},2.2) rectangle ++(0.25,0.45); \draw[blue!40] ({6.8+\i*0.3},2.2) rectangle ++(0.25,0.45);}
				\node[lbl,font=\small] at (9.2,1.5) {$\tilde{\mathbf{y}}\in\mathbb{R}^{12}$ (assembled from $B$ blocks)};
			\end{tikzpicture}
		}
		\caption{Paradigm comparison. Prior direct-register QML requires $n_q{=}d$ qubits (left). QFAN reuses a fixed $n_q$-qubit circuit across $B$ blocks (right), so the qubit requirement is set by the block size $b$ rather than by the full image dimension $d$.}
		\label{fig:paradigm}
	\end{figure}
	
	% --- Outline figure: the chain, and the two ways to close it ---
	\begin{figure*}[tb]
		\centering
		\resizebox{0.99\textwidth}{!}{%
			\begin{tikzpicture}[
				font=\small,
				q/.style={draw, rounded corners=2pt, fill=orange!16, minimum height=0.78cm,
					minimum width=1.65cm, font=\small},
				m/.style={draw, rounded corners=2pt, fill=blue!10, minimum height=1.02cm,
					minimum width=1.6cm, font=\scriptsize, align=center},
				r/.style={draw, rounded corners=2pt, fill=red!8, minimum height=0.78cm,
					minimum width=2.45cm, font=\scriptsize, align=center},
				g/.style={draw=gray!55, rounded corners=2pt, fill=gray!8, minimum
					height=0.78cm, minimum width=2.45cm, font=\scriptsize,
					align=center, text=gray!60!black},
				w/.style={draw, rounded corners=2pt, fill=yellow!25, minimum height=0.78cm,
					minimum width=1.05cm, font=\small},
				qw/.style={-{Stealth[length=2mm]}, thick},
				cw/.style={-{Stealth[length=2mm]}, thick, double, double distance=1pt},
				gw/.style={-{Stealth[length=2mm]}, thick, gray!55, double,
					double distance=1pt},
				fb/.style={-{Stealth[length=2mm]}, thick, blue!70},
				note/.style={font=\scriptsize, text=gray!60!black, align=center}]
				
				\newcommand{\bornicon}[2]{%
					\begin{scope}[shift={($(#1.center)+(0,-0.42)$)}]
						\foreach \h [count=\i from 0] in {#2}
						{\fill[blue!55] ({\i*0.078-0.28},0) rectangle ++(0.056,{\h*0.34});}
						\draw[gray!55,line width=0.3pt] (-0.30,0) -- (0.33,0);
				\end{scope}}
				% a small shower strip: 6 cells with given intensities
				\newcommand{\shower}[3]{% #1=x #2=y #3=intensities
					\begin{scope}[shift={(#1,#2)}]
						\foreach \v [count=\i from 0] in {#3}
						{\fill[blue!\v!white,draw=gray!45,line width=0.2pt]
							({\i*0.16},0) rectangle ++(0.15,0.30);}
				\end{scope}}
				
				\foreach \b/\name/\yy in {0/1/0, 1/2/-1.85, 2/B/-3.70} {
					\node (in\b) at (0,\yy) {$\ket{0}^{\otimes n_q}$};
					\node[q] (U\b) at (2.05,\yy) {$U(\ba_{\name},\btheta)$};
					\node[m] (M\b) at (4.90,\yy) {measure\\$Z,X,Y$\\[6pt]};
					\node[r] (Rc\b) at (7.70,\yy)
					{$\bar{\bm f}_{\name}=\tfrac1k\sum_{j}\bm r_j$};
					\node[w] (W\b) at (10.35,\yy) {$\bW_{\name}$};
					\node (out\b) at (11.85,\yy) {$\hat{\by}_{\name}$};
					\draw[qw] (in\b) -- (U\b);
					\draw[qw] (U\b) -- node[above,note]{$\ket{\psi_{\name}}$} (M\b);
					\draw[cw] (M\b) -- node[above,note]{$\times k$} (Rc\b);
					\draw[cw] (Rc\b) -- (W\b);
					\draw[cw] (W\b) -- (out\b);
				}
				\bornicon{M0}{0.45,1.0,0.18,0.72,0.30,0.10,0.55,0.22}
				\bornicon{M1}{0.95,0.22,0.60,0.15,0.80,0.35,0.12,0.48}
				\bornicon{M2}{0.25,0.68,1.0,0.30,0.14,0.52,0.20,0.38}
				\foreach \a/\b in {0/1, 1/2} {
					\draw[fb] (out\a.south) -- ++(0,-0.36) -| ($(U\b.north)+(-0.50,0)$)
					node[pos=0.25,below,note]{sketch update};
				}
				\draw[decorate, decoration={brace, amplitude=4pt}, gray!65]
				($(U0.north west)+(0,0.10)$) -- ($(U0.north east)+(0,0.10)$);
				\node[note,anchor=south,text width=2.3cm] at ($(U0.north)+(-0.55,0.26)$)
				{one circuit, one $\btheta$,\\reused every step};
				\node[note,anchor=south,text width=3.0cm] at ($(M0.north)+(0.35,0.18)$)
				{Born distribution\\$|\langle x|\psi_\beta\rangle|^{2}$, differs per row};
				
				% ===================== the two ways to close the loop ============
				\coordinate (hub) at (13.6,-1.85);
				% -- v5.0 branch: keep the records -> a distribution of showers
				\node[note,anchor=west,align=left,text width=4.0cm] at (13.30,-0.28)
				{\textbf{keep the $k$ records}};
				\shower{13.45}{-0.95}{15,55,90,70,35,12}
				\shower{13.63}{-1.10}{22,70,80,55,28,10}
				\shower{13.81}{-1.25}{10,45,100,82,40,18}
				\node[note,anchor=west,align=left,text width=3.7cm] at (15.05,-1.10)
				{run the whole generator\\three times: three
					\emph{different}\\showers, $\Var(E)>0$};
				\draw[cw] (out0.east) -- ++(0.35,0) |- (13.45,-1.20);
				
				% -- conventional branch: replace by expectations -> one image
				\node[g] (Ehat) at (14.55,-2.88)
				{use $\langle\mathcal O\rangle$ instead of $\bar{\bm f}_\beta$
					\ (i.e.\ $k\!\to\!\infty$)};
				\shower{13.45}{-4.20}{16,58,92,72,36,14}
				\shower{13.63}{-4.35}{16,58,92,72,36,14}
				\shower{13.81}{-4.50}{16,58,92,72,36,14}
				\node[note,anchor=west,align=left,text width=3.7cm] at (15.05,-4.35)
				{run the whole generator\\three times: the
					\emph{identical}\\shower, $\Var(E)=0$ (Sec.~\ref{sec:certificates})};
				\draw[gw] (Rc2.east) -- ++(0.30,0) |- (Ehat.west);
				\draw[gw] (Ehat.south) -- (14.55,-3.78);
				\node[red!70,font=\normalsize] at (14.55,-3.52) {$\times$};
				\draw[red!55,dashed,rounded corners]
				($(Ehat.north west)+(-0.20,0.16)$) rectangle (18.5,-5);
			\end{tikzpicture}
		}%

		\caption{QFAN unrolled, and the single design choice that makes it generative. Each row is one autoregressive step and runs the \emph{same} circuit $U$ with the \emph{same} parameters $\btheta$. Only the angles $\ba_\beta$ differ, computed from a fixed-length sketch of the pixels already generated (blue). Conditioning acts by reshaping the prepared state $\ket{\psi_\beta}$, so the Born distribution $|\langle x|\psi_\beta\rangle|^{2}$ over measurement outcomes differs from row to row (mimicked by the miniature histograms). The circuit is then sampled $k$ times, giving parity vectors $\bm r_1,\ldots,\bm r_k\in\{\pm1\}^{p_f}$, and the block is decoded from their finite-$k$ average $\bar{\bm f}_\beta=\frac1k\sum_j\bm r_j$ [Eq.~\eqref{eq:fbar}]. \emph{Right.} The two ways of closing this loop, which is where the present architecture departs from a variational feature model. Keeping the records (top) leaves the sampling fluctuation in the output. Running the generator from the same starting point three separate times then yields three \emph{different} showers, because each run draws its own measurement records. This fluctuation is the model's only stochastic element and the resulting spread scales as $k^{-1/2}$. Replacing $\bar{\bm f}_\beta$ by the corresponding expectation values $\langle\mathcal O\rangle$ (bottom, greyed) removes the shot noise, and with it the model's only stochastic element. The rollout becomes a deterministic function of the initial sketch, so the same three runs return the \emph{identical} shower and every pixel has zero variance (Sec.~\ref{sec:certificates}). Enlarging the image adds rows on the left, not qubits to the register.}
		\label{fig:outline}
	\end{figure*}
	
	The paper proceeds as follows.
	Section~\ref{sec:problem} introduces the generation task and the autoregressive decomposition.
	Section~\ref{sec:architecture} follows the data through the pipeline.
	Section~\ref{sec:training} covers training and per-step cost, and Sec.~\ref{sec:noise} the propagation of shot noise.
	Section~\ref{sec:dual} describes the two execution paths and Sec.~\ref{sec:results} the results at $d{=}12$.
	Section~\ref{sec:ds1} applies the same architecture, with four qubits, to all 368 voxels of the CaloChallenge Dataset~1 photon showers and derives the limit on the dispersion of decoded measurement records.
	Section~\ref{sec:certificates} isolates the contribution of the quantum circuit from that of the classical decoder.
	Section~\ref{sec:outlook} collects the resource outlook, including the hardware transfer gap and the extrapolations to benchmark geometries, all of which are explicitly labeled as projections.
	Section~\ref{sec:comparison} positions QFAN among prior work, and Sec.~\ref{sec:next} sets out the next steps.
	Appendix~\ref{app:25px} reports the 25-pixel images on the simulator and on hardware.

	% ═══════════════════════════════════════════════════════
	\section{Problem Formulation}\label{sec:problem}
	% ═══════════════════════════════════════════════════════
	
	A calorimeter shower image $\by\in\R^d$ records the energy deposited in $d$ detector cells.
	The target distribution $p_{\mathrm{data}}$ combines two difficulties. Each cell marginal is strongly non-Gaussian, typically unimodal with a skewed tail whose width and peak position vary from cell to cell. In addition, the cells are strongly correlated across detector layers through the longitudinal and lateral development of the shower.
	A useful surrogate must reproduce both. Matching the marginals while producing independent pixels is not sufficient for downstream analysis.
	
	The constraint that drives every design decision here is $n_q\ll d$. At twelve pixels the ratio is modest ($n_q{=}3$, $d{=}12$). Benchmark calorimeter images are far larger~\cite{Krause_2025}, and in Sec.~\ref{sec:ds1} four qubits generate 368 voxels.
	
	Our approach uses the standard autoregressive factorization of an image distribution, as used by classical pixel-recurrent and pixel-convolutional models~\cite{pixelrnn}.
	Partition the $d$ pixels into $B$ contiguous blocks $\by_1,\ldots,\by_B$ of at most $b$ pixels each, and write $\by_{<\beta}\equiv(\by_1,\ldots,\by_{\beta-1})$ for the pixels generated before block $\beta$. The joint distribution factors exactly as
	\begin{equation}
		\label{eq:factorization}
		p_{\btheta}(\by)=\prod_{\beta=1}^{B} p_{\btheta}\!\left(\by_{\beta}\,\big|\,\by_{<\beta}\right).
	\end{equation}
	The difference from a classical autoregressive model is where the conditional lives. Each factor $p_{\btheta}(\by_\beta|\by_{<\beta})$ is produced by one invocation of a small quantum circuit, and the conditioning history is compressed into a fixed-length sketch so that the circuit input does not grow with $d$.
	Throughout, $b$ denotes the block size in pixels and $B=\lceil d/b\rceil$ the number of blocks, i.e.\ the number of autoregressive steps.
	
	% ═══════════════════════════════════════════════════════
	\section{Architecture}\label{sec:architecture}
	% ═══════════════════════════════════════════════════════
	
	Figure~\ref{fig:outline} shows the pipeline unrolled. We now describe each stage in details.
	
	\subsection{Block decomposition}\label{sec:blocks}
	
	The image is partitioned into $B{=}\lceil d/b\rceil$ contiguous blocks of at most $b$ pixels. At $d{=}12$ with $b{=}2$ this gives $B{=}6$. The choice of $b$ is a trade-off between two modes. Large blocks starve the ridge decoder, which can produce at most $p_f$ independent output directions from $p_f$ measured features (Sec.~\ref{sec:ridge}). Small blocks lengthen the chain, so sketch distortion and model bias accumulate over more steps (Sec.~\ref{sec:noise}). In Sec.~\ref{sec:outlook} we quantify this trade-off and use it to set the block size at larger image sizes.
	
	\subsection{Sketch conditioning}\label{sec:sketch}
	
	Equation~\eqref{eq:factorization} requires block $\beta$ to be conditioned on all pixels generated before it. Feeding the raw prefix $\by_{<\beta}$ into the circuit would make the circuit input grow with $d$, which is exactly what we are trying to avoid. We therefore compress the prefix into a fixed-length vector.
	
	The compression is a count-sketch~\cite{charikar2002} into a fixed-length vector $\bs\in\R^{m}$, whose $m{=}32$ entries are accumulator buckets. The length $m$ is set once and does not grow with $d$, which is what keeps the circuit input 	independent of the image size. Two fixed random functions are drawn once at initialization and never changed. A hash $h:[d]\to[m]$ that assigns each pixel to one of $m$ buckets, and a sign $\varsigma:[d]\to\{\pm1\}$. After block $\beta$ is generated, the sketch is updated by adding each new pixel into its bucket with its sign,
	\begin{equation}
		\label{eq:sketch_update}
		s_j \mathrel{+}= \sum_{k\in\mathrm{block}\,\beta} \varsigma(k)\,\tilde{y}_k\,\mathbf{1}[h(k){=}j].
	\end{equation}
	
	%	In other words, the sketch is a fixed-length running summary of the shower so far, in which each pixel contributes to exactly one entry with a random sign. Randomized signs keep the summary unbiased, in the sense that inner products between sketched vectors estimate inner products between the original vectors (Appendix~\ref{app:sketch}). Before the sketch enters the circuit it passes through a fixed near-identity mixing layer $\tilde{\bs}{=}\tanh(\bs\bM^\top)$, with $\bM{=}\bI_m{+}\varepsilon\bm{E}$ for a fixed random $\bm{E}$ and $\varepsilon{\ll}1$, which spreads out the effect of hash collisions.

	In other words, the sketch is a fixed-length running summary of the shower so far, in which each pixel contributes to exactly one entry with a random sign. Randomized signs keep the summary unbiased, in the sense that inner products between sketched vectors estimate inner products between the original vectors (Appendix~\ref{app:sketch}). Before the sketch enters the circuit it passes through a fixed mixing layer, $\tilde{\bs}{=}\tanh(\bs\bM^\top)$. The matrix $\bM$ is a small perturbation of the identity, $\bM{=}\bI_m{+}\varepsilon\bm{E}$, where $\bI_m$ is the $m\times m$ identity, $\bm{E}$ is a random matrix drawn once at initialization and held fixed, and $\varepsilon{\ll}1$ sets the size of the perturbation. Remaining close to the identity preserves the bucket structure of the sketch, while the off-diagonal terms spread each hash collision across several circuit angles instead of concentrating it in one.
	
	We use a count-sketch rather than a learned embedding for three reasons. Its update costs $\calO(b)$ per block and $\calO(m)$ memory, both independent of how many pixels have already been generated. It adds no trainable parameters, so training remains a single-level problem in $\btheta$. And because it is a streaming summary, the teacher-forced sketches for every block can be precomputed once into a cache $\calC\in\R^{B\times N\times m}$ and reused at every training step.
	The hash and sign functions, the mixing matrix $\bM$, and are fixed at initialization. The iteratively trained quantities are the circuit parameters $\btheta$, the encoding $(\bA,\bm{b})$, and the per-block record shares $\rho_\beta$.
	
	\subsection{Parameterized quantum circuit}\label{sec:pqc}
	
	The circuit $U(\ba,\btheta)$ acts on $n_q{=}3$ qubits with $L{=}2$ layers at $d{=}12$ (Fig.~\ref{fig:circuit}) and $L{=}3$ layers at $d{=}25$ (Appendix~\ref{app:25px}), and on $n_q{=}4$ qubits with $L{=}3$ layers for Dataset~1 (Sec.~\ref{sec:ds1}). Each layer interleaves data re-uploading~\cite{dataReuploading} of the sketch-derived angles $\ba{=}\sigma(\tilde{\bs}\bA^\top{+}\bm{b})\in[0,1]^{L_a}$ with trainable $R_ZR_Y$ rotations, followed by a CZ ring $\mathrm{CZ}_{01},\mathrm{CZ}_{12},\mathrm{CZ}_{20}$.
	Here $\sigma$ is the elementwise logistic function, $\bA\in\R^{L_a\times m}$ and $\bm{b}\in\R^{L_a}$ are a trained projection and bias with $L_a{=}8$ at $d{=}12$ and $L_a{=}16$ at $d{=}25$ and for Dataset~1.
	The trainable circuit parameter count is $p{=}2Ln_q$, giving 12 at $d{=}12$, 18 at $d{=}25$ and 24 for Dataset~1, shared across all blocks. The encoding $(\bA,\bm{b})$ is trained jointly with $\btheta$ rather than held fixed.
	
	The encoding map, $\bA$, $\bm{b}$ and the logistic squashing, is an engineering choice. It is fixed, cheap, and keeps angles in a bounded range. A different fixed encoder of the same width is equally admissible.
	
	\begin{figure}[tb]
		\centering
		\resizebox{\columnwidth}{!}{%
			\begin{quantikz}[row sep={0.85cm,between origins}, column sep=0.4cm]
				\lstick{$\ket{0}$} & \gate{R_Y}\gategroup[3,steps=2,style={dashed,rounded corners,fill=blue!10},background,label style={label position=above,anchor=south}]{$\boldsymbol{a}$} & \gate{R_Z} & \gate{R_Z}\gategroup[3,steps=2,style={dashed,rounded corners,fill=orange!15},background,label style={label position=above,anchor=south}]{$\boldsymbol{\theta}$} & \gate{R_Y} & \ctrl{1}\gategroup[3,steps=3,style={dashed,rounded corners,fill=gray!12},background,label style={label position=above,anchor=south}]{CZ ring} & \qw & \control{} & \gate{R_Y}\gategroup[3,steps=2,style={dashed,rounded corners,fill=blue!10},background,label style={label position=above,anchor=south}]{$\boldsymbol{a}$} & \gate{R_Z} & \gate{R_Z}\gategroup[3,steps=2,style={dashed,rounded corners,fill=orange!15},background,label style={label position=above,anchor=south}]{$\boldsymbol{\theta}$} & \gate{R_Y} & \ctrl{1}\gategroup[3,steps=3,style={dashed,rounded corners,fill=gray!12},background,label style={label position=above,anchor=south}]{CZ ring} & \qw & \control{} & \meter{} \\
				\lstick{$\ket{0}$} & \gate{R_Y} & \gate{R_Z} & \gate{R_Z} & \gate{R_Y} & \control{} & \ctrl{1} & \qw & \gate{R_Y} & \gate{R_Z} & \gate{R_Z} & \gate{R_Y} & \control{} & \ctrl{1} & \qw & \meter{} \\
				\lstick{$\ket{0}$} & \gate{R_Y} & \gate{R_Z} & \gate{R_Z} & \gate{R_Y} & \qw & \control{} & \ctrl{-2} & \gate{R_Y} & \gate{R_Z} & \gate{R_Z} & \gate{R_Y} & \qw & \control{} & \ctrl{-2} & \meter{}
			\end{quantikz}
		}
		\vspace{1pt}
		\caption{Variational circuit $U(\ba,\btheta)$ for $n_q{=}3$, $L{=}2$. In each layer the first two gates per qubit encode the sketch angles $\ba$ (data re-uploading), the next two are trainable $R_ZR_Y$ rotations (12 parameters in total), and the last three columns form the CZ ring. The circuit is measured in the $Z$, $X$ and $Y$ bases ($G{=}3$ settings). At $d{=}25$ the same structure is used with $L{=}3$, and for Dataset~1 with $n_q{=}4$ and $L{=}3$.}
		\label{fig:circuit}
	\end{figure}
	
	Sharing $\btheta$ across blocks is deliberate. The same physical processes govern each detector region, with only the local energy budget, carried by the sketch, changing between blocks. Sharing reduces the trainable circuit-parameter count from $pB$ to $p$ and keeps the cost of an exact gradient step independent of the chain length.
	
	\subsection{Born measurement records}\label{sec:pauli}
	
	This is the stage at which the present architecture differs most from a conventional variational feature model, so we state it carefully.
	
	The circuit is executed $k$ times per block per sample. Each execution returns a bitstring, and from each bitstring we form the parity vector of the measured observables. We call one such vector a \emph{record}. For the observable family
	\[
	\mathcal O=\{P_i\}_{i=1}^{n_q}\cup\{P_iP_j\}_{1\le i<j\le n_q},\qquad P\in\{Z,X,Y\},
	\]
	each of the $G{=}3$ tensor-product settings supplies $n_q+\binom{n_q}{2}$ parities from the same bitstring, giving
	\begin{equation}\label{eq:pf}
		p_f = G\left(n_q+\binom{n_q}{2}\right)=\tfrac{3}{2}n_q(n_q+1),
	\end{equation}
	which is $p_f{=}18$ at $n_q{=}3$ and $p_f{=}30$ at $n_q{=}4$.
	
	Write $\bm r_1,\ldots,\bm r_k\in\{-1,+1\}^{p_f}$ for the parity vectors of the $k$ records, and
	\begin{equation}\label{eq:fbar}
		\bar{\bm f}_\beta \;=\; \frac1k\sum_{j=1}^{k}\bm r_j \;\in\;[-1,1]^{p_f}
	\end{equation}
	for their average, which is the quantity the decoder receives. It is an unbiased estimator of the Pauli expectation vector, $\E[\bar{\bm f}_\beta]=\langle\mathcal O\rangle$, but it is used at finite $k$ and never replaced by that expectation. The distinction is the subject of the rest of this subsection and of the conditional-mean substitution.
	
	The decisive point is what is done with the records. A conventional model would average all $k$ records into an estimate of $\langle\mathcal O\rangle$ and decode the estimate, treating shot noise as an error to be suppressed. Here the finite-$k$ record average \emph{is} the model output. The block is decoded from a $k$-shot average whose fluctuations are Born randomness, and those fluctuations are the only source of stochasticity in the generated shower. There is no classical noise model, no residual sampler and no learned latent distribution anywhere in the pipeline. Consequently $k$ is a model parameter, not a precision knob. Increasing it does not merely reduce error, it changes the distribution being generated by narrowing the sampling noise.
	
	\paragraph{Record sharing and intra-block correlation.} Pixels within one block are decoded from the same circuit invocation, and we control their statistical dependence explicitly. Of the $k$ records used for a block of $b$ pixels, a fraction $\rho_\beta$ is \emph{shared} by all pixels of the block and the remainder is drawn \emph{exclusively} per pixel. The two counts are
	\begin{equation}\label{eq:share}
		k_{\mathrm{sh}}=\lfloor\rho_\beta k\rceil,\qquad k_{\mathrm{ex}}=k-k_{\mathrm{sh}},
	\end{equation}
	so that pixel $j$ is decoded from its own average $\bar{\bm f}^{(j)}_\beta$ formed over the $k_{\mathrm{sh}}$ shared records plus its own $k_{\mathrm{ex}}$ exclusive ones, and the circuit is queried $k_{\mathrm{sh}}+b\,k_{\mathrm{ex}}$ times in total. At $\rho_\beta{=}1$ all pixels of a block see identical records and their sampling noise is perfectly correlated. At $\rho_\beta{=}0$ their noise is independent. The share $\rho_\beta$ is a trainable parameter, one per block, optimized jointly with the circuit. It is the mechanism by which measurement noise is shaped into physical intra-block correlation rather than merely tolerated, and Sec.~\ref{sec:certificates} shows that removing it collapses the model.
	
	\subsection{Noise-aware ridge decoder}\label{sec:ridge}
	
	Records are mapped to pixel values by ridge regression. Because the regressor input is a $k$-shot average rather than an exact expectation, the fit is corrected for the known sampling covariance,
	\begin{equation}\label{eq:ridge}
		\bW_\beta = \Bigl(\bF^\top\bF + \tfrac{1}{k}\bm{\Sigma}_\beta + \alpha\bI\Bigr)^{-1}\bF^\top\bm{Y}_\beta ,
	\end{equation}
	where $\bF$ collects the exact feature expectations and $\bm{\Sigma}_\beta$ is the summed parity covariance of the measured family under the circuit's own output distribution, both computable in closed form from the setting probabilities. The term $\bm{\Sigma}_\beta/k$ is the errors-in-variables correction for regressing on a noisy design matrix. Omitting it biases the decoder toward zero at small $k$.
	
	Two properties are worth noting. The solution is closed-form and therefore a smooth function of the circuit parameters, so the decoder tracks the circuit exactly during optimization and no bilevel problem arises. And $\bW_\beta$ maps $\R^{p_f}\to\R^{b}$, so its rank is at most $p_f$. Blocks larger than $p_f$ contain target directions the decoder cannot represent. At the operating point used here, $b{=}2$ and $p_f{=}18$, the decoder is far from that limit.
	
	\paragraph{Marginal calibration.} We report two variants throughout. The \emph{raw} pipeline is the one described above, with no post-processing of any kind. Every generated intensity is a decoded Born record average. The \emph{calibrated} pipeline additionally maps each generated pixel through the empirical quantile function of the corresponding training pixel, a rank-preserving copula-type construction in which the model supplies the dependence structure and the calibration supplies the marginals. Calibration substantially improves per-pixel metrics and leaves the rank correlations produced by the model untouched. Table~\ref{tab:certificates} lists both variants.
	
	\paragraph{Teacher forcing versus free running.}
	During training the sketch for block $\beta$ is built from ground-truth prefixes. During generation it is built from the model's own outputs. The analysis in Sec.~\ref{sec:noise} concerns measurement noise under the stated conditioning and does not by itself bound model-bias accumulation at large $B$.
	
	% ═══════════════════════════════════════════════════════
	\section{Training}\label{sec:training}
	% ═══════════════════════════════════════════════════════
	
	The trainable quantities are the circuit parameters $\btheta$, the fixed-width encoding $(\bA,\bm{b})$, and the per-block record shares $\{\rho_\beta\}$. All are optimized against a single objective with exact analytic gradients.
	
	\paragraph{Objective.} Because the model output is a $k$-shot average of Born records, its distribution is available in closed form through its characteristic function, and we match distributions directly rather than matching moments. Let $\varphi_i(\bm{w})$ denote the characteristic function of the decoded contribution of one record for sample $i$, evaluated at frequency $\bm{w}$. The $k$ records are independent given the conditioning, so the $k$-shot average has characteristic function $\varphi_i(\bm{w}/k)^k$, and the model and data characteristic functions at frequency $\bm{w}$ are
	\begin{equation}\label{eq:cf}
		\psi_{\mathrm{model}}(\bm{w})=\frac1n\sum_i \varphi_i\!\left(\tfrac{\bm{w}}{k}\right)^{k},
		\qquad
		\psi_{\mathrm{data}}(\bm{w})=\frac1n\sum_i e^{\,\mathrm{i}\,\bm{w}\cdot\by_{i,\beta}} .
	\end{equation}
	The loss is the squared discrepancy averaged over random Fourier frequencies $\bm{w}$ drawn from a multi-bandwidth family,
	\begin{equation}\label{eq:cfmmd}
		\calL_\beta=\frac{1}{n_w}\sum_{\bm{w}}\bigl|\psi_{\mathrm{model}}(\bm{w})-\psi_{\mathrm{data}}(\bm{w})\bigr|^{2},
	\end{equation}
	a characteristic-function maximum mean discrepancy. A fraction of the frequencies additionally carries a phase $e^{\mathrm{i}\bm{w}_c\cdot\by_{<\beta}}$ on the conditioning prefix, so that the objective matches the joint law of block and prefix rather than the block marginal alone. We use $n_w{=}1024$ frequencies and a conditioning fraction of $0.5$.
	
	Two features distinguish Eq.~\eqref{eq:cfmmd} from a sample-based discrepancy. It is exact at finite $k$. No sampling of the model is required to evaluate it, because $\varphi_i$ is computed from the circuit's setting probabilities. And it is differentiable in closed form, so gradients are analytic rather than estimated.
	
	\paragraph{Gradients.} Derivatives with respect to $\btheta$ follow from the parameter-shift rule applied to the setting probabilities, derivatives with respect to $(\bA,\bm{b})$ from the same rule applied to the encoding angles through the chain rule, and derivatives with respect to $\rho_\beta$ in closed form from Eq.~\eqref{eq:share}. The implicit dependence of the closed-form decoder $\bW_\beta$ on the parameters is included exactly rather than being treated as a constant. Every gradient path is verified against central finite differences to a relative error below $10^{-6}$.
	
	This replaces the stochastic-perturbation optimization of earlier versions of this architecture. The cost is that a gradient step is no longer two loss evaluations. Each of the $p$ circuit parameters requires a shifted pair of evaluations, so a step evaluates $\calO(pG n_b)$ circuits rather than $\calO(Gn_b)$. The benefit is that the gradient is exact and the optimizer is not limited by perturbation variance, which at this parameter count is the dominant consideration.
	
	\begin{theorem}[Per-step circuit count is independent of image size]\label{thm:cost}
		With $G$ tensor-product measurement settings, $p$ circuit parameters and minibatch size $n_b$, one exact gradient step for a block evaluates $2pGn_b$ circuits, independent of the image dimension $d$, the dataset size $N$, and the number of blocks $B$.
	\end{theorem}
	\begin{proof}
		Each of the $p$ parameters contributes a shifted pair of evaluations, each requiring $G$ settings for each of the $n_b$ minibatch samples. None of these counts refers to $d$, $N$ or $B$. The image dimension enters only through which block is selected, and the sketch that conditions it has fixed width $m$.
	\end{proof}
	
	The practical consequence is unchanged from the design goal. Enlarging the image adds blocks to \emph{generation} but leaves the cost of an individual training step fixed.
	
	\begin{algorithm}
		\KwData{Data $\bm{Y}$, teacher-forced sketch cache $\calC$, backend}
		\KwIn{Adam step sizes for $\btheta$, $(\bA,\bm{b})$ and $\{\rho_\beta\}$; shots per block $k$; frequency count $n_w$}
		Draw random Fourier frequencies once per block\;
		\For{$t=0$ \KwTo $T-1$}{
			minibatch $\mathcal{I}$ of size $n_b$\;
			\ForEach{block $\beta$}{
				$\bs\leftarrow\calC[\beta,\mathcal{I}]$\;
				$P\leftarrow$ setting probabilities of $U(\ba(\bs),\btheta)$ \tcp*[r]{$Gn_b$ circuits}
				$\bW_\beta\leftarrow$ noise-aware ridge, Eq.~\eqref{eq:ridge}\;
				$\calL_\beta,\nabla\calL_\beta\leftarrow$ CF-MMD and exact gradients, Eq.~\eqref{eq:cfmmd}\;
			}
			Adam update of $\btheta$, $(\bA,\bm{b})$, $\{\rho_\beta\}$ from $\sum_\beta\nabla\calL_\beta$\;
		}
		\Return{$\btheta,\bA,\bm{b},\{\rho_\beta\}$}
		\caption{QFAN training with exact CF-MMD gradients}\label{alg:train}
	\end{algorithm}
	
	% ═══════════════════════════════════════════════════════
	\section{Shot-Noise Propagation}\label{sec:noise}
	% ═══════════════════════════════════════════════════════
	
	Because each block is conditioned on sketches built from previously generated pixels, measurement noise entering at one step is carried forward. We bound this accumulation under a deliberately pessimistic model, in which per-block errors are allowed to add coherently rather than averaging down. The resulting statement, together with its proof, is given in Appendix~\ref{app:noise}. Here we state only the consequence.
	
	Writing $\kappa$ for a bound on the per-block decoder gain, which ridge regularization controls through $\|\bW_\beta\|_F\le\|\bm{Y}_\beta\|_F/(2\sqrt{\alpha})$, and $k$ for the records per block, the expected end-to-end sketch perturbation after $B$ blocks satisfies a bound that grows at most linearly in $B$ and decreases as $k^{-1/2}$. Requiring a fixed end-to-end signal-to-noise ratio then yields a sufficient shot budget scaling as
	\begin{equation}\label{eq:shots}
		k \;=\; \calO\!\left(d^{2}\,p_f\right),
	\end{equation}
	which should be read as a worst-case envelope, not as a prediction of the shots actually required. It is worst-case in two specific ways. It assumes coherent (rather than incoherent) error addition across blocks, and it uses the loosest available bound on the decoder gain. In our experiments $k{=}64$ suffices at $d{=}12$ and $d{=}25$, far below what Eq.~\eqref{eq:shots} would demand, which is the expected behavior of a pessimistic envelope and not evidence in its favor.
	
	The analysis bounds the propagation of measurement noise. Model-bias propagation under free-running generation is a separate mechanism, addressed in Sec.~\ref{sec:next}.
	
	% ═══════════════════════════════════════════════════════
	\section{Execution Paths}\label{sec:dual}
	% ═══════════════════════════════════════════════════════
	
	To separate algorithmic behavior from hardware effects we run identical circuit logic on two backends. A noiseless Aer simulator~\cite{qiskit_aer} and IBM \texttt{ibm\_fez}, a 156-qubit Heron r2 processor~\cite{ibm_fez}. The 3-qubit circuit needs at most one SWAP to realize the wrap-around $\mathrm{CZ}_{20}$ on the heavy-hex topology.
	
	Training is performed on the simulator. The hardware path executes the trained circuit to draw measurement records, with the decoder refitted on device-measured features to absorb the linear part of the channel response. Hardware and simulator therefore differ in device noise and in decoder calibration, not in optimization budget.
	
	\begin{table}[t]
		\centering
		\caption{Execution parameters and measured performance. Training is performed on the simulator. The hardware path executes the trained circuit to draw measurement records, with the decoder refitted on device-measured features. Uncertainties are 95\% bootstrap intervals. The hardware runs use $n{=}200$ generated samples against a marginal-fidelity floor of $0.0069$, the simulator runs $n{=}1200$ against a floor of $0.0032$.}
		\label{tab:paths}
		\small
		\begin{tabular}{@{}lcc@{}}
			\toprule
			& \textbf{Simulator} & \textbf{\texttt{ibm\_fez}} \\
			\midrule
			$k$ (records/block) & 64 & 64 \\
			$G$ (settings) & 3 & 3 \\
			$p_f$ & 18 & 18 \\
			Gradients & exact & none (inference only) \\
			Decoder & train-fit & refit on device features \\
			\midrule
			\multicolumn{3}{@{}l}{\emph{$d{=}12$, $B{=}6$}}\\
			$\bar W_1$ & 0.0092 [0.0084, 0.0108] & 0.0203 [0.0144, 0.0264] \\
			off-diag & 0.0925 [0.0779, 0.1208] & 0.1775 [0.1415, 0.2258] \\
			Spearman off & 0.0929 [0.0773, 0.1135] & 0.1861 [0.1478, 0.2410] \\
			$W_1^{(E)}$ & 0.0440 [0.0381, 0.0530] & 0.0891 [0.0474, 0.1326] \\
			\midrule
			\multicolumn{3}{@{}l}{\emph{$d{=}25$, $B{=}13$}}\\
			$\bar W_1$ & 0.0128 [0.0117, 0.0144] & 0.0185 [0.0125, 0.0249] \\
			off-diag & 0.1282 [0.1113, 0.1467] & 0.2244 [0.1797, 0.2836] \\
			Spearman off & 0.1165 [0.1001, 0.1361] & 0.1937 [0.1547, 0.2454] \\
			$W_1^{(E)}$ & 0.1152 [0.0997, 0.1313] & 0.2924 [0.1756, 0.4171] \\
			\bottomrule
		\end{tabular}
	\end{table}
	
	% ═══════════════════════════════════════════════════════
	\section{Results}\label{sec:results}
	% ═══════════════════════════════════════════════════════
	
	\begingroup
	\captionsetup[subfigure]{labelformat=empty}
	\begin{figure*}[htb]
		\centering
		\subfloat[Pixel 0]{\includegraphics[width=0.24\textwidth]{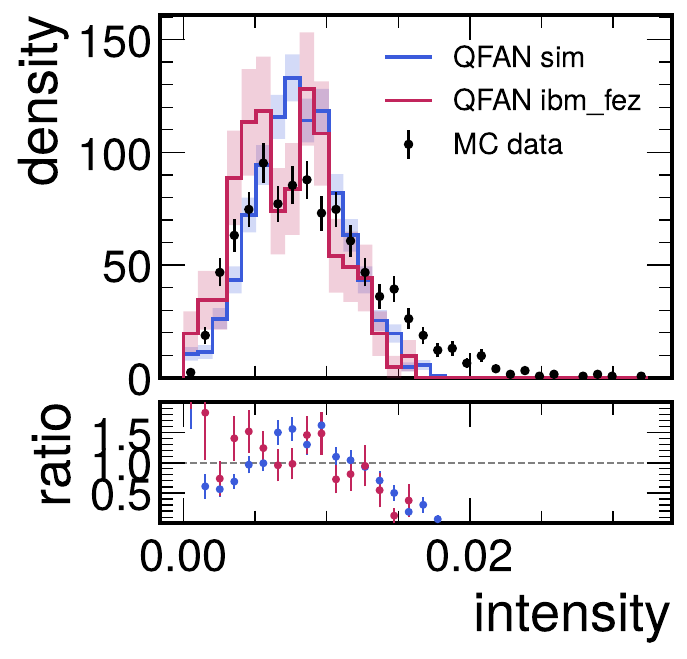}}\hfill
		\subfloat[Pixel 1]{\includegraphics[width=0.24\textwidth]{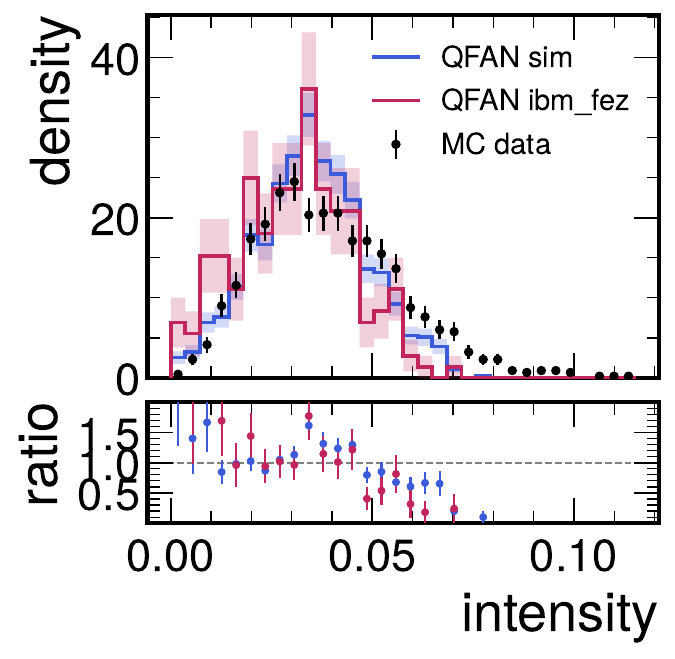}}\hfill
		\subfloat[Pixel 2]{\includegraphics[width=0.24\textwidth]{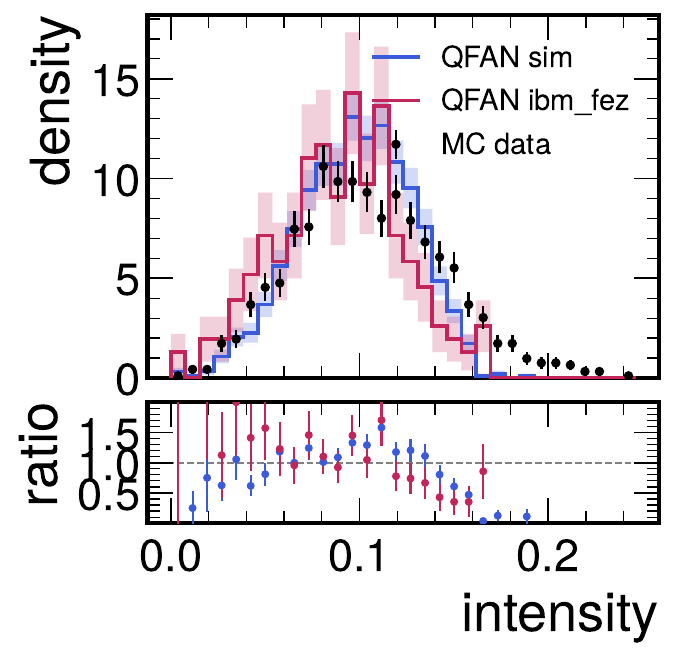}}\hfill
		\subfloat[Pixel 3]{\includegraphics[width=0.24\textwidth]{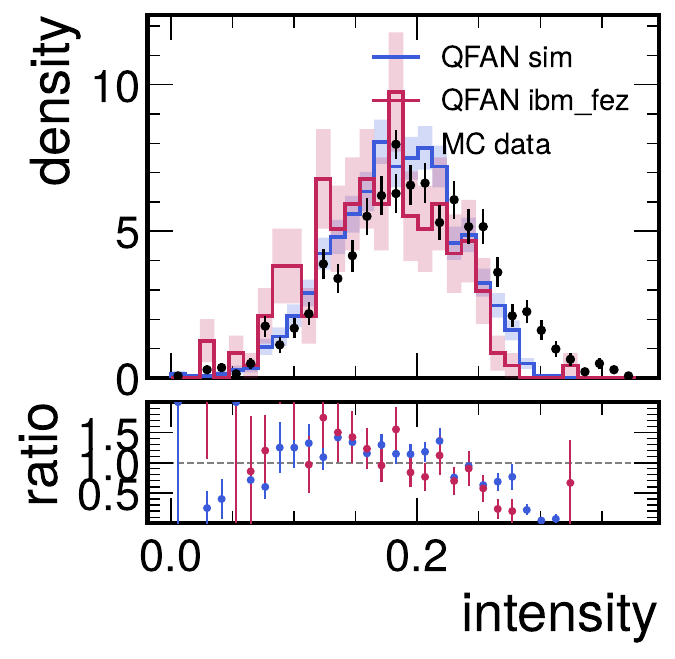}}\hfill
		\subfloat[Pixel 4]{\includegraphics[width=0.24\textwidth]{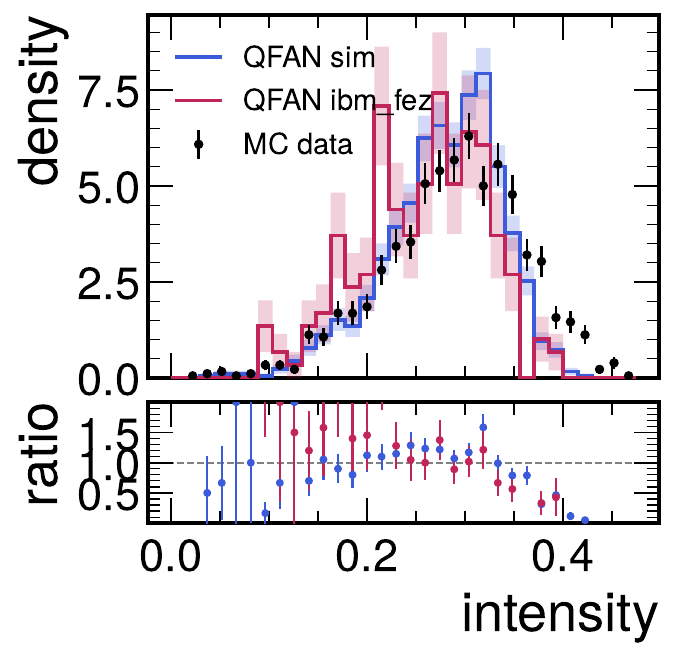}}\hfill
		\subfloat[Pixel 5]{\includegraphics[width=0.24\textwidth]{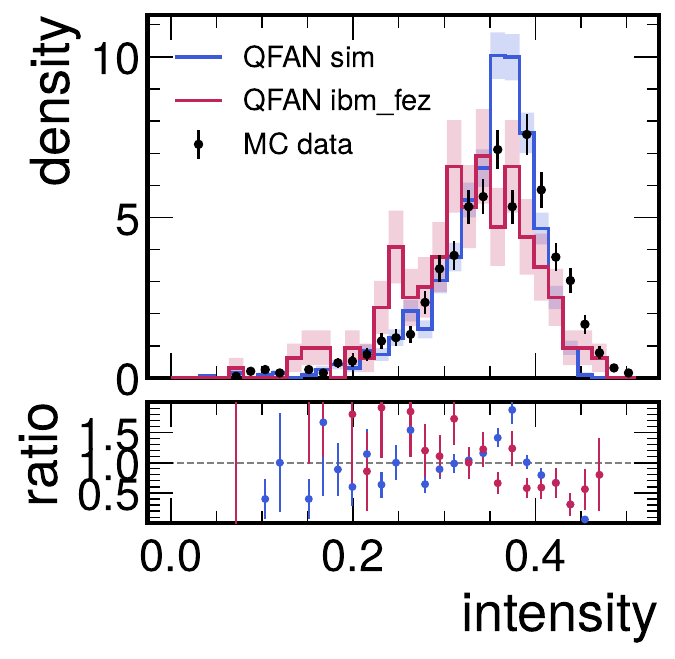}}\hfill
		\subfloat[Pixel 6]{\includegraphics[width=0.24\textwidth]{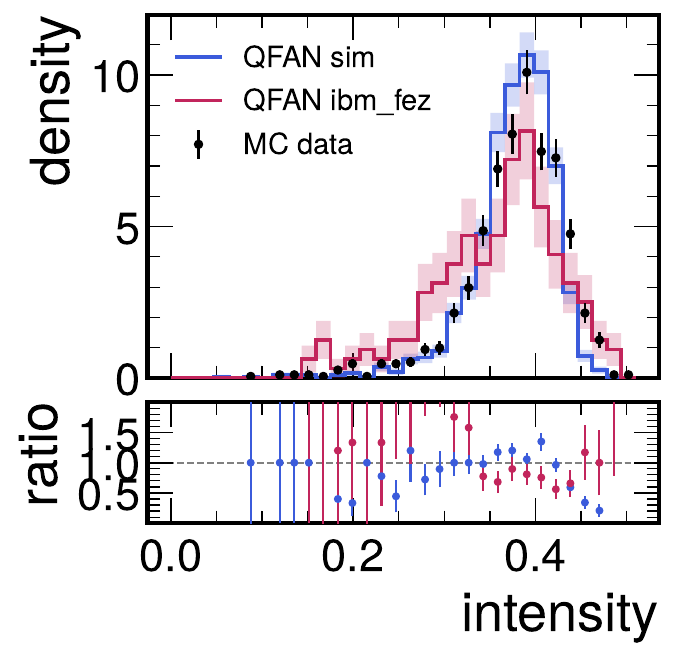}}\hfill
		\subfloat[Pixel 7]{\includegraphics[width=0.24\textwidth]{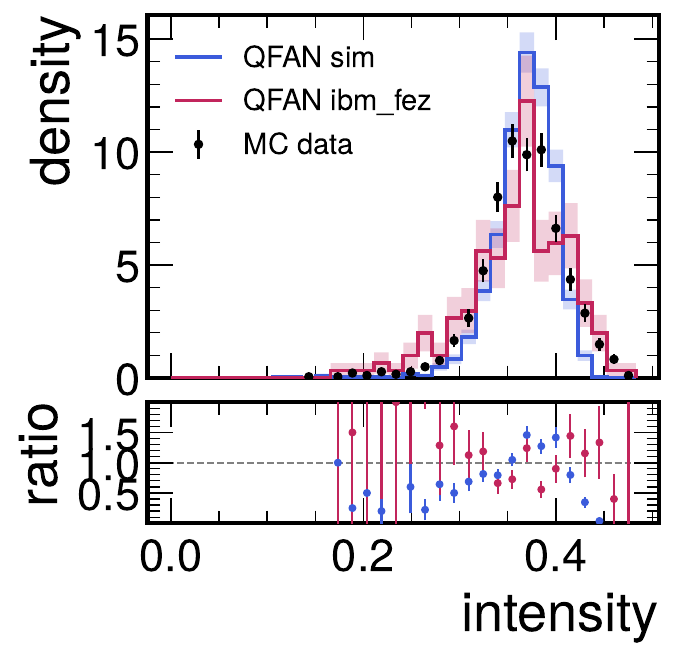}}\hfill
		\subfloat[Pixel 8]{\includegraphics[width=0.24\textwidth]{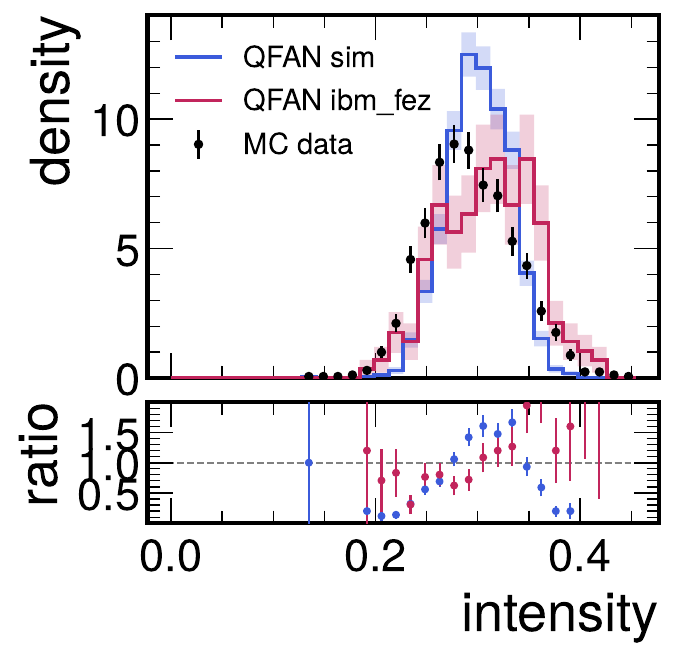}}\hfill
		\subfloat[Pixel 9]{\includegraphics[width=0.24\textwidth]{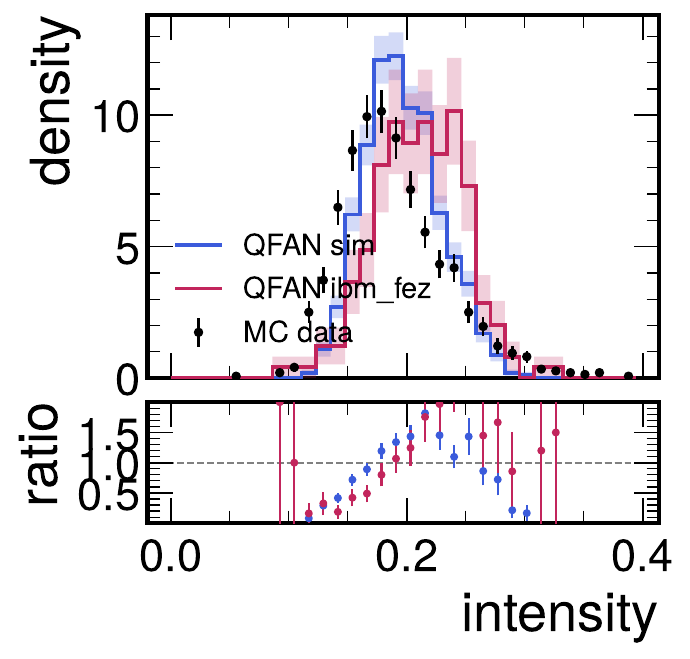}}\hfill
		\subfloat[Pixel 10]{\includegraphics[width=0.24\textwidth]{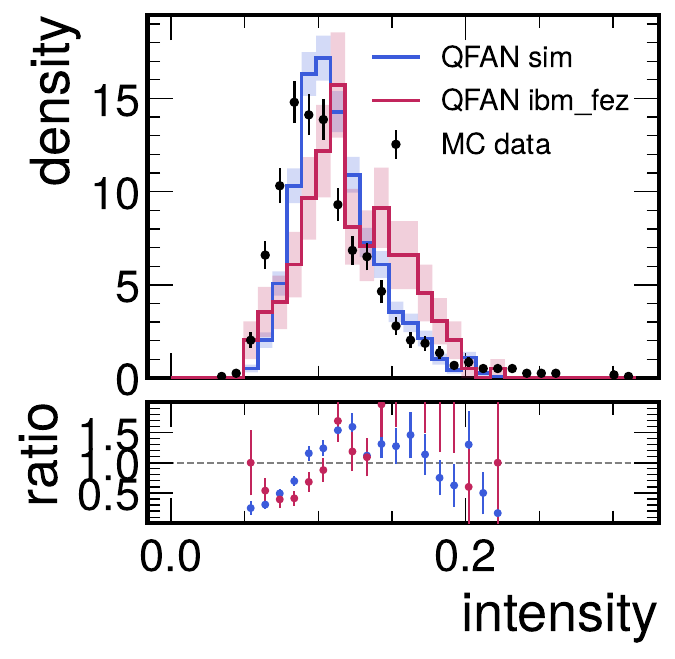}}\hfill
		\subfloat[Pixel 11]{\includegraphics[width=0.24\textwidth]{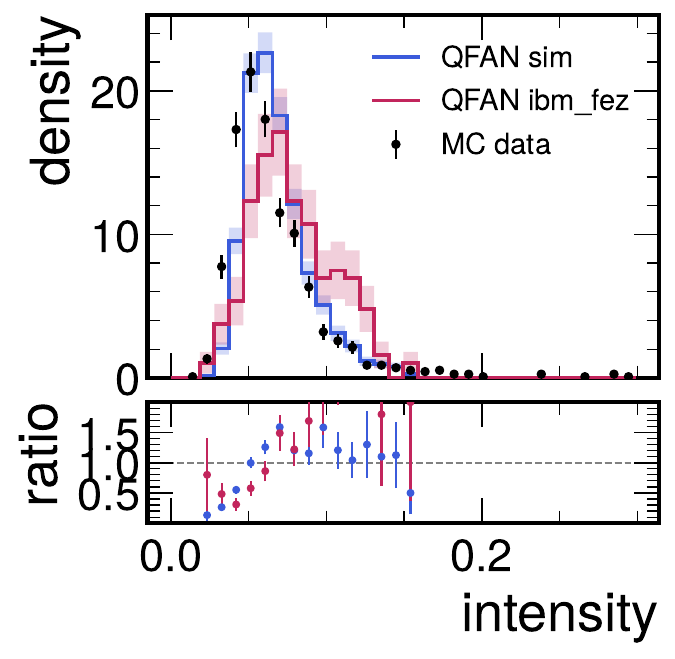}}
		\caption{Per-pixel marginal intensity distributions at $d{=}12$ (variable intensity scale). MC data (black points), QFAN on the simulator (blue) and QFAN on \texttt{ibm\_fez} (red), with model/MC ratio panels below. The peak position shifts systematically from low intensity at pixel~0 to high intensity at pixel~5, reflecting longitudinal shower development. Curves are the calibrated pipeline. See Sec.~\ref{sec:ridge} and Table~\ref{tab:certificates}.}
		\label{fig:marginals}
	\end{figure*}
	\endgroup
	
	The Monte Carlo (MC) data come from the CLIC electromagnetic shower dataset, downsampled to 12 pixels~\cite{clic_ds}. MC data, simulator QFAN and hardware QFAN are compared on a common held-out set of 1000 images with three metrics of increasing difficulty, all lower-is-better. Per-pixel Wasserstein-1 distance (marginal fidelity), inter-pixel Pearson correlation error (joint structure) and the total deposited energy distribution (global consistency). Shared hyperparameters are $n_q{=}3$, $d{=}12$, $b{=}2$, $B{=}6$, $m{=}32$, $k{=}64$ records per block, $G{=}3$ settings, $p_f{=}18$, and $N{=}4800$ training and $1200$ test images. 
	
	\subsection{Per-pixel marginals}
	
	Figure~\ref{fig:marginals} shows the twelve marginals. Each is unimodal and positively skewed, with a peak that moves to higher intensity through block~1 and a similar gradient in block~2, reflecting longitudinal shower development. Both QFAN paths follow the MC shapes across all pixels, and the ratio panels stay near unity over the bulk of each distribution.
	
	Table~\ref{tab:w1_pixels} lists the per-pixel Wasserstein-1 distances. The simulator mean is $\bar{W}_1{=}0.0092$ and the hardware mean $0.0203$, both well above the $n$-dependent statistical floor ($0.0032$ and $0.0069$ respectively), so the hardware degradation is resolved rather than being a finite-sample effect. The per-pixel error varies by an order of magnitude across the image on both backends, from $0.0013$ to $0.0149$ on the simulator and from $0.0020$ to $0.0355$ on hardware, and the two profiles are strongly correlated. The pixels that are hardest for the simulator are the hardest for the device. The error tracks the local intensity scale rather than the block structure, since with $b{=}2$ there are six block boundaries and the largest errors instead concentrate in the high-intensity interior (pixels 3 to 6). The hardware degradation is therefore roughly multiplicative in the per-pixel error rather than being localized at the points where information must cross between circuit invocations.
	
	These marginal numbers include the calibration of Sec.~\ref{sec:ridge}. Table~\ref{tab:certificates} gives the raw counterparts, and the correlation results below are unaffected by it.
	
	\begin{table}[b]
		\centering
		\caption{Per-pixel Wasserstein-1 distances at $d{=}12$, raw (uncalibrated) pipeline. Statistical floors set by the finite evaluation sample are $0.0032$ for the simulator ($n{=}1200$) and $0.0069$ for the hardware ($n{=}200$). All hardware entries except pixel~0 lie above their floor, pixel~0 being consistent with no resolvable deviation. The error tracks the local intensity scale rather than the block structure. With $b{=}2$ there are six block boundaries, and the largest errors on both backends fall in the high-intensity interior (pixels 3 to 6) rather than at boundaries.}
		\label{tab:w1_pixels}
		\small
		\begin{tabular}{@{}ccc@{}}
			\toprule
			\textbf{Pixel} & \textbf{simulator} & \textbf{\texttt{ibm\_fez}} \\
			\midrule
			0 & 0.00130 & 0.00204 \\
			1 & 0.00442 & 0.00934 \\
			2 & 0.00809 & 0.01742 \\
			3 & 0.01490 & 0.03413 \\
			4 & 0.01305 & 0.03553 \\
			5 & 0.01092 & 0.03129 \\
			6 & 0.00670 & 0.02698 \\
			7 & 0.00901 & 0.01150 \\
			8 & 0.01358 & 0.01815 \\
			9 & 0.01293 & 0.02302 \\
			10 & 0.00859 & 0.01766 \\
			11 & 0.00711 & 0.01610 \\
			\midrule
			Mean & 0.00922 & 0.02026 \\
			Median & 0.00880 & 0.01790 \\
			Max & 0.01490 & 0.03553 \\
			\bottomrule
		\end{tabular}
	\end{table}
	
	\subsection{Correlation structure}
	
	\begin{figure*}[tb]
		\centering
		\subfloat[MC data]{\includegraphics[width=0.32\textwidth]{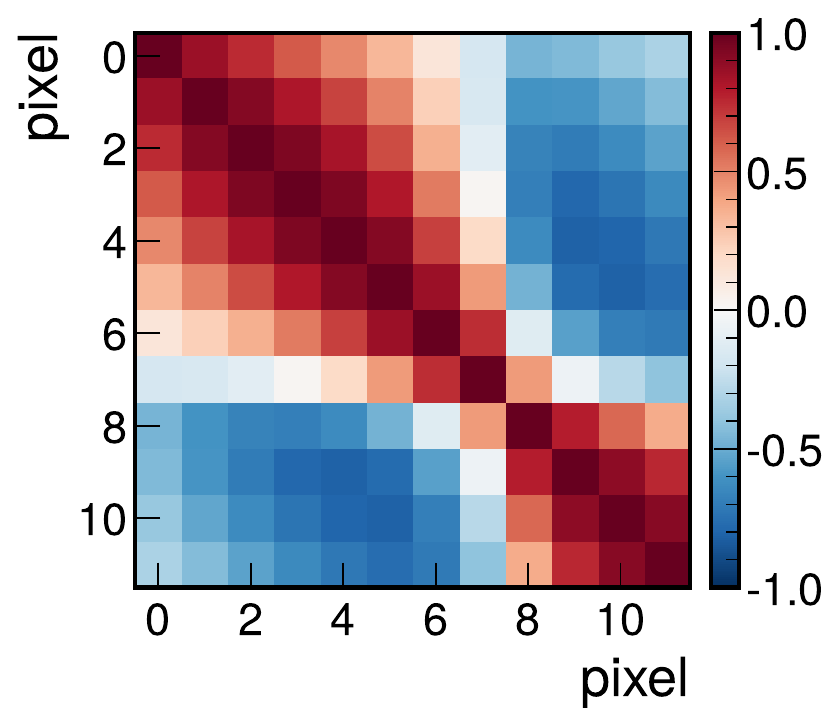}}\hfill
		\subfloat[QFAN (simulator)]{\includegraphics[width=0.32\textwidth]{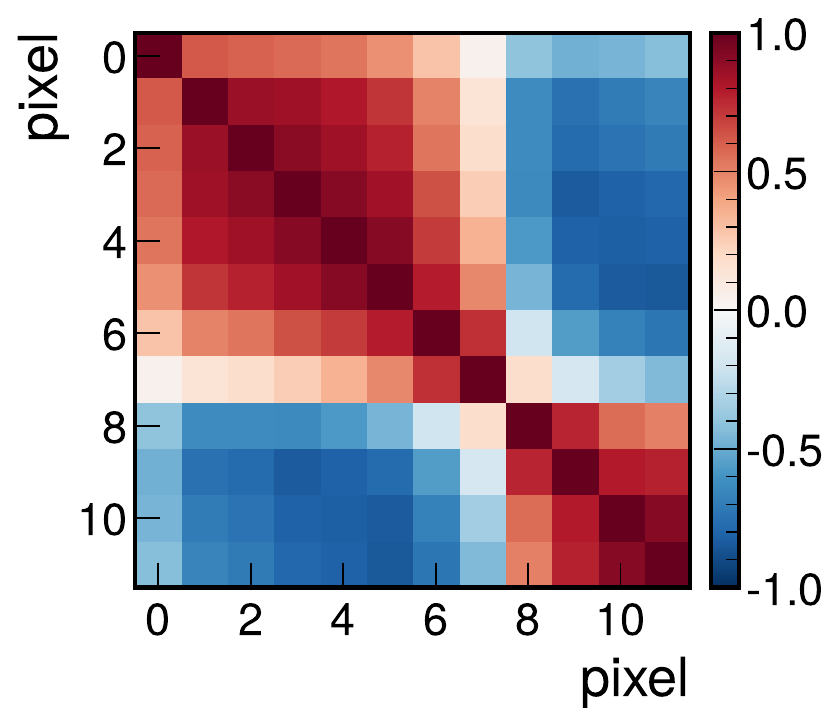}}\hfill
		\subfloat[QFAN (\texttt{ibm\_fez})]{\includegraphics[width=0.32\textwidth]{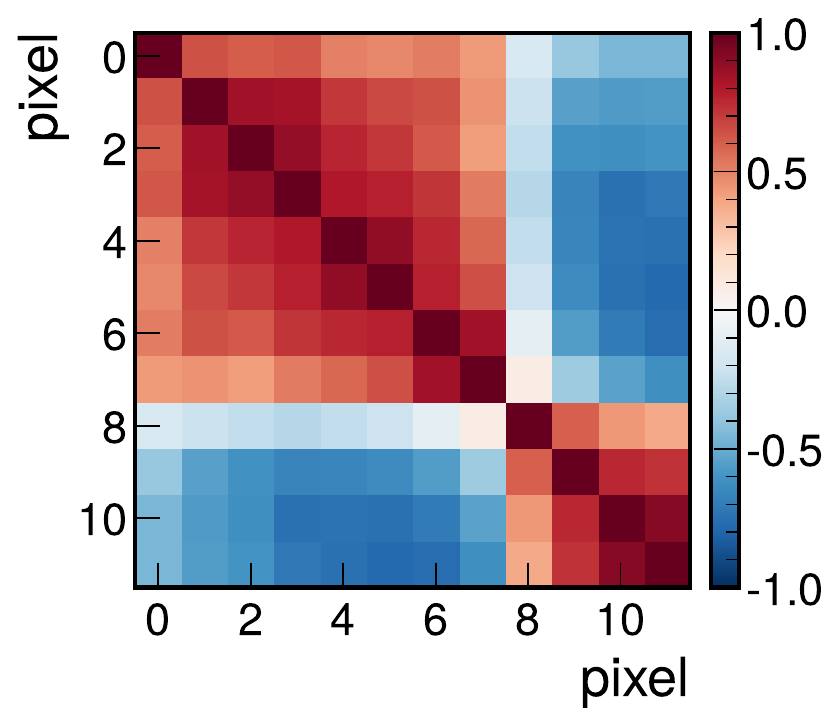}}
		\caption{Pearson correlation matrices at $d{=}12$ representing respectively, the MC data (a), simulator QFAN (b), hardware QFAN (c). Both models reproduce the strong positive intra-block correlations and the negative cross-block correlations that come from the longitudinal development of the shower. The residuals are largest at the block boundaries, where inter-block information passes only through the sketch.}
		\label{fig:corr}
	\end{figure*}
	
	Reproducing marginals is necessary but not sufficient. A model could match every marginal exactly and still generate statistically independent pixels. The correlation matrix is therefore the direct test of whether the circuit and the sketch together capture the joint distribution, and, because the marginal calibration is rank-preserving, it is also the part of the result the calibration cannot manufacture.
	
	Figure~\ref{fig:corr} shows the $12\times12$ Pearson matrices. The MC matrix has two clear signatures. Strong positive correlations within the early detector region and large negative correlations between the two halves of the shower. A shower that starts early deposits more energy in the front and less in the back, and vice versa. Both QFAN paths reproduce this structure, including the sign pattern of the cross-block quadrant. 
	
	That the errors concentrate at the boundary has a direct architectural reading. Within a block, pixel correlations are produced by the entanglement of a single circuit invocation. Across a block boundary, they must be carried by the $m$-dimensional sketch, which is a lossy summary. The boundary is thus the point of maximum information loss in the architecture, and it is where the model is weakest. Importantly, the boundary error does not grow along the chain. At $d{=}12$ this rests on the five block boundaries of a chain of $B{=}6$ blocks. Appendix~\ref{app:25px} extends the test to the twelve boundaries of a 13-block chain, where the error at the last boundary is smaller than at the first, and Sec.~\ref{sec:ds1} extends the chain to $B{=}185$, where the correlation structure is still carried across the full geometry.
	
	\subsection{Total energy}
	
	\begin{figure}[b]
		\centering
		\includegraphics[width=\columnwidth]{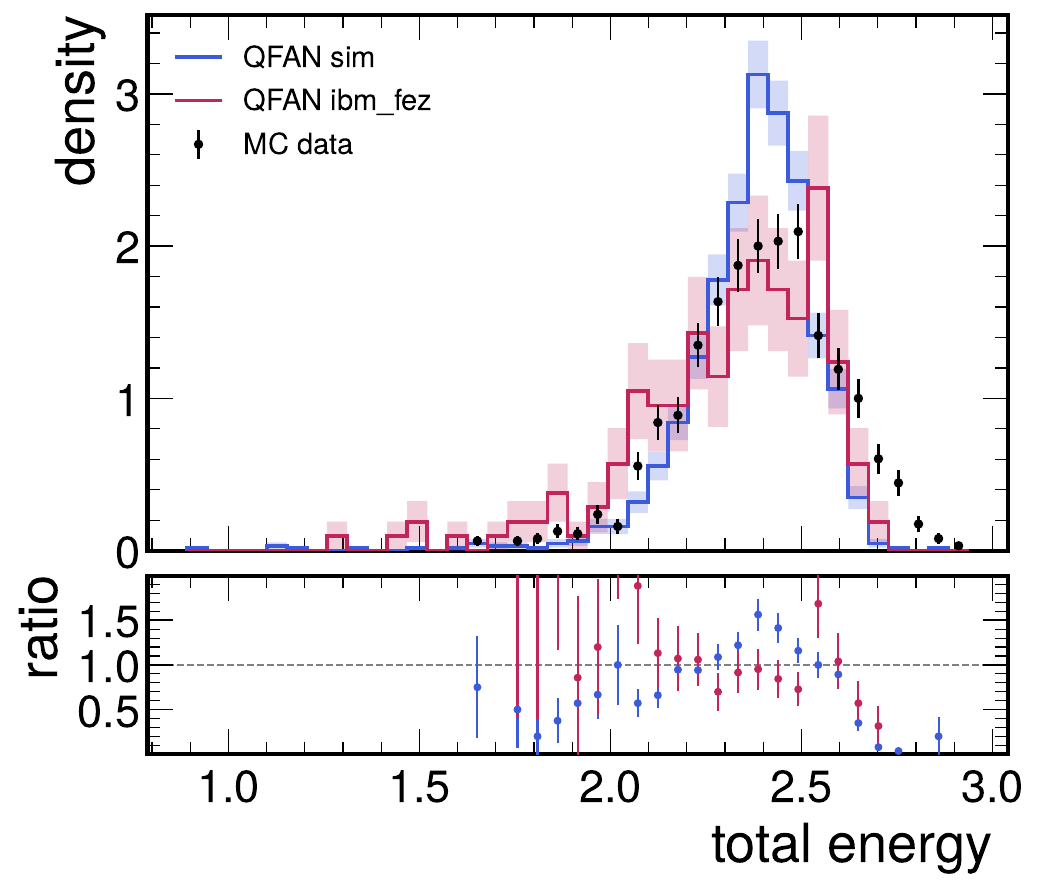}
		\caption{Total energy $E{=}\sum_j y_j$ at $d{=}12$. MC data (black points), simulator QFAN (blue), hardware QFAN (red), with ratio panel. Both paths reproduce the peak position and width, indicating that the sketch carries the energy budget from block~1 into block~2.}
		\label{fig:energy}
	\end{figure}
	
	The total deposited energy $E{=}\sum_j y_j$ aggregates both blocks and is therefore the most sensitive single observable to a failure of inter-block conditioning. If the sketch carried no useful energy information, the two blocks would be sampled almost independently and the variance of $E$ would be the sum of the block variances, giving a visibly broader distribution than MC.
	
	Figure~\ref{fig:energy} shows that this does not happen. Both QFAN paths reproduce the peak position and the width. Quantitatively, the simulator gives $W_1^{(E)}{=}0.0440$ and the hardware $W_1^{(E)}{=}0.0891$. These are larger than typical per-pixel values because a small coherent per-pixel bias adds up across twelve pixels in the sum. They test a stricter global consistency condition rather than being on the same numerical scale as the single-pixel distances.
	
	We note one interpretive caution that applies to this observable in particular. The peak position of $E$ is largely fixed by the per-pixel means, which the classical decoder reproduces well on its own. The informative feature is the \emph{width} of the distribution, since $\Var(E)=\sum_{jk}\mathrm{Cov}(y_j,y_k)$ depends on precisely the inter-pixel covariances that the sketch and the circuit have to supply. A model with correct marginals and no inter-pixel correlation would show the correct peak and a wrong width. We therefore recommend interpreting Fig.~\ref{fig:energy} through its width and ratio panel, and we make this test explicit in Sec.~\ref{sec:certificates}.
	
	% ═══════════════════════════════════════════════════════
	\section{Full Resolution on CaloChallenge Dataset~1}\label{sec:ds1}
	% ═══════════════════════════════════════════════════════
	
	At twelve pixels the chain is only six blocks long, which says little about the regime QFAN was designed for. Realistic detector geometries need chains of hundreds of blocks, and the open question from Sec.~\ref{sec:intro} is whether generation still works when every block is conditioned on the model's own earlier output. To find out, we trained the same Born-record sampler on the photon showers of CaloChallenge Dataset~1~\cite{Krause_2025,ds1_zenodo}, an established benchmark for fast calorimeter simulation with 368 voxels in the ATLAS calorimeter geometry~\cite{atlasfast3}. With blocks of two values, generating one shower takes a chain of 185 blocks, all produced by the same four-qubit circuit. The decoders are still the closed-form ridge regressions of Secs.~\ref{sec:architecture} and~\ref{sec:training}, and nothing is fitted to generated output. We needed two additions to get there, a preprocessing that makes every generated shower consistent with its own total energy and a separate number of measurement records for each block.
	
	\subsection{Data and preprocessing}\label{sec:ds1_data}
	
	Dataset~1 consists of \GEANT{} simulations of single photons entering the calorimeter at $0.2<|\eta|<0.25$, voxelized in cylindrical coordinates over five layers (Table~\ref{tab:ds1_geometry}). We work at a single incident energy of 65~GeV, which spares us from conditioning on the energy. Of the 10{,}000 showers available at this energy we use a random subset of 6{,}000, the same data budget as for the 12-pixel images, split into 4{,}800 for training and 1{,}200 for testing. Within each layer the voxels are ordered ring by ring, so the ten angular bins of a ring sit next to each other in the chain.
	
	\begin{table}[b]
		\caption{\label{tab:ds1_geometry}Photon geometry of Dataset~1 and the share of the mean deposited energy per layer at 65~GeV.}
		\begin{ruledtabular}
			\begin{tabular}{lccc}
				layer & angular $\times$ radial bins & voxels & energy share \\
				\midrule
				0  & $1\times8$   & 8   & 0.7\% \\
				1  & $10\times16$ & 160 & 17.8\% \\
				2  & $10\times19$ & 190 & 80.8\% \\
				3  & $1\times5$   & 5   & 0.4\% \\
				12 & $1\times5$   & 5   & 0.2\% \\
			\end{tabular}
		\end{ruledtabular}
	\end{table}
	
	We divide each voxel by the incident energy and multiply all values by one global constant. The important step comes next. Each shower $\bm{e}$ is rewritten as its total energy followed by its energy fractions,
	\begin{equation}\label{eq:fractions}
		\bm{x}=\bigl(t',\;c\,\bm{e}/E\bigr),\qquad t'=t_0+w\,\frac{E-\tilde{E}}{\sigma_E},\qquad E=\textstyle\sum_v e_v ,
	\end{equation}
	Here $\tilde{E}$ is the median of the training totals and $\sigma_E$ their median absolute deviation (MAD) scaled by $1.4826$, the constants $t_0{=}0.3$ and $w{=}0.1$ place the standardized total $t'$ in the working range of the decoder, and $c\,\bm{e}/E$ is the vector of energy fractions multiplied by a constant $c$ that brings them into the intensity range we used for the 12-pixel images. After generation we normalize the fractions to unit sum and multiply them by the generated total. This inverts the map exactly and makes the voxel energies of every generated shower sum to its generated total. Without it the chain would have to hold the total energy fixed by itself. At fixed incident energy the deposited total varies by only about $1.3\%$ from shower to shower, and a chain whose blocks each draw fresh measurement noise has no way of keeping to a constraint that tight. The prepared shower has 369 values. They are generated in $B{=}185$ blocks of $b{=}2$ by the circuit of Fig.~\ref{fig:circuit} with $n_q{=}4$ qubits and $L{=}3$ layers, which has 24 trainable angles and $p_f{=}30$ features, and a 32-dimensional sketch.
	
	\subsection{Records per block}\label{sec:ds1_records}
	
	Each block is decoded from an average over $k$ measurement records, so its Born noise has a variance proportional to $1/k$. For the 12-pixel images a single value of $k$ worked for every block, but a long chain is less forgiving. The block that holds the total energy comes first and has no history to condition on, so all of its spread has to come from Born noise. A block whose fluctuation is already fixed by the history, on the other hand, needs very little noise of its own. After training we therefore choose a separate number of records $k_\beta\in[2,1024]$ for each block, namely the value that minimizes the training loss of Sec.~\ref{sec:training} on training showers with the circuit held fixed, and refit that block's ridge decoder for $k_\beta$. During generation the block then averages $k_\beta$ records. The randomness is still entirely Born randomness, and the model gains one integer per block. The fitted budgets range from 2 to 1017 records, with a median of 82 and only 9 for the first block, so a single shower costs between $7.2\times10^4$ and $1.4\times10^5$ measurement shots. For the untrained control we fit budgets in exactly the same way.
	
	\subsection{Results}\label{sec:ds1_results}
	
	We generate 1{,}200 showers and compare them with the 1{,}200 test showers. In addition to the metrics of Sec.~\ref{sec:results} we look at the per-voxel Wasserstein distance in units of the voxel's standard deviation, reported separately for dense voxels and for sparse ones (zero in more than $5\%$ of showers or with a coefficient of variation above one), the fraction of zeros in the sparse voxels, and the MAD of the total energy relative to \GEANT. To put these numbers in context we evaluate two references on the same test showers. The \emph{floor} compares an equal number of training showers with the test showers and tells us how large the errors are from finite statistics alone. The \emph{no-correlation} model takes the test showers, shuffles every value independently in the prepared space and applies the same inverse. It therefore has the \GEANT{} marginals, and the voxel energies of each shower still sum to its total, but it carries no correlations at all. Between these two references we define the share of the achievable correlation structure that a model captures,
	\begin{equation}\label{eq:captured}
		\kappa=\frac{\epsilon_{\rm null}-\epsilon_{\rm model}}{\epsilon_{\rm null}-\epsilon_{\rm floor}},
	\end{equation}
	where $\epsilon$ is the mean absolute error of the off-diagonal Spearman correlation matrix.
	
	\begin{table*}[t]
		\caption{\label{tab:ds1_results}Dataset~1 at 368 voxels against \GEANT{} test showers. W1 is the per-voxel Wasserstein distance over the voxel's standard deviation, for dense and sparse voxels, and zeros is the zero fraction of sparse voxels. Pearson and Spearman give the mean absolute error of the off-diagonal correlation matrix, $\kappa$ the captured correlation structure of Eq.~\eqref{eq:captured}, and $E$ spread the MAD of the total energy relative to \GEANT. One training seed, with errors from eight bootstrap resamples of test and generated showers.}
		\begin{ruledtabular}
			\small\setlength{\tabcolsep}{3.5pt}
			\begin{tabular}{lccccccc}
				& W1 dense & W1 sparse & zeros & Pearson & Spearman & $\kappa$ & $E$ spread \\
				\midrule
				QFAN & $0.276\pm0.003$ & $0.381\pm0.010$ & $3.8\pm0.2\%$ & $0.118\pm0.003$ & $0.105\pm0.003$ & $48\pm3\%$ & $0.62\pm0.03$ \\
				untrained circuit & $0.632\pm0.002$ & $0.586\pm0.010$ & $0.2\%$ & $0.139\pm0.002$ & $0.176\pm0.003$ & $-10\pm1\%$ & $0.07\pm0.00$ \\
				no correlations & $0.058\pm0.003$ & $0.054\pm0.002$ & $21.8\pm0.8\%$ & $0.143\pm0.002$ & $0.165\pm0.003$ & $0\%$ & $0.97\pm0.04$ \\
				floor & $0.068\pm0.009$ & $0.070\pm0.004$ & $22.2\pm0.9\%$ & $0.045\pm0.003$ & $0.040\pm0.002$ & $100\%$ & $0.97\pm0.07$ \\
			\end{tabular}
		\end{ruledtabular}
	\end{table*}
	
	\begin{figure*}[tb]
		\centering
		\subfloat[MC data]{\includegraphics[width=0.32\textwidth]{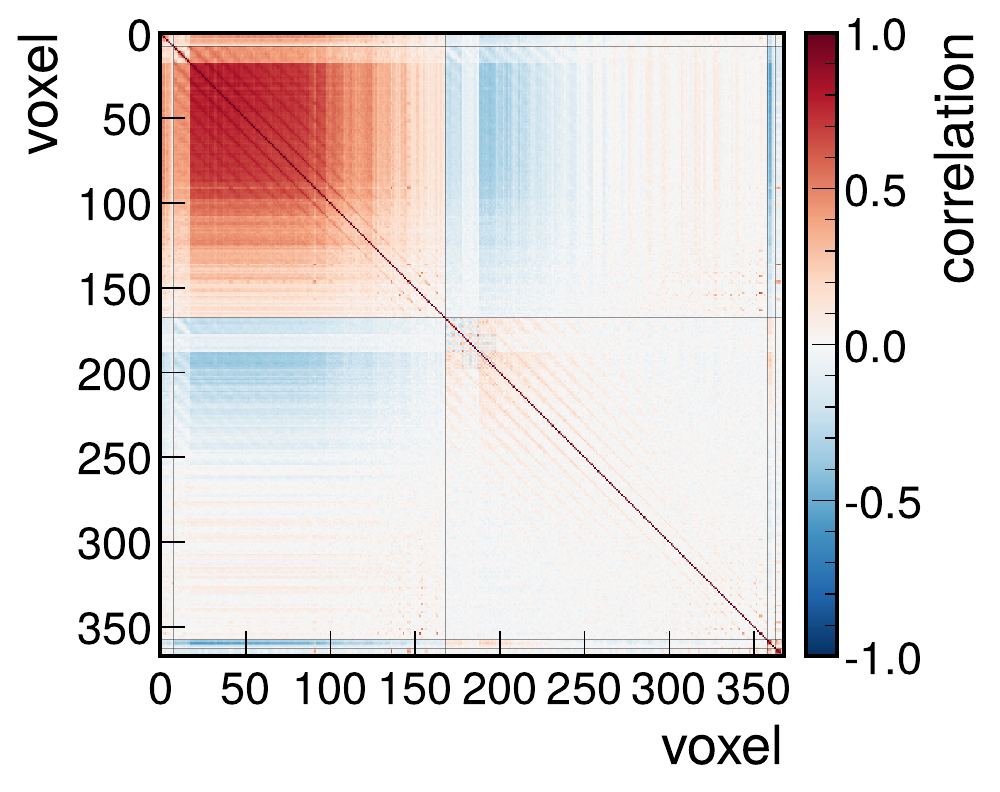}}\hfill
		\subfloat[QFAN]{\includegraphics[width=0.32\textwidth]{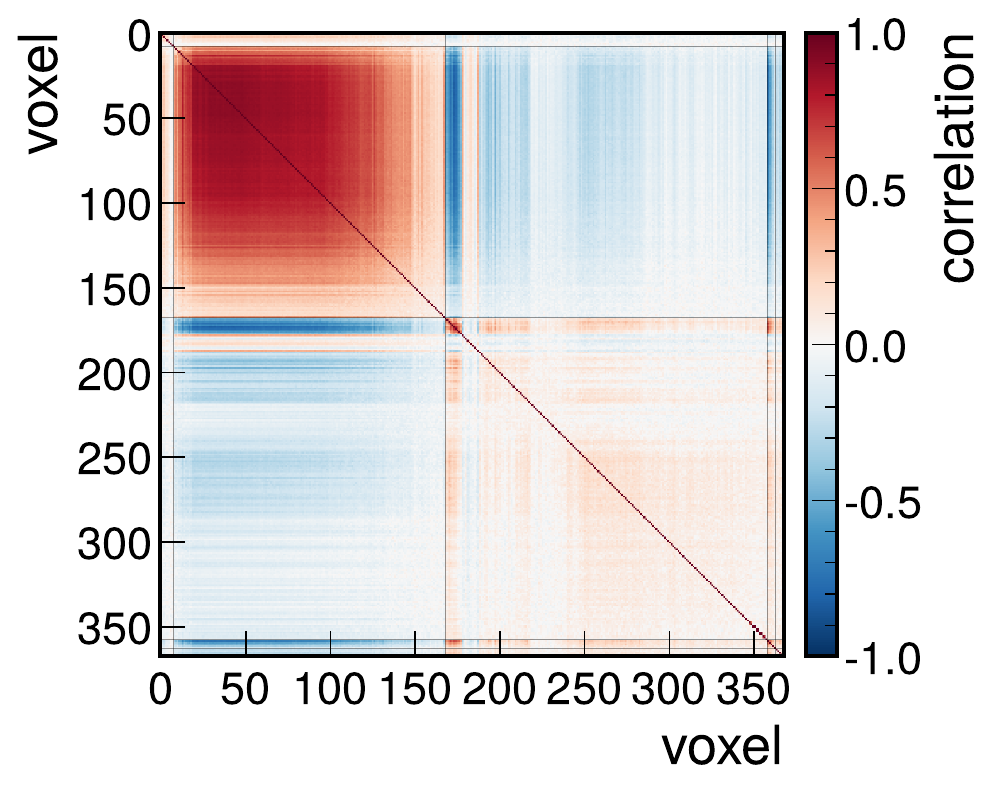}}\hfill
		\subfloat[QFAN, untrained circuit]{\includegraphics[width=0.32\textwidth]{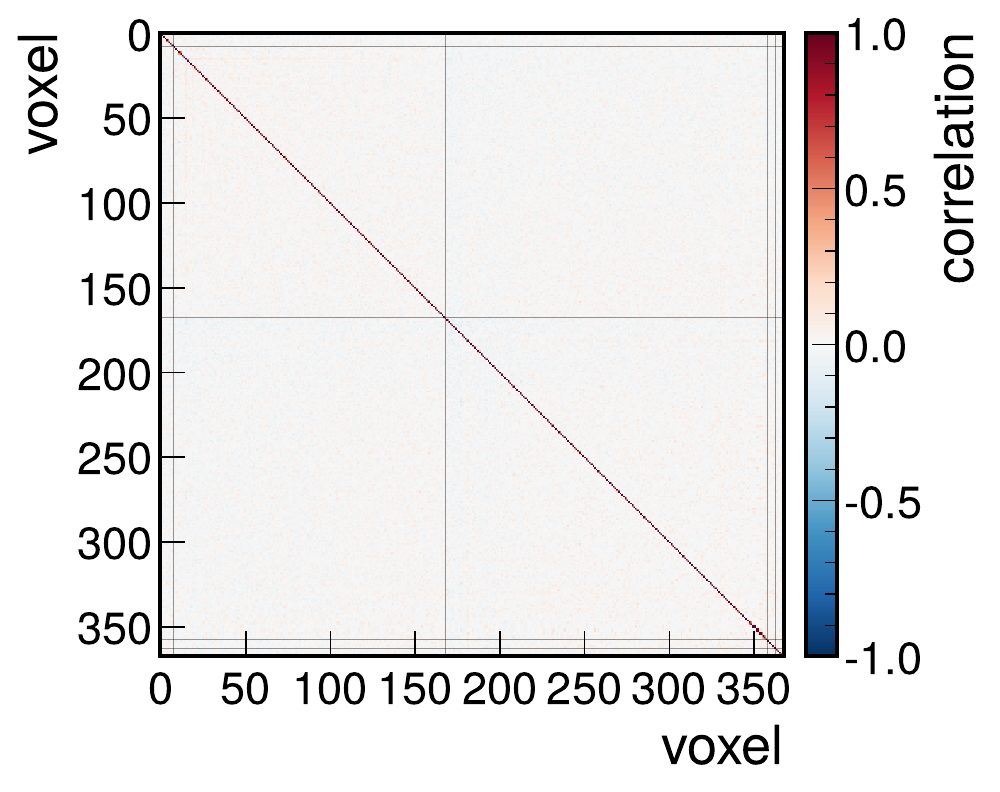}}
		\caption{\label{fig:ds1_corr}Voxel-voxel correlation matrices of the 368-voxel showers, ordered ring by ring. Lines mark the calorimeter layers.}
	\end{figure*}
	
	The main numbers are collected in Table~\ref{tab:ds1_results}. Generating freely along 185 blocks, QFAN captures $\kappa=48\pm3\%$ of the achievable correlation structure. The untrained circuit, with every classical stage fitted in the same way, captures none of it. The matrices themselves are shown in Fig.~\ref{fig:ds1_corr}. QFAN reproduces the strongly correlated block of layer~1, the anticorrelation between layers~1 and~2 across the radial rings, and the weaker structure inside layer~2, while the untrained circuit shows none of these features.
	
	\begin{figure*}[tb]
		\centering
		\subfloat[Total energy]{\includegraphics[width=0.32\textwidth]{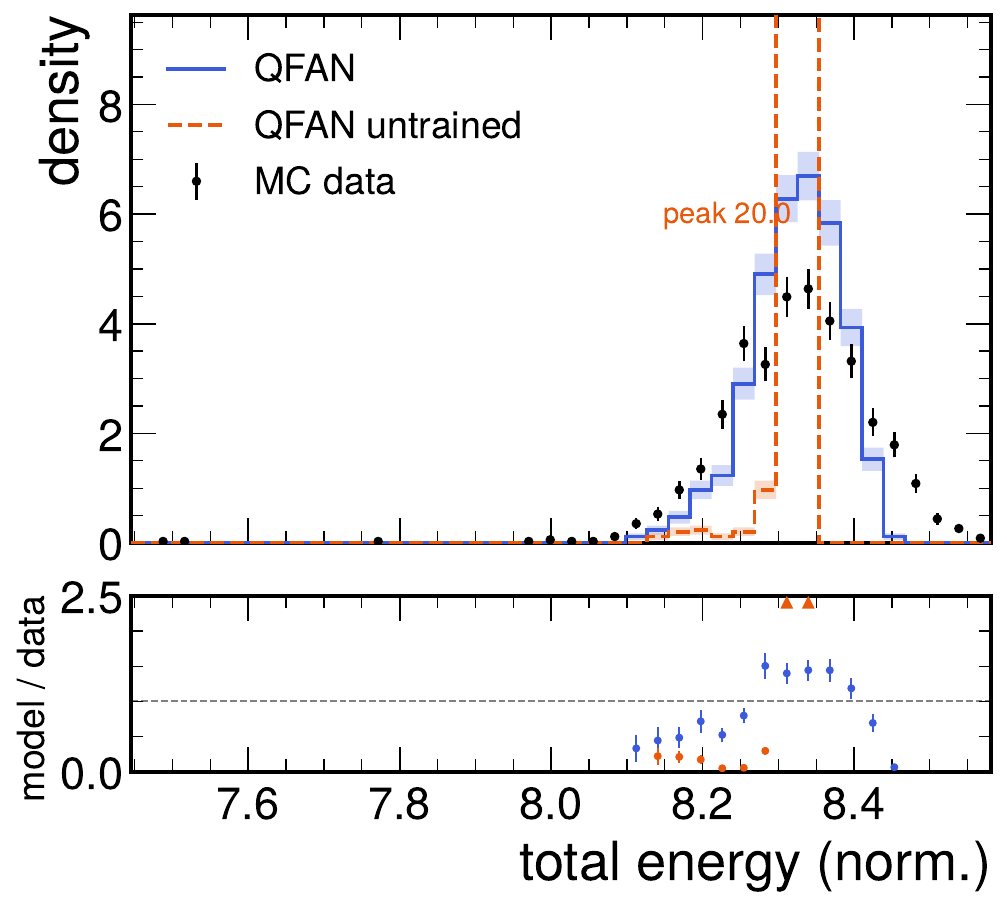}}\hfill
		\subfloat[Layer~0 energy]{\includegraphics[width=0.32\textwidth]{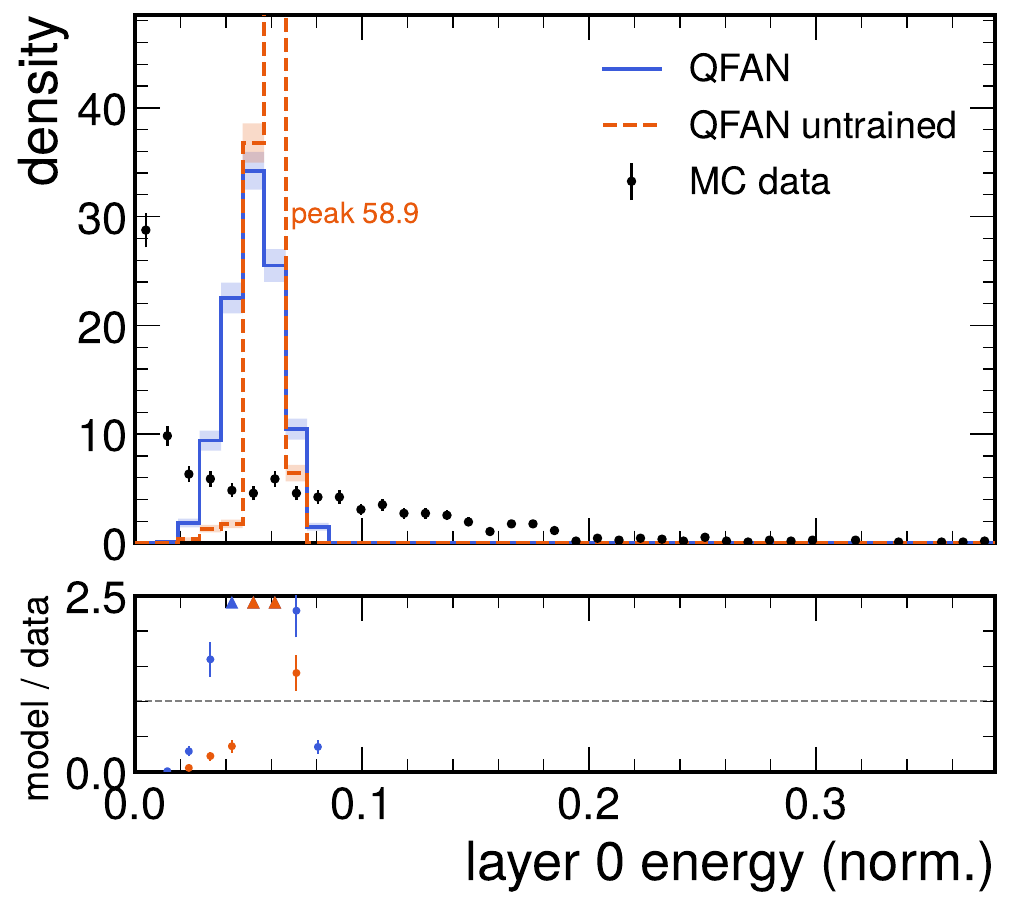}}\hfill
		\subfloat[Layer~1 energy]{\includegraphics[width=0.32\textwidth]{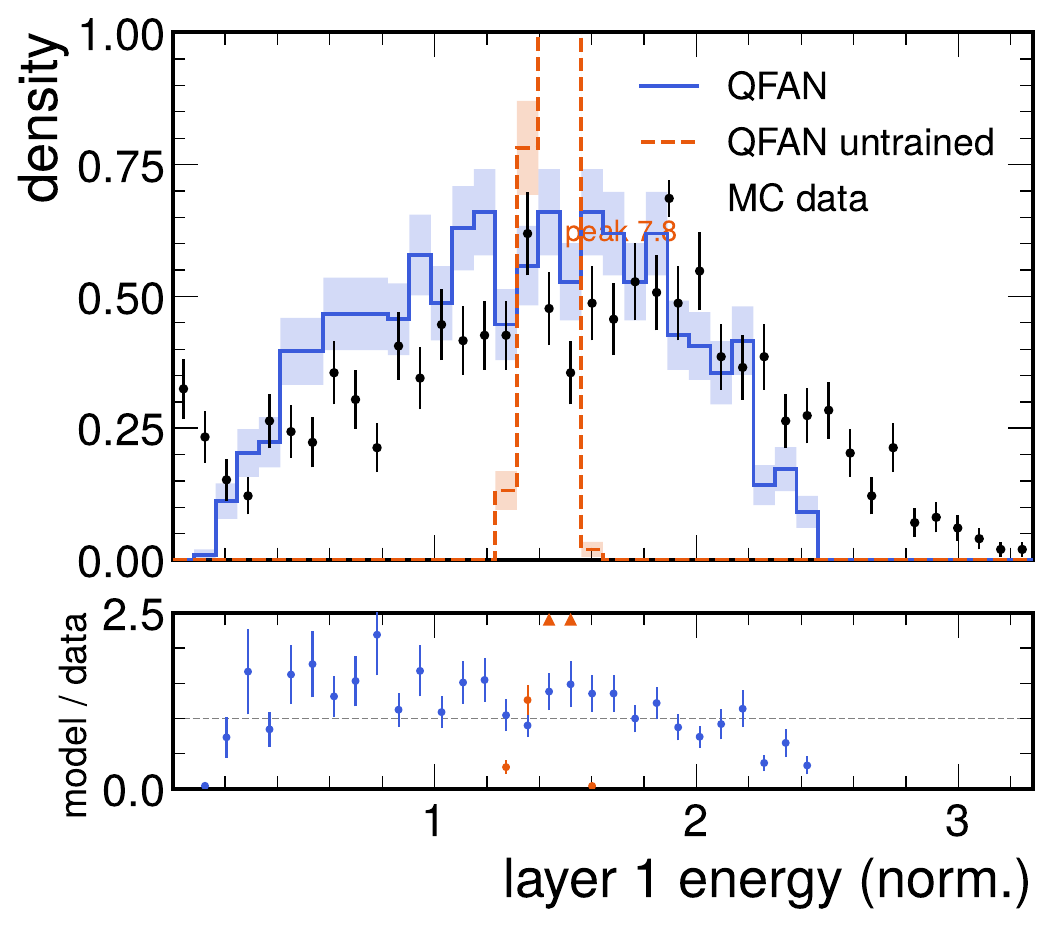}}\\
		\subfloat[Layer~2 energy]{\includegraphics[width=0.32\textwidth]{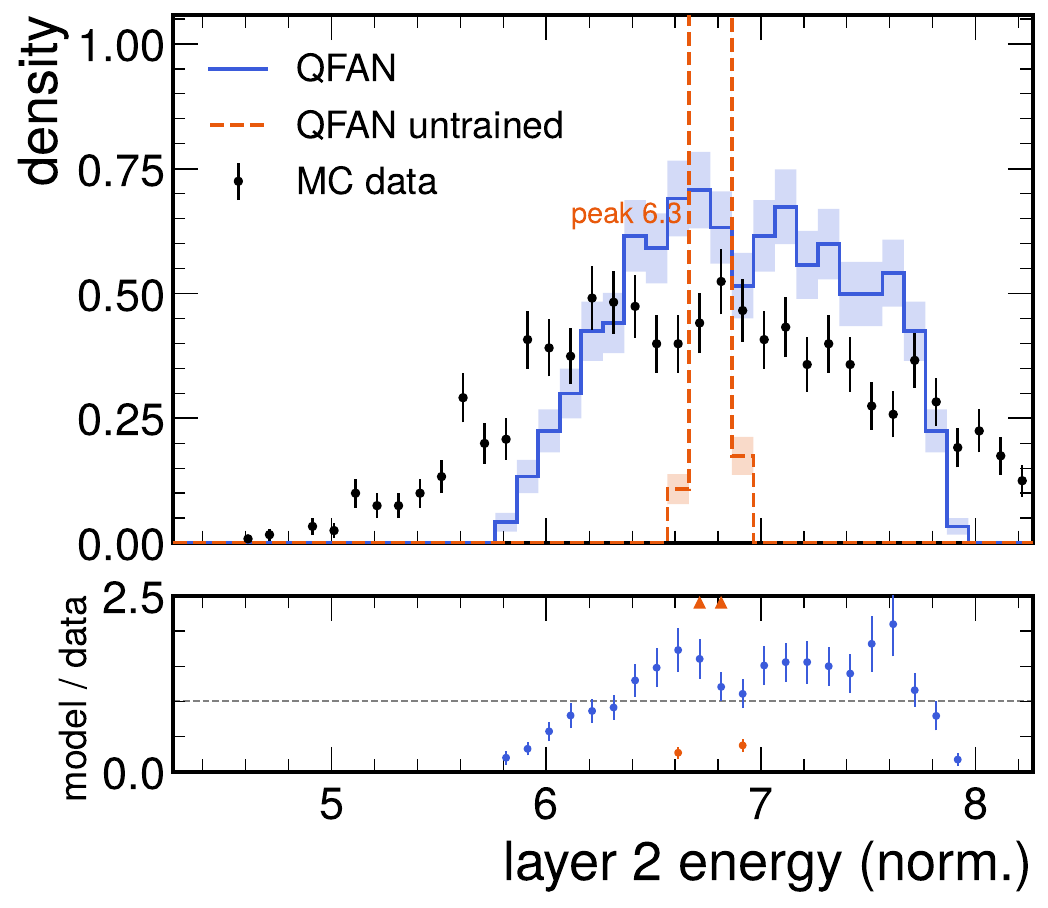}}\hfill
		\subfloat[Layer~3 energy]{\includegraphics[width=0.32\textwidth]{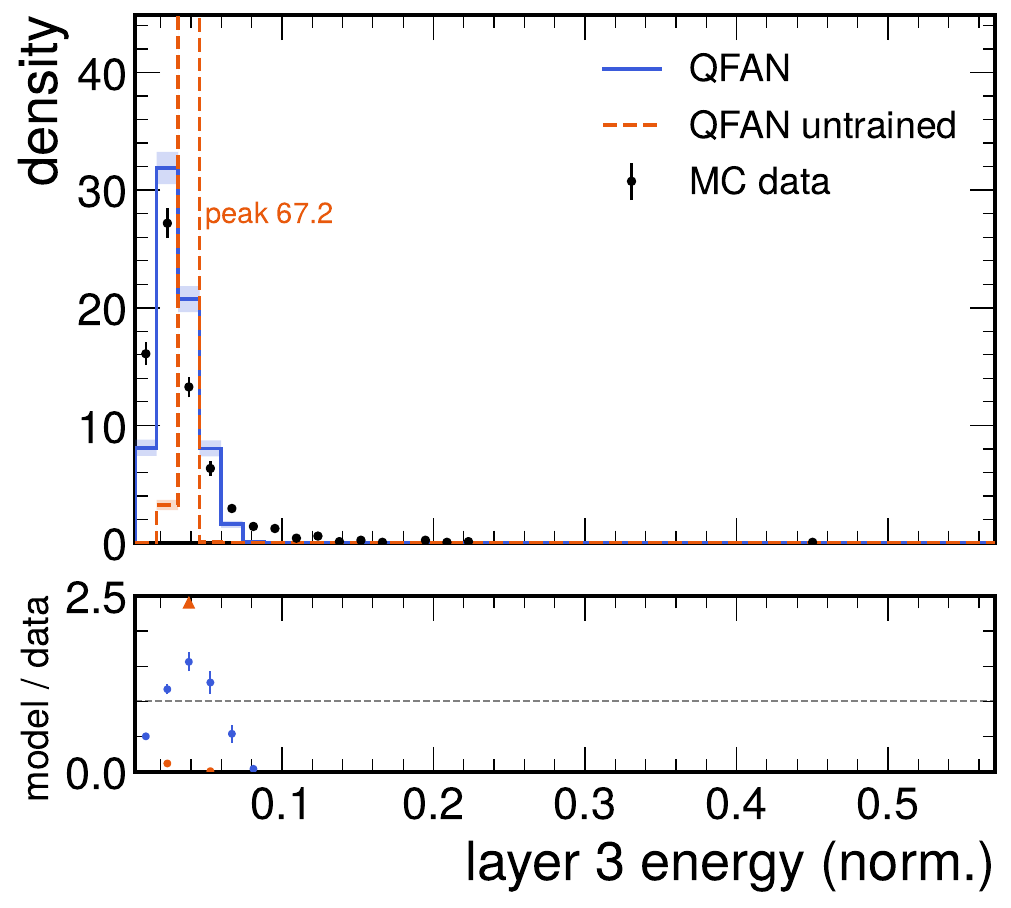}}\hfill
		\subfloat[Layer~12 energy]{\includegraphics[width=0.32\textwidth]{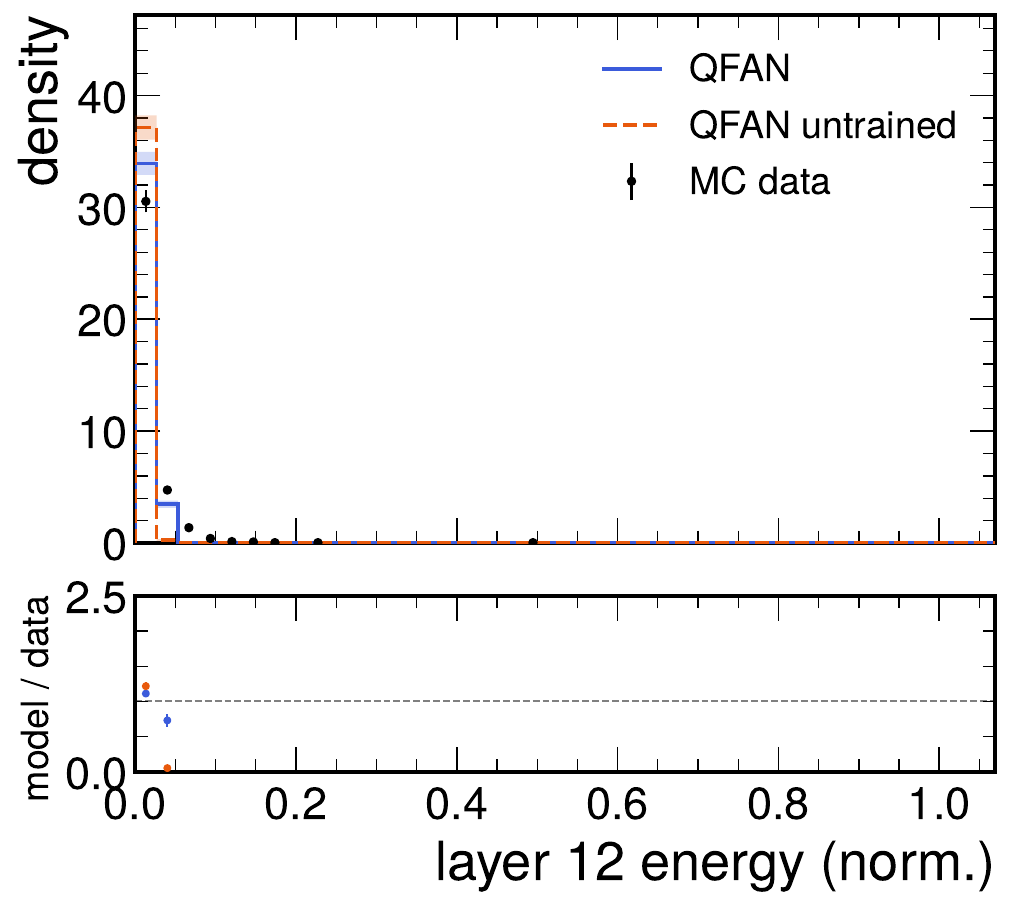}}
		\caption{\label{fig:ds1_energy}Total and layer energies of the 368-voxel showers for MC data (points), QFAN (blue) and the untrained circuit with every classical stage fitted identically (orange, dashed), with the ratio of each model to MC data underneath. The vertical scale follows MC data and QFAN. Where the untrained circuit's narrow peak exceeds it, its height is given, and ratios beyond the panel are marked by triangles. Error bars are statistical.}
	\end{figure*}
	
	Correlations are not the whole story, since a generator also has to produce the right shower-to-shower fluctuations (Fig.~\ref{fig:ds1_energy}). Layers~1 and~2 hold $98.6\%$ of the deposited energy between them, and how that energy is shared between the two depends strongly on how deep the shower develops. QFAN reproduces $75\%$ and $66\%$ of the \GEANT{} standard deviation of these two layer energies, at Wasserstein distances of $0.30$ and $0.34$ standard deviations. The untrained circuit gets $6\%$ of either, at $0.78$ and $0.79$. For the total energy, QFAN reaches $62\%$ of the \GEANT{} MAD and the untrained circuit $7\%$. The picture is the same voxel by voxel (Fig.~\ref{fig:ds1_voxels}). For every one of the 368 voxels the trained circuit is closer to \GEANT{} than the untrained one, with a median distance of $0.28$ standard deviations compared with $0.64$.
	
	The average shower deserves a comment of its own (Fig.~\ref{fig:ds1_polar}). QFAN reproduces it to within about $25\%$ per voxel, with median ratios to \GEANT{} between $0.77$ in layer~0 and $0.95$ in layers~2 and~12, and individual voxels between $0.69$ and $1.36$. Somewhat counterintuitively, the untrained circuit does better here, with median ratios between $0.94$ and $1.04$. The reason is that each decoder restores its voxel's training mean given the true history. Along a free-running chain the history is the model's own, and a model that responds to its history drifts with it, whereas the untrained circuit barely responds at all. The average shower is therefore a measure of drift along the chain, and the fluctuations discussed above are the better measure of what the circuit has learned.
	
	\begin{figure*}[tb]
		\centering
		\subfloat[Layer~0, ring~1]{\includegraphics[width=0.24\textwidth]{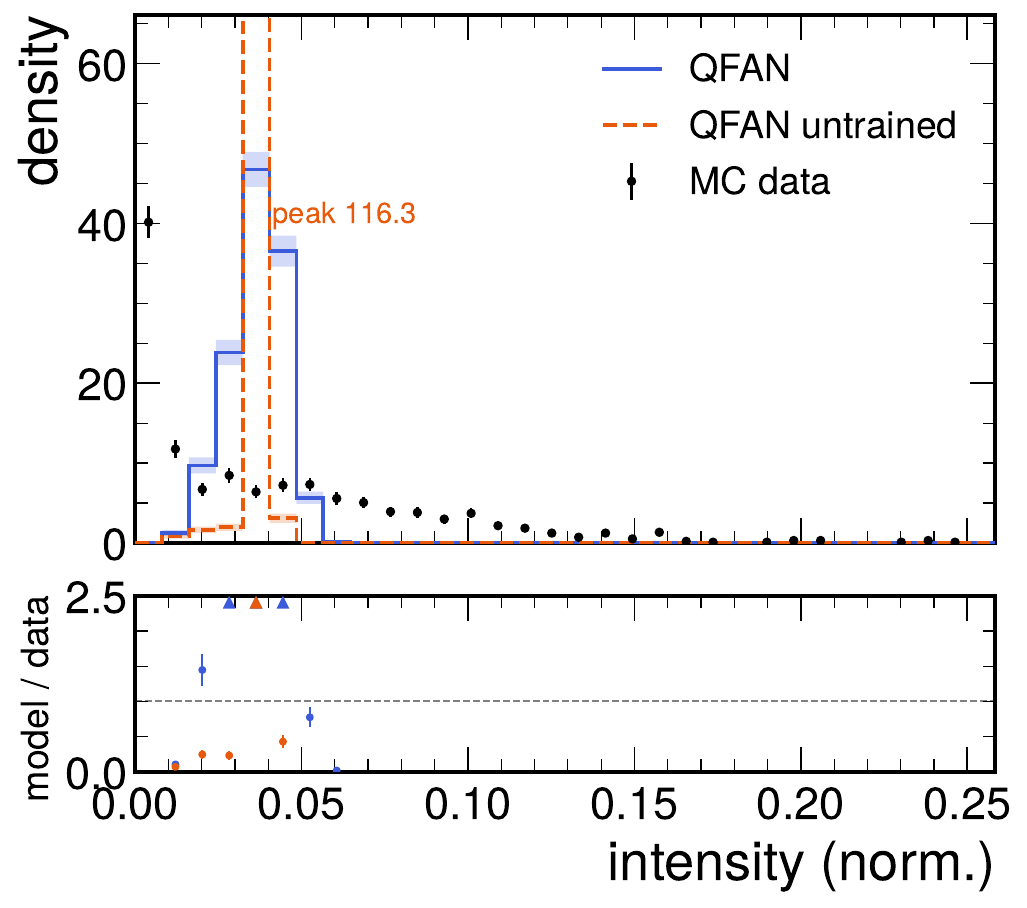}}\hfill
		\subfloat[Layer~1, ring~3]{\includegraphics[width=0.24\textwidth]{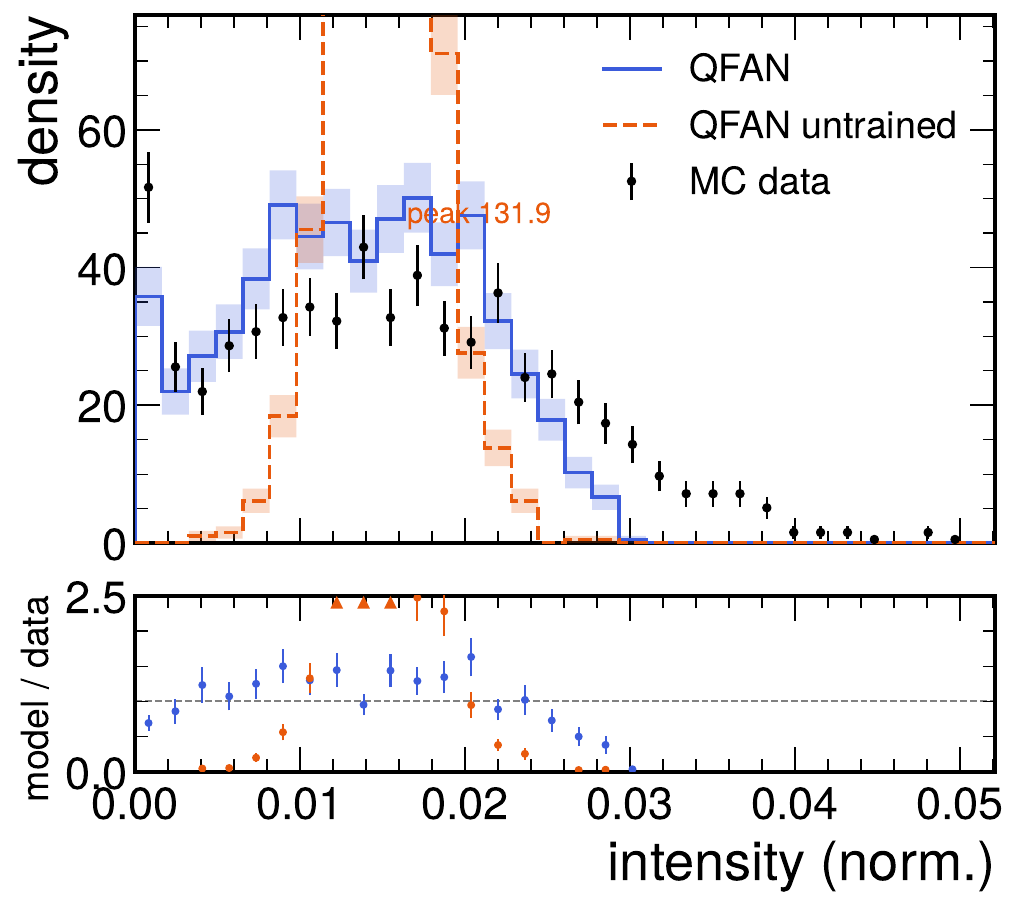}}\hfill
		\subfloat[Layer~1, ring~6]{\includegraphics[width=0.24\textwidth]{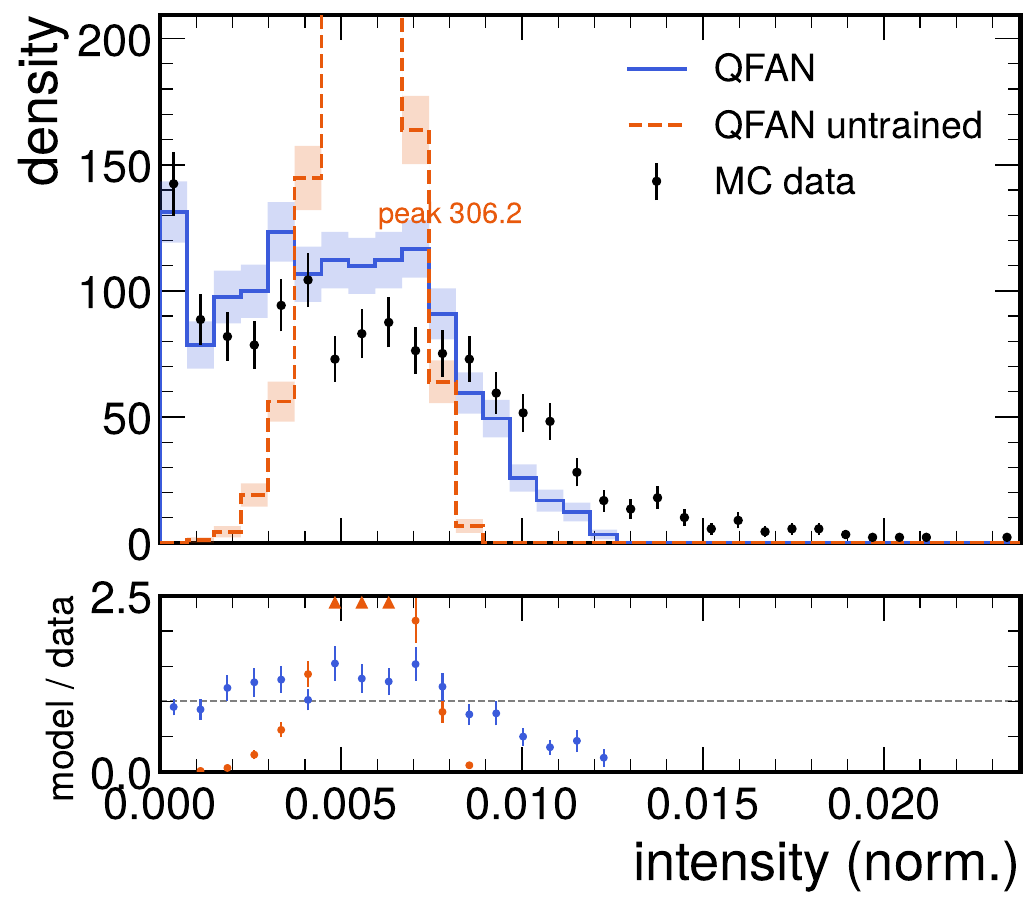}}\hfill
		\subfloat[Layer~1, ring~10]{\includegraphics[width=0.24\textwidth]{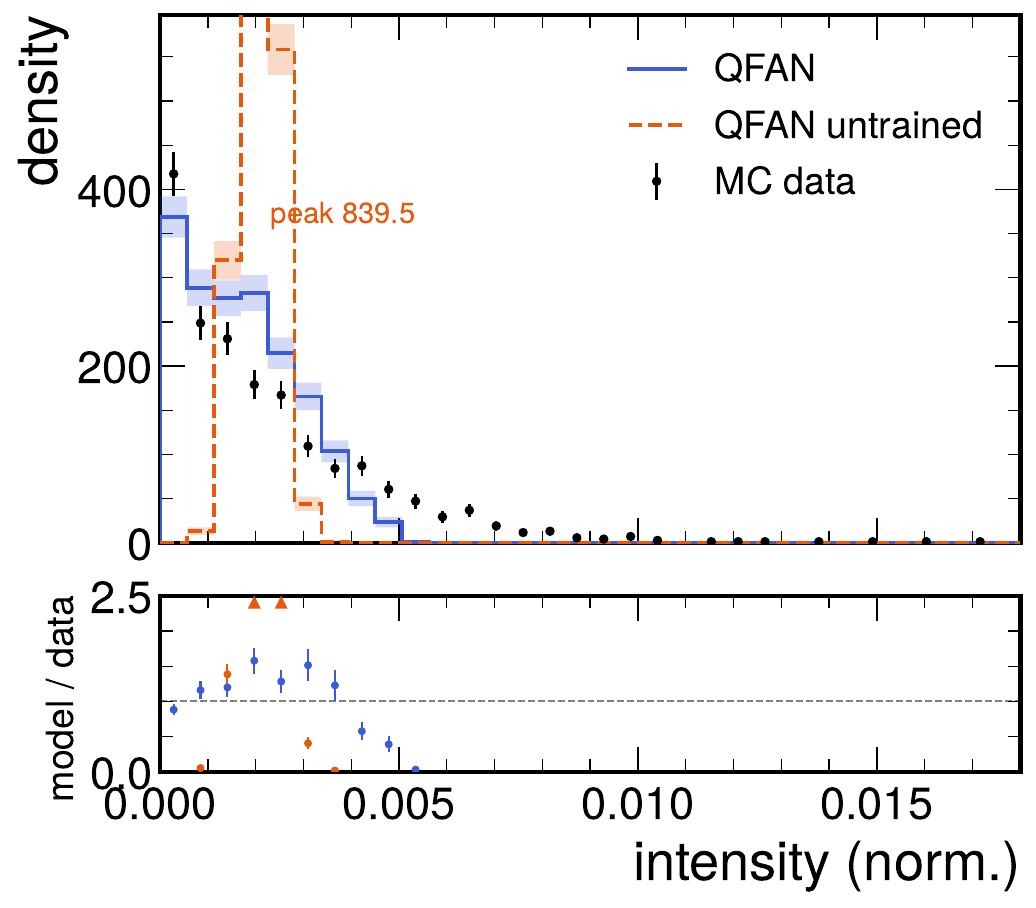}}\\
		\subfloat[Layer~1, ring~13]{\includegraphics[width=0.24\textwidth]{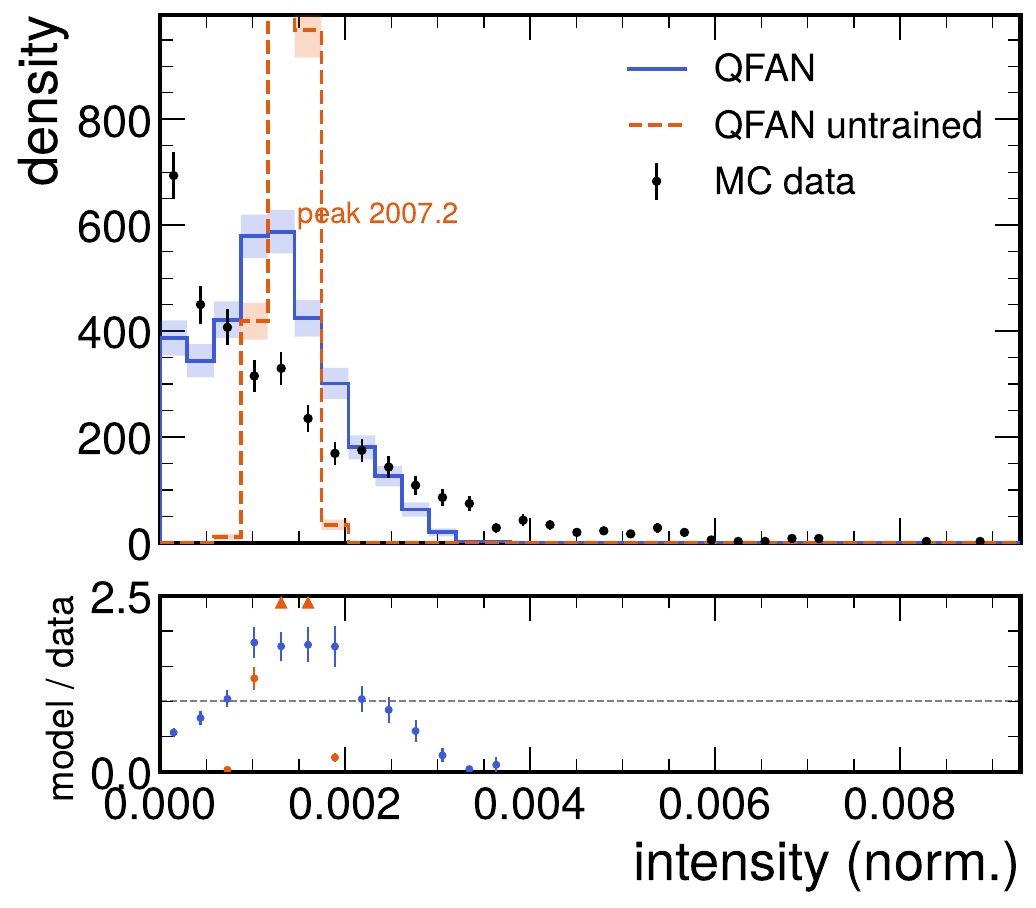}}\hfill
		\subfloat[Layer~1, ring~16]{\includegraphics[width=0.24\textwidth]{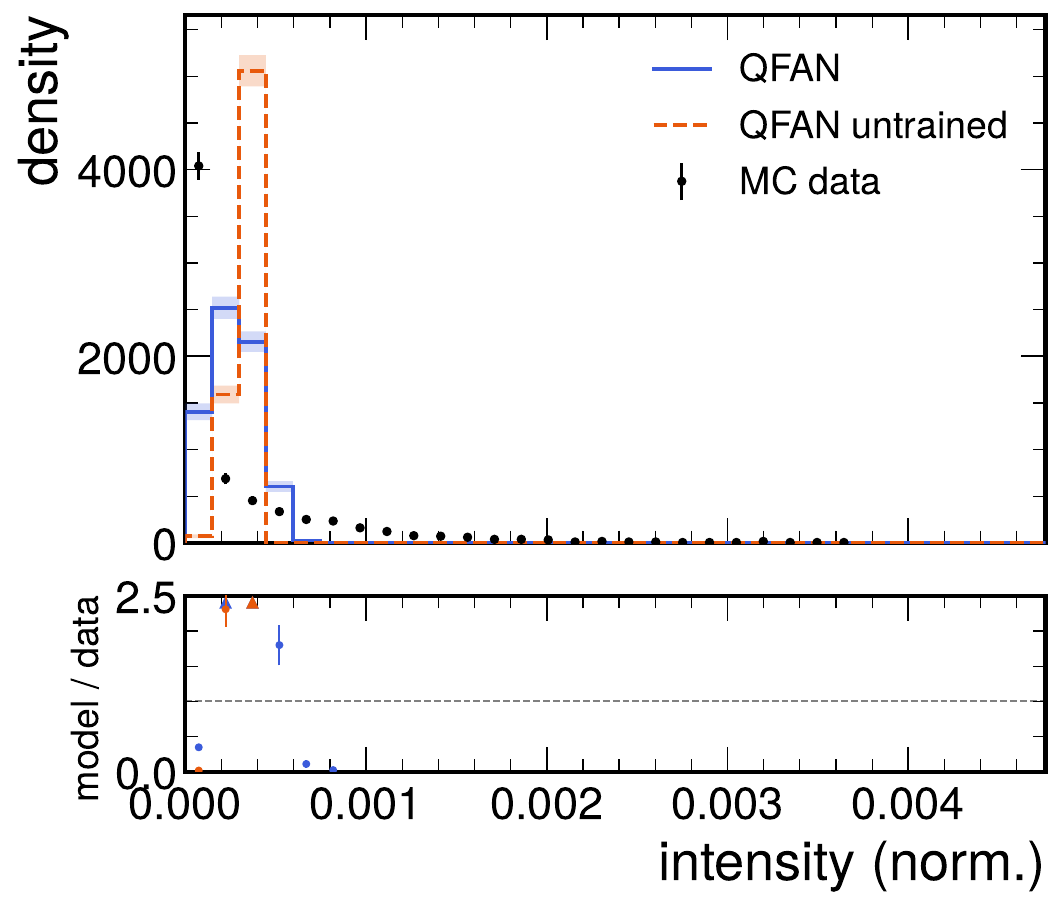}}\hfill
		\subfloat[Layer~2, ring~4]{\includegraphics[width=0.24\textwidth]{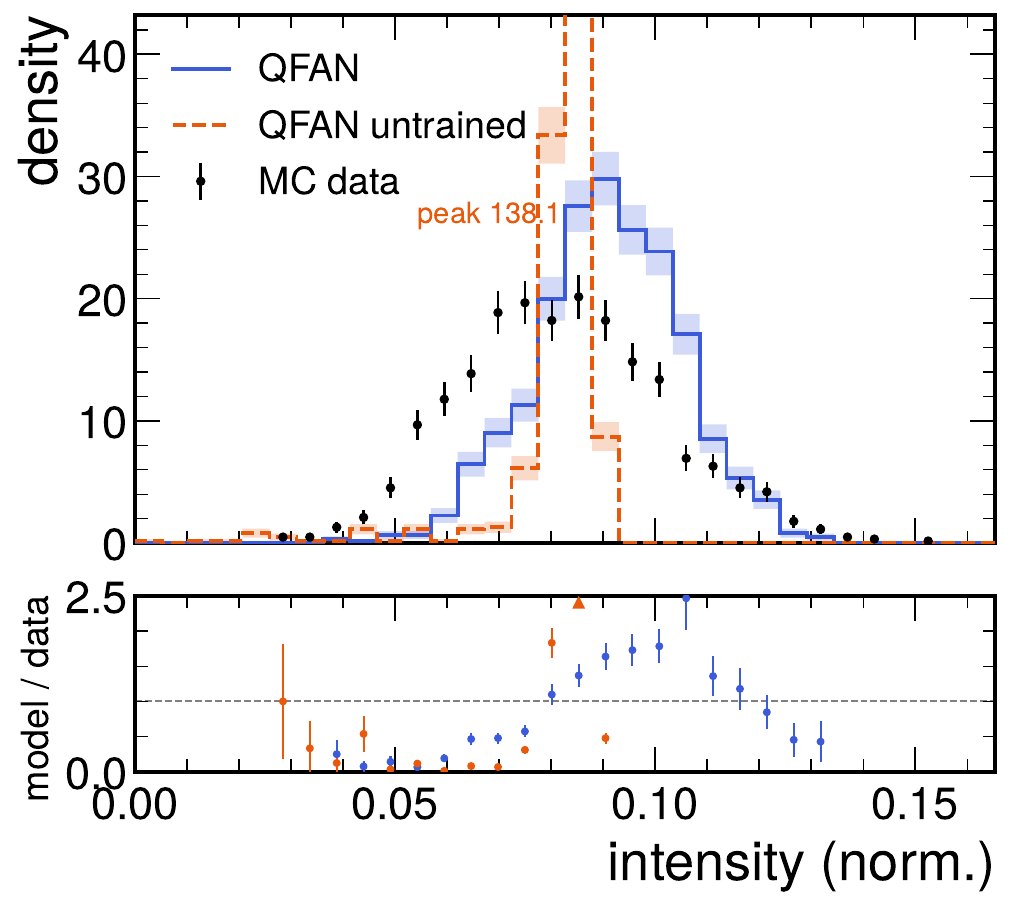}}\hfill
		\subfloat[Layer~2, ring~7]{\includegraphics[width=0.24\textwidth]{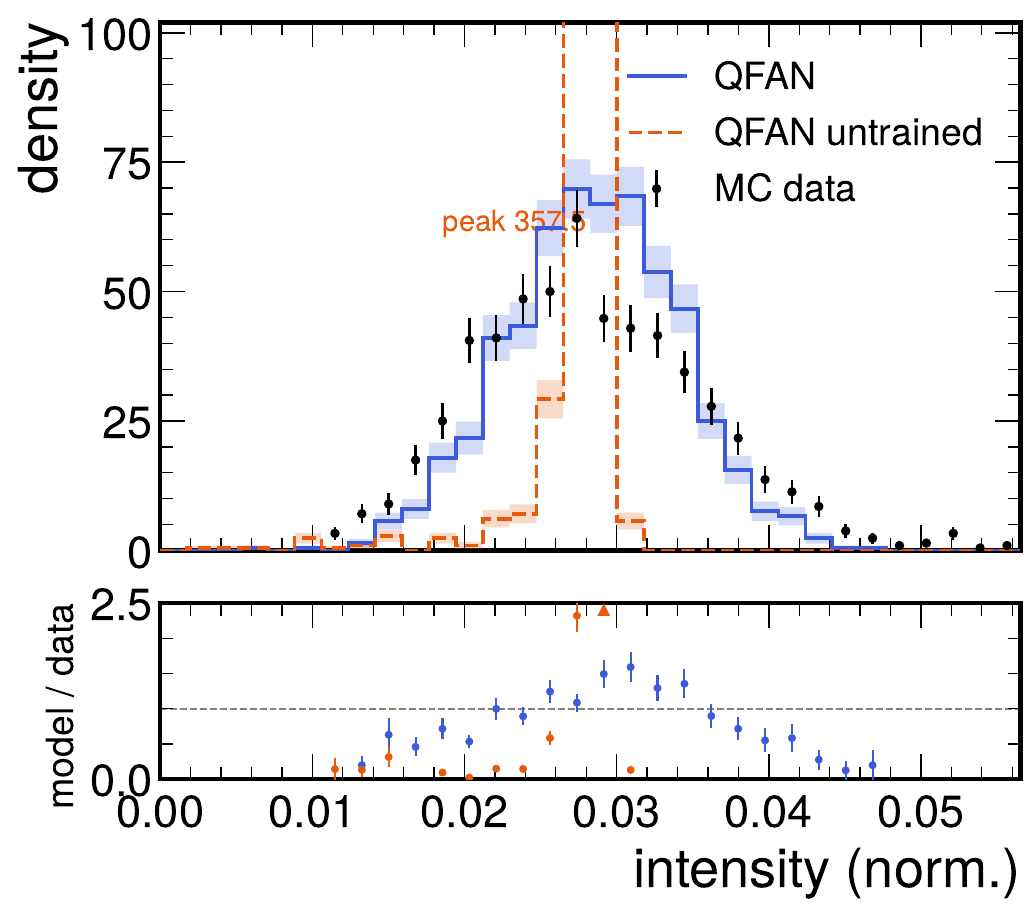}}\\
		\subfloat[Layer~2, ring~10]{\includegraphics[width=0.24\textwidth]{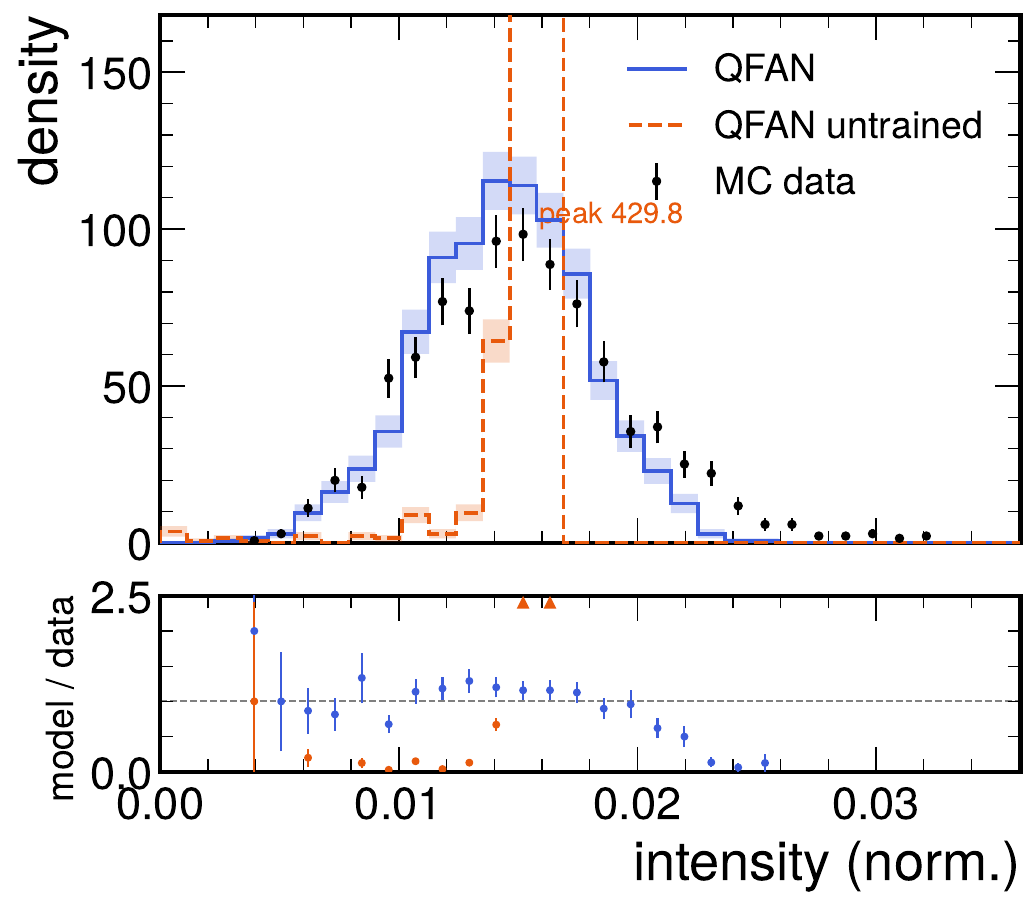}}\hfill
		\subfloat[Layer~2, ring~14]{\includegraphics[width=0.24\textwidth]{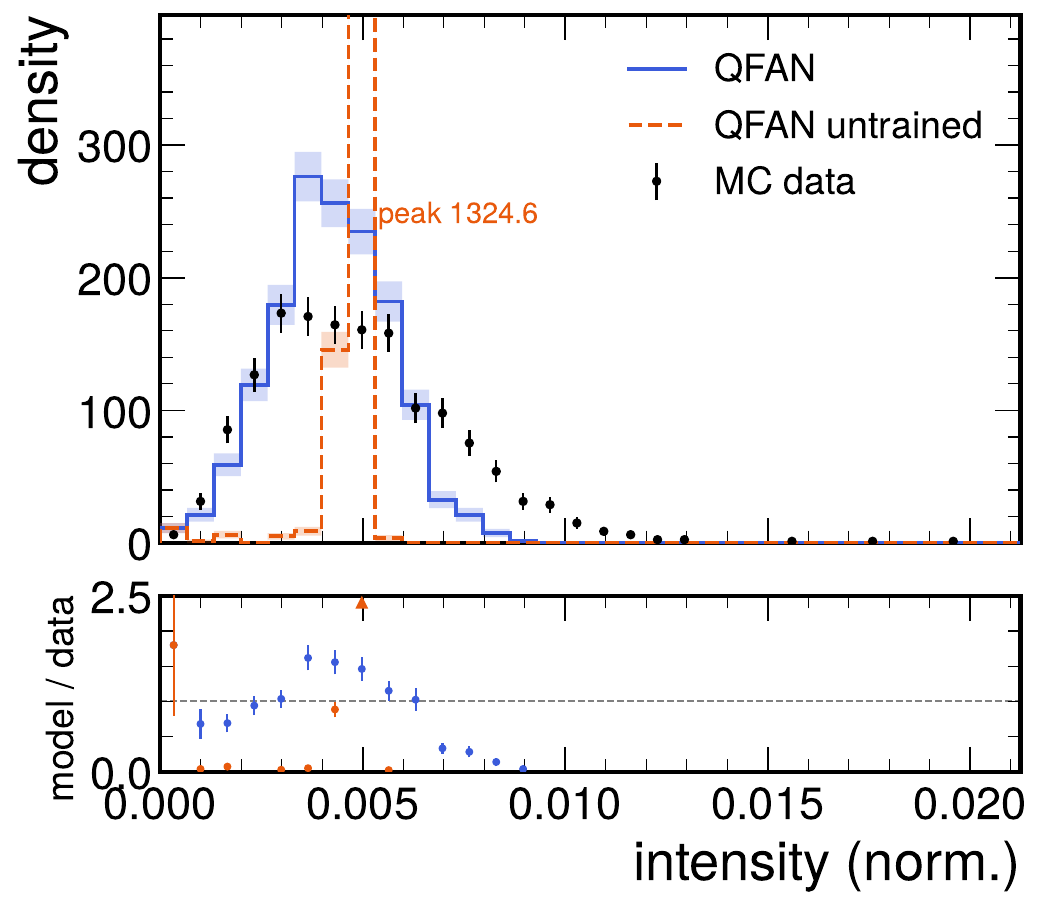}}\hfill
		\subfloat[Layer~2, ring~17]{\includegraphics[width=0.24\textwidth]{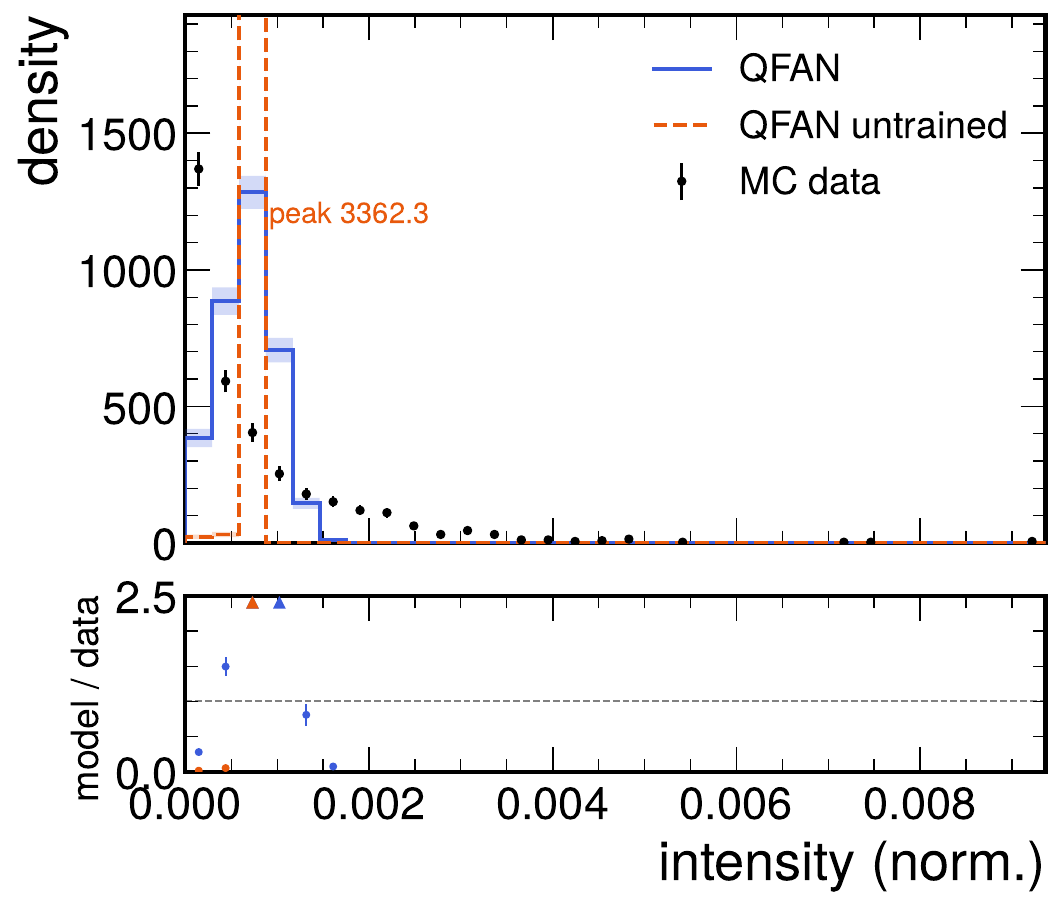}}\hfill
		\subfloat[Layer~12, ring~5]{\includegraphics[width=0.24\textwidth]{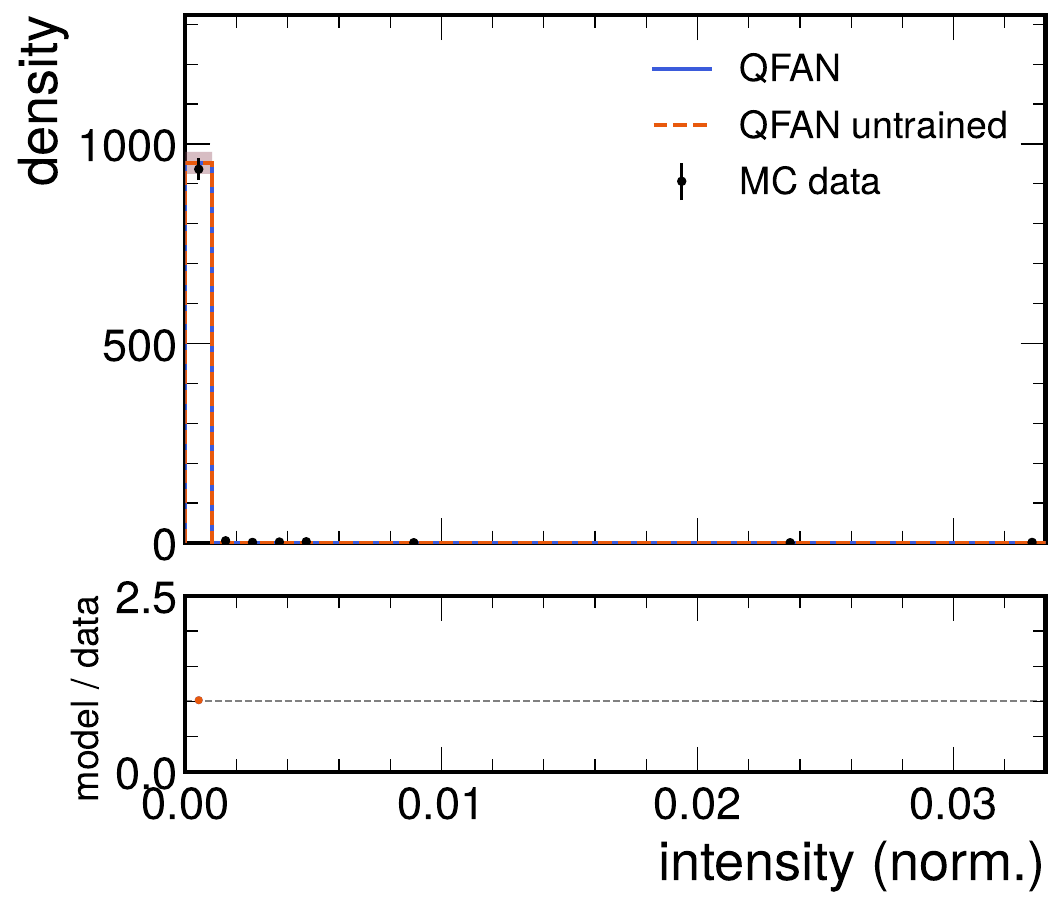}}
		\caption{\label{fig:ds1_voxels}Per-voxel marginal intensity distributions at 368 voxels for twelve voxels spread over the detector from layer~0 to layer~12, for MC data (points), QFAN (blue) and the untrained circuit (orange, dashed), with the ratio of each model to MC data underneath. Scales and markers are as in Fig.~\ref{fig:ds1_energy}.}
	\end{figure*}
	
	\begin{figure*}[tb]
		\centering
		\includegraphics[width=\textwidth]{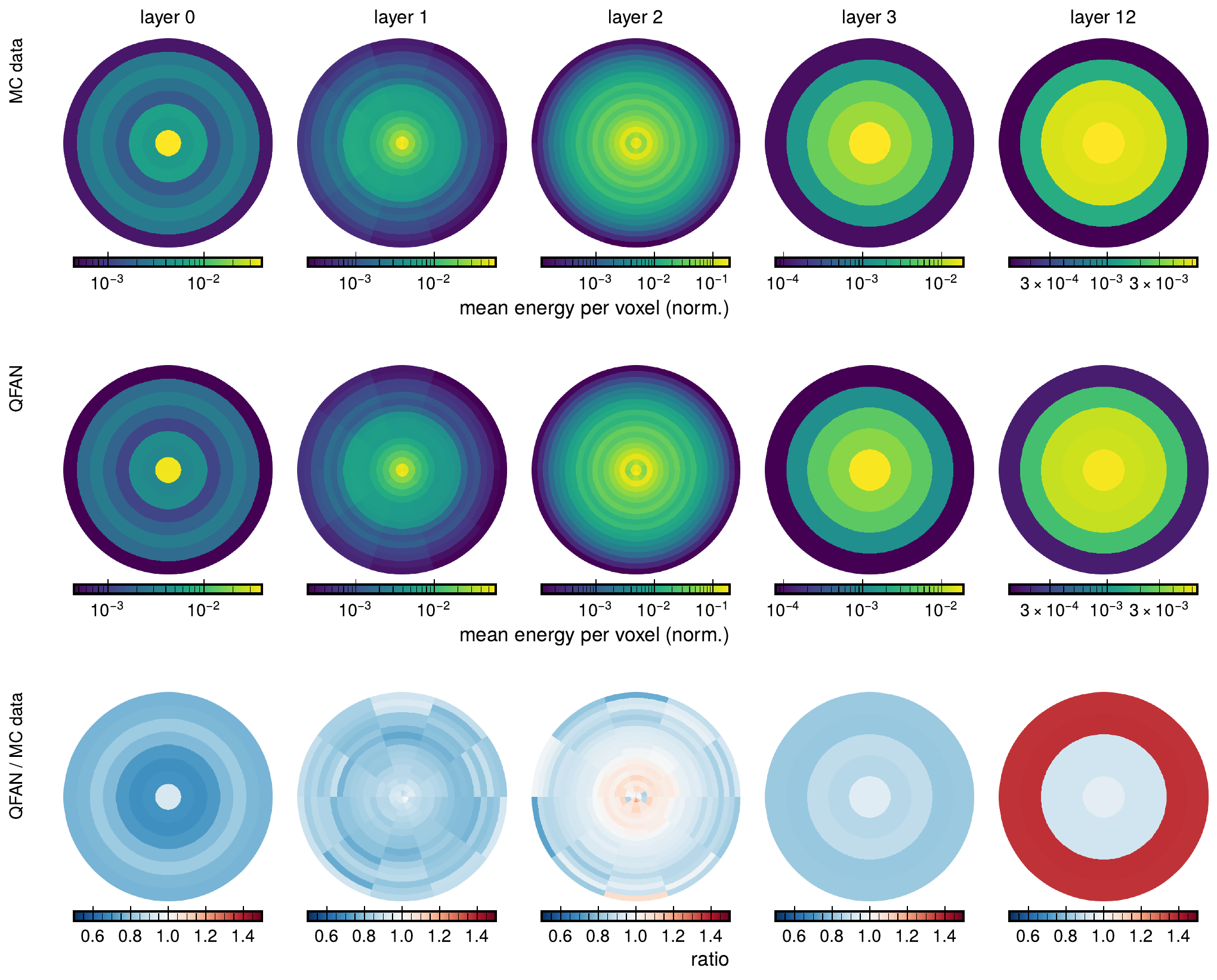}
		\caption{\label{fig:ds1_polar}Average shower per layer at 368 voxels, for MC data (top) and QFAN (middle) on a common logarithmic scale per layer, and their ratio (bottom). Rings are drawn with equal width.}
	\end{figure*}
	
	The chain used here is about thirty times longer than the one in Sec.~\ref{sec:results}. In the introduction we named exposure bias as the most likely way for such long chains to fail. At 185 blocks it clearly does not prevent the model from carrying the correlation structure and the layer-level fluctuations across the whole detector, although this conclusion rests on a single training seed at this resolution.
	
	\subsection{Dispersion limit}\label{sec:dispersion}
	
	The model does not reach the whole distribution, and it is worth being precise about where it falls short. The layer~0 and layer~12 energies reach only $17\%$ and $10\%$ of their \GEANT{} spread, the far tails of all spectra are missing, and sparse voxels are zero in $3.8\%$ of generated showers where \GEANT{} has $22\%$. None of this is specific to our circuit. It follows from the decoder, for any circuit and any number of records. To see why, consider a value whose features carry no information about its fluctuation, which is the situation of the total energy in the first block. Every row of the feature matrix is then the same vector $\bm{f}$, and the ridge fit of Sec.~\ref{sec:ridge}, with the shot-noise term $\bm{A}=\bm{\Sigma}/k$ and vanishing regularization, follows from the Sherman-Morrison identity as
	\begin{equation}
		\bm{w}=\frac{\bar y\,\bm{A}^{-1}\bm{f}}{1+q},\qquad q=\bm{f}^{\top}\bm{A}^{-1}\bm{f},
	\end{equation}
	where $\bar y$ is the mean of the value. The Born-noise variance of the decoded value is then
	\begin{equation}\label{eq:dispersion}
		\bm{w}^{\top}\bm{A}\,\bm{w}=\bar y^{\,2}\,\frac{q}{(1+q)^2}\le\frac{\bar y^{\,2}}{4},
	\end{equation}
	with equality at $q{=}1$. In other words, the standard deviation of the decoded value can never exceed half its mean. The circuit controls $\bm{f}$ and $\bm{\Sigma}$ and the budget controls $k$, so together they can move $q$ and with it the spread, but not past this bound. The bound comes from the shot-noise term in the decoder. A decoder that predicts the mean as well as possible from $k$-record averages suppresses any noise beyond this point, and the same mechanism limits the part of every value that the history does not predict. The total energy needs a coefficient of variation of about $0.33$ in prepared units, which lies within the bound, and the record budgets get it most of the way there. The sparse voxels are a different matter. On the test showers, 252 of the 368 voxel fractions have coefficients of variation above $0.5$, and 113 voxels are zero in more than $5\%$ of showers, both well beyond what the bound allows. Dropping the shot-noise term would remove the bound, but in a test at lower resolution it also handed the untrained circuit a realistic energy spread and realistic zero fractions. Those properties would then tell us nothing about the circuit, so we keep the term. Reaching the sparsest voxels will require a way of decoding Born records that is not fitted to their mean and does not add classical capacity.
	
	Three caveats apply. All results are from noiseless simulation, we used a single incident energy while Dataset~1 covers fifteen, and at this resolution we have trained only one seed. Classical generative models on the CaloChallenge also reach considerably higher fidelity~\cite{Krause_2025}. What this section shows is something different, namely that a quantum generator whose register does not grow with the image can be trained and evaluated at full benchmark resolution.

	% ═══════════════════════════════════════════════════════
	\section{Discussion}\label{sec:certificates}
	% ═══════════════════════════════════════════════════════
	
	The model is generated by the trained circuit and by the statistics of its measurements. We establish this by removing one element of the pipeline at a time and refitting the remainder on the same training data, so that the surviving stages absorb what they can. Table~\ref{tab:certificates} reports the study at $d{=}12$, and Table~\ref{tab:whatisquantum} lists which stages are trained.
	
	\begin{table*}[t]
		\centering
		\caption{Component removal at $d{=}12$ ($B{=}6$). Each row removes one component and refits all remaining stages on the same training data. ``off'' is the mean absolute error of the off-diagonal Pearson correlation matrix against the test set, ``$\bar W_1$'' the mean per-pixel Wasserstein-1 distance. The marginal-fidelity floor set by finite test statistics is $\bar W_1{=}0.0032$ at $n{=}1200$.}
		\label{tab:certificates}
		\small
		\begin{tabular}{@{}llccc@{}}
			\toprule
			& & \multicolumn{3}{c}{$d{=}12$, $B{=}6$} \\
			\cmidrule(lr){3-5}
			Configuration & Component removed & $\bar W_1$ & off & within \\
			\midrule
			QFAN, raw (simulator) & none & 0.0092 & 0.0925 & 0.0526 \\
			QFAN, calibrated & none & 0.0033 & 0.0526 & 0.0224 \\
			QFAN on \texttt{ibm\_fez} & none & 0.0203 & 0.1775 & 0.1113 \\
			\midrule
			Conditional means & Born measurement randomness & 0.0334 & \multicolumn{2}{c}{\textbf{degenerate}} \\
			Classical features & quantum feature map & 0.0113 & 0.1411 & 0.0996 \\
			Scrambled sketch & prefix conditioning & 0.0155 & 0.4936 & 0.2234 \\
			Untrained circuit & training of $\btheta$ & 0.0220 & 0.5153 & 0.2239 \\
			Plain least squares & noise-aware decoding & 0.0877 & 0.3339 & 0.1926 \\
			\bottomrule
		\end{tabular}
	\end{table*}
	
	\paragraph{The sampling step.} Replacing each block's $k$-shot record average by its exact conditional expectation removes the sampling fluctuation and changes nothing else. The rollout then becomes a deterministic map from the fixed initial sketch, producing one distinct sample among $1200$ and zero variance in every pixel, against a total-energy variance of $0.031$ in the data. Correlation is undefined for a constant, so those columns are marked \emph{degenerate}. The measurement fluctuation is the generative mechanism of the model, and the shared-record parameter $\rho_\beta$ of Eq.~\eqref{eq:share} shapes it into inter-pixel correlation.
	
	\paragraph{The training.} Freezing $\btheta$ at its random initialization and refitting the classical stages leaves a mean absolute off-diagonal correlation of $0.08$, against a data value of $0.58$, and per-pixel $\bar W_1$ at roughly seven times the statistical floor where the trained model sits at three. The untrained circuit reproduces neither the correlations nor the marginals. Training the shared circuit produces both, and its $12$ to $18$ parameters are the only quantities in the pipeline that adapt to the data.
	
	\paragraph{The conditioning.} Scrambling the sketch destroys the prefix information while preserving its statistics, and the correlation error rises to $0.494$ and $0.520$. The autoregressive conditioning therefore contributes directly to the correlation structure.
	
	\paragraph{The noise-aware decoder.} Decoding with a plain least-squares fit, which ignores the sampling covariance of the $k$-shot average, raises the correlation error from $0.093$ to $0.334$. Omitting the errors-in-variables term of Eq.~\eqref{eq:ridge} biases the decoder toward zero, and the attenuation propagates along the chain. 
	
	\paragraph{The feature map.} Substituting a classical feature map of the same width, with the sampling and decoding machinery unchanged, gives $0.141$ against $0.093$.
	
	\paragraph{Calibration.} The rank-preserving marginal map improves $\bar W_1$ by a factor of about three, from $0.0092$ to $0.0033$, and changes the correlation error by $0.04$. Per-pixel accuracy in the calibrated variant is therefore substantially a property of the calibration, while the correlation numbers are a property of the model in both variants.
	
	\begin{table}[t]
		\centering
		\caption{Stage-by-stage accounting. The circuit is the only stochastic element. All generative randomness is Born measurement randomness.}
		\label{tab:whatisquantum}
		\small
		\begin{tabular}{@{}llll@{}}
			\toprule
			Stage & Quantum & Trained & Role \\
			\midrule
			Sketch (hash, sign, $\bM$) & no & fixed & compress prefix \\
			Angle projection $(\bA,\bm{b})$ & no & \textbf{yes} & encode \\
			Circuit $U(\ba,\btheta)$ & \textbf{yes} & \textbf{yes} & generate \\
			Measurement records & \textbf{yes} & no & \textbf{all randomness} \\
			Record share $\rho_\beta$ & no & \textbf{yes} & intra-block corr. \\
			Noise-aware ridge & no & closed form & decode \\
			Marginal calibration & no & post hoc & marginals only \\
			\bottomrule
		\end{tabular}
	\end{table}
	
	Taken together, these locate the generative content of the model in two places. The sampling fluctuation supplies the randomness, and the training of the shared circuit supplies both the per-pixel spectra and the inter-pixel structure. The classical stages compress the prefix, decode the records and map the marginals, and they carry no trainable capacity beyond a closed-form solve. This is the proof of principle the architecture was built to demonstrate at $d{=}12$.

	\subsection{Register size and classical simulability}\label{sec:simulability}
	
	Ten qubits suffice for images with $\calO(10^4)$ pixels under the capacity rule, and circuits of that width are classically simulable. This subsection establishes the conditions under which a physical device becomes necessary.
	
	%	\paragraph{Where the present circuit sits.} The circuit used here is classically easy by every relevant measure, and by a wide margin. Its state is a $2^3$-dimensional vector. Its depth is $L{=}2$ with ring connectivity, so the entanglement across any cut is bounded by a small constant and a matrix-product-state simulation is exact with a tiny bond dimension \emph{at any width}. Its observables have weight at most two, so under backward Heisenberg evolution through two layers they spread to weight $\lesssim 6$ and Pauli-propagation methods evaluate them with a polynomial number of terms~\cite{paulipropagation,begusic2024}. Nothing in this work is evidence of quantum advantage. The hardware runs demonstrate portability of the circuit logic, not performance.
	
	\paragraph{Simulability of the present circuit.} The circuit is three qubits wide, with two layers at $d{=}12$ and three at $d{=}25$, and four qubits wide with three layers for Dataset~1, with ring connectivity and observables of weight at most two. At these sizes classical methods reproduce it directly. The state is a vector of dimension $2^3$ or $2^4$, the entanglement across any cut is bounded by a small constant so a matrix-product-state simulation is exact at small bond dimension, and the measured observables spread to weight $\lesssim 6$ under backward Heisenberg evolution, which Pauli-propagation methods evaluate with a polynomial number of terms~\cite{paulipropagation,begusic2024}. This is what makes the simulator path of Sec.~\ref{sec:dual} an exact reference against which the hardware runs are compared. 
	
	\paragraph{Register size within the present architecture.} The register size is set by the block size through $p_f{=}\tfrac32 n_q(n_q{+}1)$ (Eq.~\eqref{eq:pf}) and $b\le p_f/\rho_{\min}$, and the block size is set by how short one needs the autoregressive chain to be, $B{=}\lceil d/b\rceil$. Larger registers are therefore required precisely when short chains are required. Table~\ref{tab:whenquantum} makes this explicit at CaloChallenge Dataset~3 scale ($d{=}40\,500$). Tolerating $B{\sim}500$ steps needs only $n_q{=}10$, whereas restricting the chain to $B{\lesssim}20$ steps needs $n_q{\approx}45$, beyond exact state-vector simulation.
	
	Increasing $n_q$ at fixed depth and fixed observable weight~2 moves the model out of state-vector range while leaving it comfortably inside the range of tensor-network and Pauli-propagation surrogates, which are insensitive to width at this depth. A regime in which the \emph{feature map} is classically hard requires all three of width, depth, and observable weight to grow together. Enough qubits that state-vector simulation fails, enough depth that the Heisenberg-evolved observables spread over a finite fraction of the register, and observables whose support is not confined to a few qubits.
	
	Raising the observable weight to $w$ gives $p_f\sim(2n_q)^w/w!$ features and therefore \emph{reduces} the qubits needed for a given block size, pushing back toward the classically easy corner. Width and observable complexity cannot be increased independently in this architecture. Second, and more seriously for a near-term device, the same noise that makes deep circuits hard to execute also makes them easier to simulate. For noisy circuits the Pauli coefficients decay and truncated surrogate methods become accurate~\cite{angrisani2024}. A QFAN feature map that is simultaneously hard to simulate and faithfully executable would therefore require error rates well below what we used here, and plausibly error correction. This characterizes the resource scaling.
	
	\begin{table}[t]
		\centering
		\caption{Register size required at CaloChallenge Dataset~3 scale ($d{=}40\,500$) as a function of the autoregressive chain length one is willing to tolerate, under the capacity rule of Eq.~\eqref{eq:bmax}. Memory is for an exact single-precision state vector. Short chains force large registers. At the depth and observable weight used here, tensor-network and Pauli-propagation surrogates remain efficient at every entry in this table, so register width is one of three conditions rather than a sufficient one.}
		\label{tab:whenquantum}
		\small
		\begin{tabular}{@{}rrrrl@{}}
			\toprule
			$B$ & $b$ & $p_f$ & $n_q$ & state vector \\
			\midrule
			553 & 73 & 110 & 10 & $<1$\,MB \\
			94 & 433 & 650 & 25 & 268\,MB \\
			49 & 840 & 1260 & 35 & 275\,GB \\
			20 & 2053 & 3080 & 55 & 288\,PB \\
			10 & 4108 & 6162 & 78 & infeasible \\
			\bottomrule
		\end{tabular}
	\end{table}
	
	\paragraph{Sampling rather than estimation.} A variational model queries its circuit for few-body expectation values. QFAN queries it for measurement records, and a finite-$k$ average of those records is the model output, so the object computed is a \emph{sample} from the circuit's output distribution rather than an estimate of an observable. Section~\ref{sec:certificates} shows that the distinction is structural, since replacing the samples by the corresponding expectation values makes the generator deterministic.
	
	This matters for the simulability question because the two tasks have different classical status. Expectation values of weight $\le2$ observables on a shallow circuit are efficiently computable by surrogate methods. Sampling from the output distribution of even a constant-depth circuit is not believed to be classically tractable in general, and is the basis of several quantum-advantage arguments~\cite{bremner2011}, as well as the standard setting of quantum Born machines used as generative models~\cite{benedetti2019}. QFAN therefore sits in the class of algorithm for which a quantum-advantage argument could in principle be made, which is a stronger position than an expectation-value feature model occupies.
	
	At the present scale this is a structural statement rather than a performance one. Sampling from circuits of three or four qubits is classically immediate, since one computes at most $2^4$ output probabilities and draws from them, which is what the simulator path does. The hardness results for sampling apply asymptotically and to registers larger than the capacity rule demands at calorimeter scale. The architecture supplies the structural precondition. The quantity the device produces is a sample rather than an expectation value, and the register is no longer consumed by holding the image, so available qubits can be spent on the sampling task itself.
	
	\paragraph{What sets the required register size.} The two conditions above are linked by the scaling of sample fidelity with chain length. Large registers are demanded when the chain is kept short, and short chains are required when fidelity degrades with $B$. At $n_q{\sim}10$ the capacity rule accommodates benchmark geometries with long chains, while a fidelity constraint that caps $B$ drives the register into the regime of Table~\ref{tab:whenquantum}. Measuring sample fidelity as a function of $B$ therefore fixes the register size the architecture needs. Section~\ref{sec:ds1} provides a first such measurement, at $B{=}185$ on Dataset~1.
	
	\subsection{Weight-2 features and block size}
	
	Two further removals support the architectural choices of Sec.~\ref{sec:architecture}. Dropping the weight-2 Pauli operators halves $p_f$ and roughly doubles both error metrics at no saving in circuit count, since $G{=}2$ either way (Table~\ref{tab:ab_w2}). Sweeping the block size (Table~\ref{tab:ab_block}) shows a shallow optimum at $b{=}6$. At $b{=}12$ performance degrades sharply, and at $b{=}3$ it degrades mildly because the longer chain accumulates more sketch distortion.
	
	At $d{=}12$, setting $b{=}12$ forces $B{=}1$, which removes the autoregressive conditioning. That row therefore combines two effects, decoder capacity and loss of conditioning.
	
	\begin{table}[t]
		\centering
		\caption{Effect of weight-2 Pauli features at $d{=}12$.}
		\label{tab:ab_w2}
		\small
		\begin{tabular}{@{}lccc@{}}
			\toprule
			Config & $p_f$ & $\bar{W}_1$ & $\|\Delta\bm{C}\|_F/d$ \\
			\midrule
			Weight-1 only & 6 & 0.031 & 0.068 \\
			Weight-1$+$2 & 12 & 0.017 & 0.032 \\
			\bottomrule
		\end{tabular}
	\end{table}
	
	\begin{table}[b]
		\centering
		\caption{Block-size sweep at $n_q{=}3$, $d{=}12$. \textbf{Measured under the earlier variant of this architecture} (simultaneous-perturbation stochastic training~\cite{spall1992}, plain ridge decoding, classical residual sampler), which is the only configuration for which a block-size sweep exists. The $b{=}12$ row also forces $B{=}1$, which removes autoregressive conditioning.}
		\label{tab:ab_block}
		\small
		\begin{tabular}{@{}ccccc@{}}
			\toprule
			$b$ & $B$ & $\rho{=}p_f/b$ & $\bar{W}_1$ & $\|\Delta\bm{C}\|_F/d$ \\
			\midrule
			3 & 4 & 4.0 & 0.022 & 0.042 \\
			4 & 3 & 3.0 & 0.019 & 0.035 \\
			6 & 2 & 2.0 & 0.017 & 0.032 \\
			12 & 1 & 1.0 & 0.045 & 0.120 \\
			\bottomrule
		\end{tabular}
	\end{table}
	
	% ═══════════════════════════════════════════════════════
	\section{Resource Outlook}\label{sec:outlook}
	% ═══════════════════════════════════════════════════════
	
	%	This section collects what the measurements imply for larger geometries. Everything in it is a projection from a single demonstration scale, and we mark it as such rather than presenting it as a result.
	
	%	\paragraph{Hardware transfer.} The trained model runs on \texttt{ibm\_fez} at both image sizes with a measurable but bounded loss of fidelity (Table~\ref{tab:paths}).
	The off-diagonal correlation error changes from $0.093$ to $0.178$ at $d{=}12$ and from $0.128$ to $0.224$ at $d{=}25$, roughly a factor of two in both cases, and the
	bootstrap intervals of the simulator and hardware runs do not overlap. Marginal fidelity degrades by about a factor of two at $d{=}12$ and by $45\%$ at $d{=}25$, in both cases remaining well above the $n{=}200$ statistical floor of $0.0069$. At $d{=}25$ the correlation metric moves more than the marginal metric, which is the expected signature of a channel that damps measured parities. The decoder refit absorbs the linear part of that damping and restores per-pixel scale, while the
	residual affects the inter-pixel structure that only the records carry.
	
	Two features of the hardware runs are worth reporting because they were not anticipated. First, the device outperforms its own noise model. At $d{=}25$ the
	\texttt{ibm\_fez} run gives off-diagonal error $0.224$ whereas an Aer simulation using that backend's reported noise parameters gives $0.425$, a factor of $1.9$ worse. Part
	of this gap is a confound, since the noise-model run used $k{=}32$ records against $k{=}64$ on the device, but the direction of the discrepancy indicates that the
	calibration-derived noise model is pessimistic for this circuit. Second, the hardware over-disperses the total energy, $\Var(E){=}0.54$ against a data value of $0.17$ and a
	simulator value of $0.076$, whereas the simulator under-disperses. Device noise therefore does not simply attenuate the model. It injects additional variance into the
	autoregressive chain, which the sum observable accumulates, and this is a sharper diagnostic of hardware effects than the correlation error alone.

	This section collects what the measurements imply for larger geometries. The figures in it are projections from the scales demonstrated here.
	
	\paragraph{Width at fixed block size.} Training the same pipeline at $n_q\in\{2,3,4,5\}$ on the 25-pixel CLIC images of Appendix~\ref{app:25px}, with the block size, data, schedule and epoch budget held fixed, isolates the effect of register width. The off-diagonal correlation error falls monotonically, from $0.1445{\pm}0.0030$ at $n_q{=}2$ to $0.1368{\pm}0.0015$, $0.1298{\pm}0.0042$ and $0.1065{\pm}0.0018$ at $n_q{=}5$, a reduction of $26\%$ across the range, with uncertainties taken over three generation seeds. The trend is closest	to linear in the measured feature count, with slope $-0.00103$ per feature and
	$R^{2}{=}0.95$, and it shows no flattening at its upper end. Over this range the accuracy of the model is set by the width of the register rather than by the decoder
	or by the length of the chain.

	\paragraph{Hardware transfer.} The trained model runs on \texttt{ibm\_fez} with a measurable but bounded loss of fidelity (Table~\ref{tab:paths}). The off-diagonal correlation error changes from $0.093$ to $0.178$ and the marginal error $\bar W_1$ from $0.0092$ to $0.0203$, roughly a factor of two in both, and the bootstrap intervals of the simulator and hardware runs do not overlap. The marginal error remains well above the $n{=}200$ statistical floor of $0.0069$.

	\paragraph{Decoder capacity and chain length.} Since $\bW_\beta$ has rank at most $p_f$, the ratio $\rho\equiv p_f/b$ must exceed $1$ for the decoder to span the block. Taking a practical floor $\rho_{\min}\approx1.5$ gives
	
	\begin{align}\label{eq:bmax}
		b_{\max}(n_q)&=\left\lfloor\frac{3n_q(n_q+1)}{2\rho_{\min}}\right\rfloor, \text{and, }
		\\
		B_{\min}(d,n_q)&=\left\lceil\frac{2d\,\rho_{\min}}{3n_q(n_q+1)}\right\rceil .
	\end{align}
	Because $p_f$ grows quadratically in $n_q$, each added qubit reduces the required number of blocks roughly quadratically, which is the main reason modest increases in register size are valuable here.
	
	The value $\rho_{\min}{\approx}1.5$ is interpolated from a single block-size sweep at $d{=}12$ (Table~\ref{tab:ab_block}), measured under the earlier variant of this
	architecture, between a working point $\rho{=}2$ and a failing point $\rho{=}1$ that simultaneously forces $B{=}1$. The present configuration operates at $b{=}2$ and
	therefore $\rho{=}p_f/b{=}9$, far above the rank floor, so the capacity constraint is inactive throughout this work and $\rho_{\min}$ enters only through the extrapolation
	of Table~\ref{tab:bmax}. Repeating the sweep under the present architecture at an image size with $d>p_f$, where $b$ can be large while $B>1$, fixes those projections
	directly.
	
	\begin{table}[t]
		\centering
		\caption{Projected block counts under the heuristic of Eq.~\eqref{eq:bmax}, with
			$p_f{=}3n_q(n_q{+}1)/2$ and $\rho_{\min}{=}1.5$. These are extrapolations from a
			single measured operating point.}
		\label{tab:bmax}
		\small
		\begin{tabular}{@{}ccccc@{}}
			\toprule
			$n_q$ & $p_f$ & $b_{\max}$ & $B(d{=}6480)$ & $B(d{=}40500)$ \\
			\midrule
			3  & 18  & 12  & 540 & 3375 \\
			5  & 45  & 30  & 216 & 1350 \\
			6  & 63  & 42  & 155 &  965 \\
			8  & 108 & 72  &  90 &  563 \\
			10 & 165 & 110 &  59 &  369 \\
			\bottomrule
		\end{tabular}
	\end{table}
	
	\paragraph{Scope of the projections.} Table~\ref{tab:bmax} gives of order $10^2$ autoregressive steps for a 6480-voxel geometry at 8 qubits, and by
	Theorem~\ref{thm:cost} the per-step QPU cost of training is independent of $d$. These are resource statements. Sample fidelity at that chain length is governed by exposure
	bias and is quantified by the measurement described in Sec.~\ref{sec:next}.

	% ═══════════════════════════════════════════════════════
	\section{Comparison with Prior Work}\label{sec:comparison}
	% ═══════════════════════════════════════════════════════
	
	\begin{table*}[bt]
		\centering
		\caption{Architectural comparison with prior work. This table is qualitative. The listed models differ substantially in output representation, degree of classical assistance, and evaluation scale, so a single matched metric would be misleading.}
		\label{tab:comparison}
		\small
		\renewcommand{\arraystretch}{1.25}
		\setlength{\tabcolsep}{5pt}
		\begin{tabular}{@{}p{2.8cm}p{1.9cm}p{2.2cm}p{2.5cm}p{2.3cm}p{2.6cm}@{}}
			\toprule
			\textbf{Model} & \textbf{QPU or QA tested} & \textbf{Quantum role} & \textbf{Output representation} & \textbf{Register scaling issue} & \textbf{Representative demonstrated scale} \\
			\midrule
			\textbf{QFAN}
			& \checkmark
			& shared feature generator
			& blockwise direct pixels
			& mitigated by autoregression
			& $d{=}12$, 25 and 368 (sim.) \\
			\textbf{Dual-PQC qGAN}~\cite{chang2021_dualpqcgan}
			& Sim
			& direct image generator
			& reduced-size direct pixels
			& tied to image output
			& reduced-size calorimeter images \\
			\textbf{Full qGAN}~\cite{rehm2023_fullqgan}
			& Sim / small-device study
			& direct image generator
			& downsized direct pixels
			& tied to image output
			& downsized 8-pixel showers \\
			\textbf{CaloQVAE}~\cite{caloqvae2024}
			& QA-assisted
			& latent-space sampler
			& hybrid latent model
			& shifted to latent model
			& calorimeter surrogate \\
			\textbf{Cond.\ QA surrogate}~\cite{toledomarin2025}
			& QA-assisted
			& conditioned latent-space sampler
			& conditioned hybrid latent model
			& shifted to latent model
			& CaloChallenge Dataset 2 \\
			\bottomrule
		\end{tabular}
	\end{table*}
	
	Table~\ref{tab:comparison} places QFAN among early direct quantum calorimeter proof-of-concepts and more recent quantum-assisted generators. The comparison is architectural rather than benchmark-level, since these models differ in whether they emit downscaled pixels directly or work through hybrid latent constructions. For QFAN, the ``12 shared parameters'' figure refers only to the trainable circuit parameters. The full model also contains the classical sketch, decoder and calibration stages accounted for in Table~\ref{tab:whatisquantum}.
	Related quantum generative work in HEP beyond shower generation includes qGAN-based anomaly detection~\cite{bermot2023_anomaly}.
	
	\section{Outlook}\label{sec:next}
	
	The demonstrations reported here are at $d{=}12$ and $d{=}25$ with three qubits and at the full 368-voxel resolution of CaloChallenge Dataset~1 with four qubits, and the architecture admits several direct extensions from that base.
	
	The scaling of sample fidelity with chain length $B$ fixes the register size the method requires at benchmark geometries. Section~\ref{sec:ds1} shows that free-running generation along $B{=}185$ blocks retains about half of the achievable correlation structure. Repeating it over training seeds, conditioning on the incident energy across the fifteen energies of Dataset~1, and extending it to Dataset~2 are the next steps. A block-size sweep under the present pipeline, at an image size large enough that $b$ can be large while $B>1$, determines $\rho_{\min}$ directly and sharpens the projections of Table~\ref{tab:bmax}. A simulator run at the hardware optimization budget resolves the two contributions to the transfer of Sec.~\ref{sec:dual} into their separate parts. Comparison against a trained classical encoder at equal feature count extends the study of Sec.~\ref{sec:certificates}.
	
	Beyond these, the natural direction is to widen the register. Each additional qubit raises $p_f$ quadratically and shortens the chain accordingly, and the per-step training cost of Theorem~\ref{thm:cost} is unchanged by it. Whether the additional width also improves sample fidelity at fixed block size is measurable with the same machinery used here.
	
	The dispersion limit of Sec.~\ref{sec:dispersion} marks the remaining gap at 368 voxels. Reaching the sparsest voxels, and the far tails of the spectra, requires decoding Born records in a way that is not fitted to their mean, without adding classical capacity.

	\section{Conclusion}\label{sec:conclusion}
	
	QFAN decouples the number of qubits from the image size. Instead of one qubit per pixel, one small circuit is reused across consecutive blocks, conditioned each time on a compact summary of the pixels generated so far, so that adding pixels adds autoregressive steps rather than qubits. For the observable family measured here the number of circuits per training step is independent of the image size, and the block-size trade-off links qubit count to the required chain length.
	
	At $d{=}12$ the architecture reproduces per-pixel spectra, inter-pixel correlations and total deposited energy on a noiseless simulator and on \texttt{ibm\_fez}, using three qubits and 12 shared parameters, and the hardware run extends to 25-pixel images. Removing the measurement fluctuation, or the training of the circuit, removes both the correlations and the marginals, which places the generative content of the model in the trained circuit and in the statistics of its measurements. Widening the register from two to five qubits at fixed block size, on 25-pixel images, reduces the correlation error by $26\%$ with no sign of saturation.
	
	At the full 368-voxel resolution of CaloChallenge Dataset~1 the same architecture, with four qubits and a free-running chain of 185 blocks, captures about half of the achievable correlation structure and most of the energy fluctuation of the two main calorimeter layers, where the untrained circuit produces an almost deterministic shower. Adding voxels added autoregressive steps and not qubits, as designed.
	
	The construction is a proof of principle rather than a competitive surrogate. What it establishes is that the register no longer has to hold the image, which is the precondition for spending qubits on anything else.
	
	\section*{Code and Data Availability}\label{sec:code}
	
	The code implementing the architecture, the training pipeline and all analyses is available at \url{https://github.com/jamalslim/qfan_project}. The Dataset~1 pipeline, including the data preparation, the per-block record budgets, the evaluation with its controls and all figures of Sec.~\ref{sec:ds1}, is in its \texttt{ds1} directory.

	\begingroup
	\captionsetup[subfigure]{labelformat=empty}
	\begin{figure*}[htb]
		\centering
		\subfloat[Pixel 0]{\includegraphics[width=0.19\textwidth]{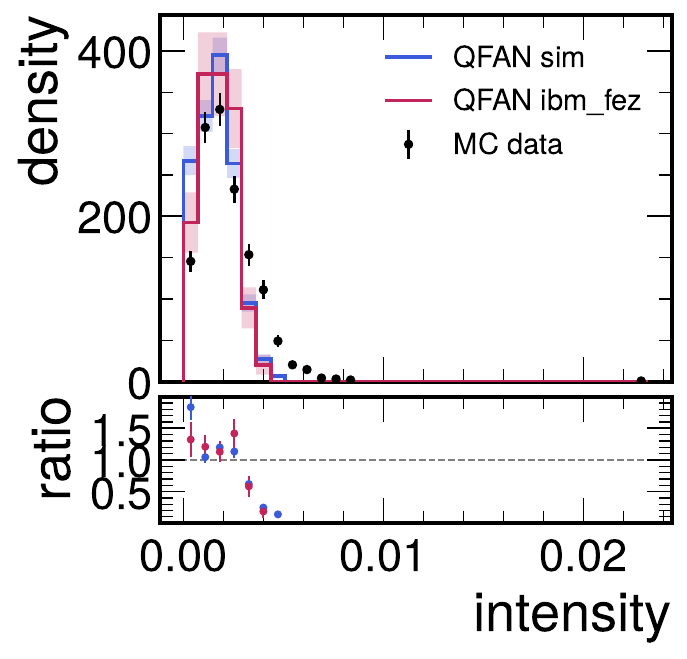}}\hfill
		\subfloat[Pixel 1]{\includegraphics[width=0.19\textwidth]{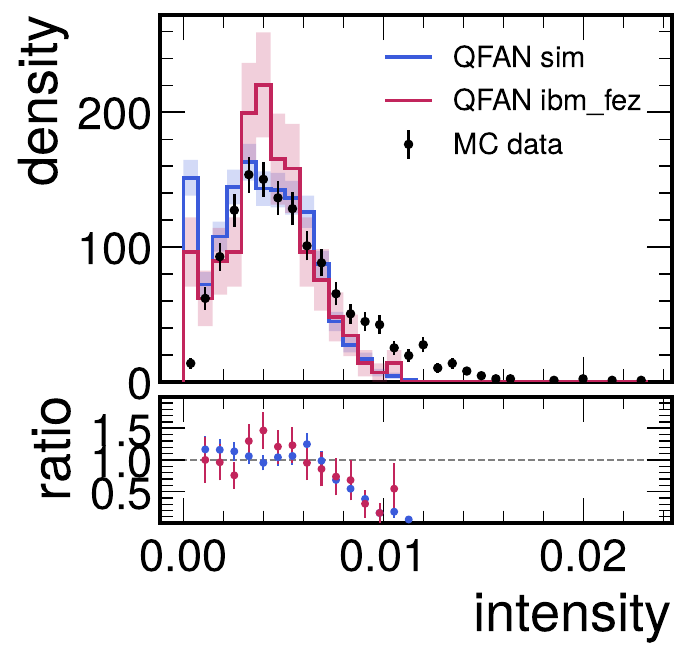}}\hfill
		\subfloat[Pixel 2]{\includegraphics[width=0.19\textwidth]{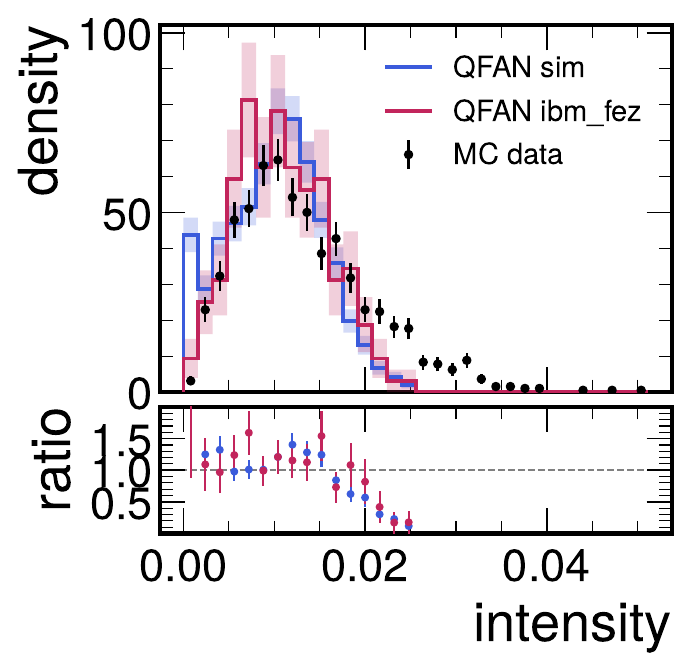}}\hfill
		\subfloat[Pixel 3]{\includegraphics[width=0.19\textwidth]{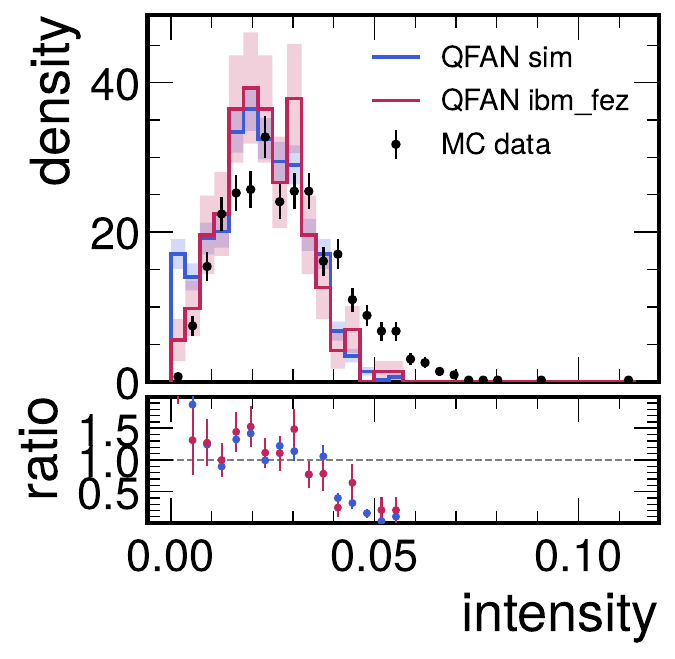}}\hfill
		\subfloat[Pixel 4]{\includegraphics[width=0.19\textwidth]{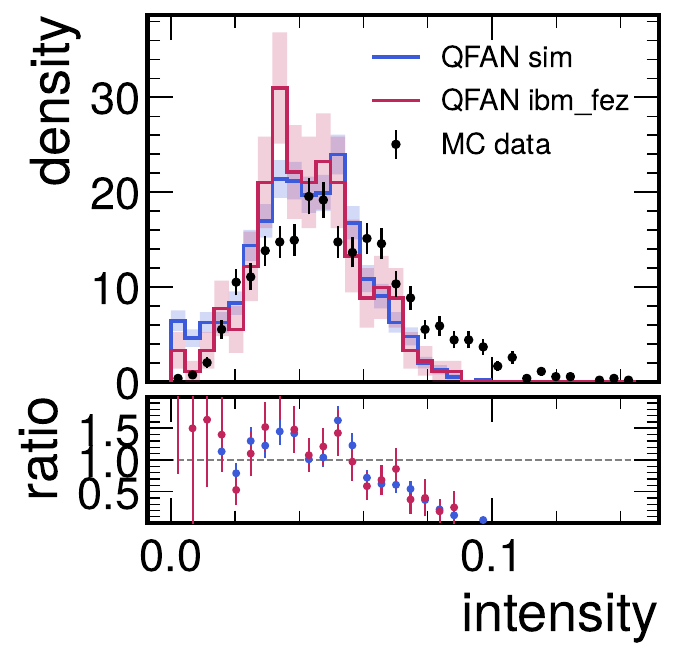}}\\
		\subfloat[Pixel 5]{\includegraphics[width=0.19\textwidth]{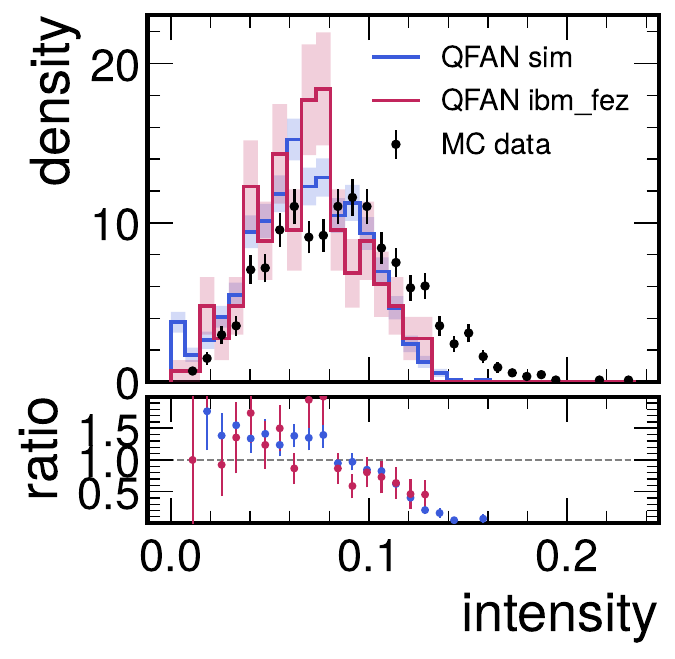}}\hfill
		\subfloat[Pixel 6]{\includegraphics[width=0.19\textwidth]{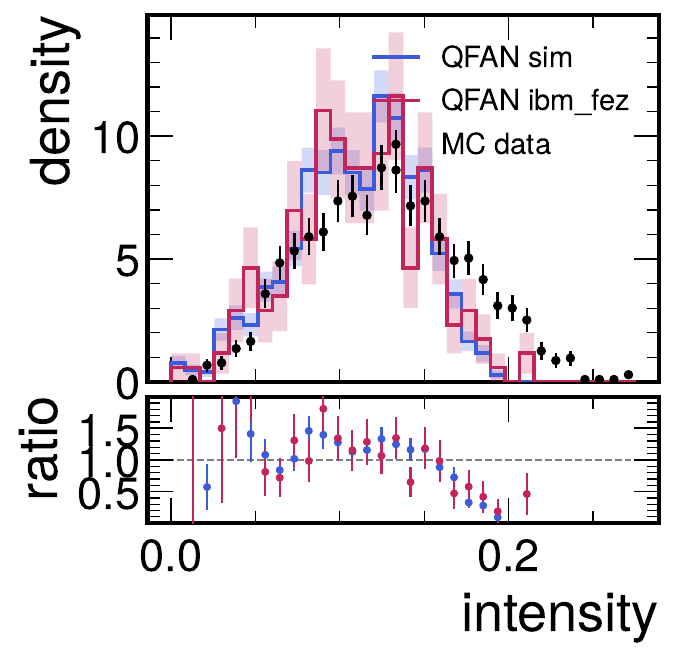}}\hfill
		\subfloat[Pixel 7]{\includegraphics[width=0.19\textwidth]{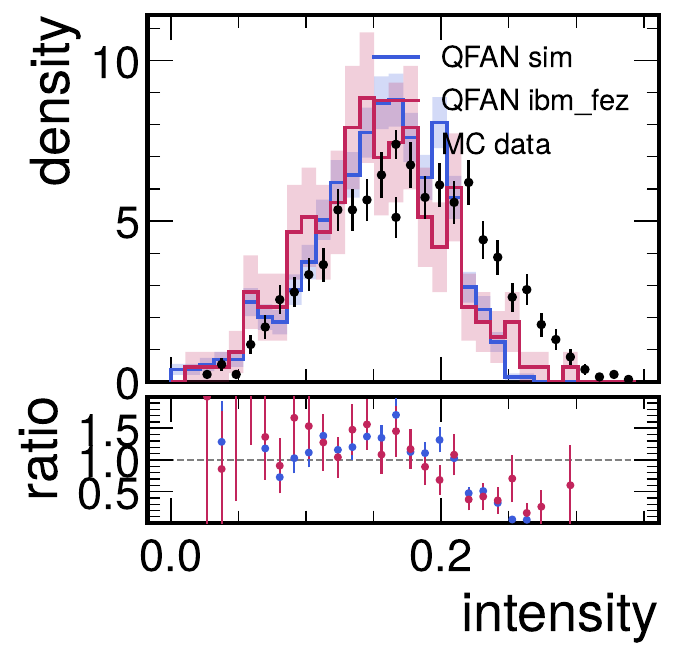}}\hfill
		\subfloat[Pixel 8]{\includegraphics[width=0.19\textwidth]{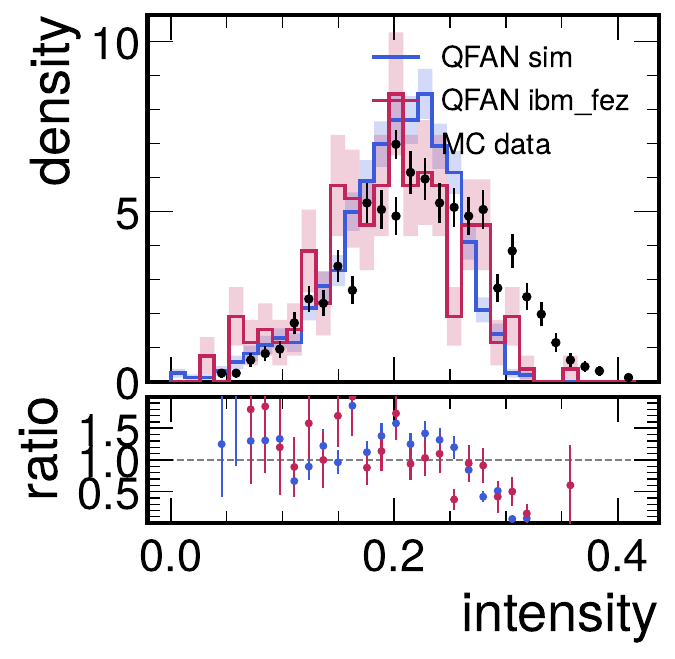}}\hfill
		\subfloat[Pixel 9]{\includegraphics[width=0.19\textwidth]{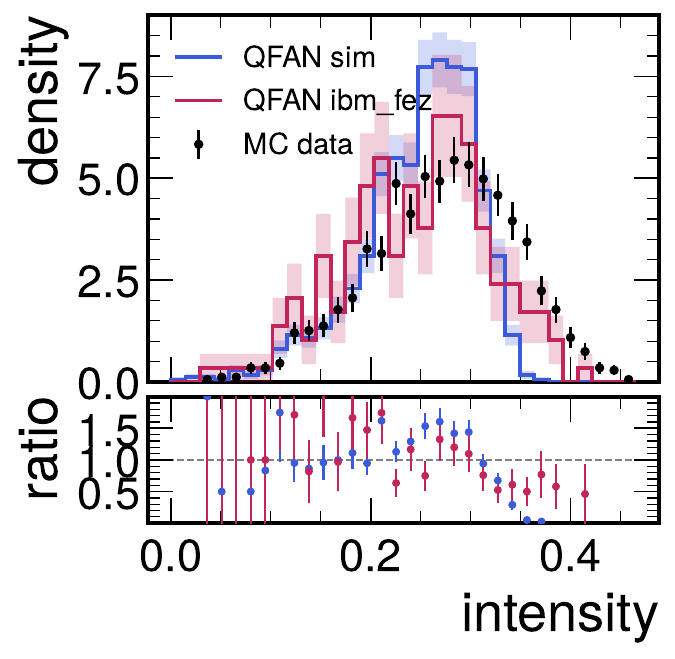}}\\
		\subfloat[Pixel 10]{\includegraphics[width=0.19\textwidth]{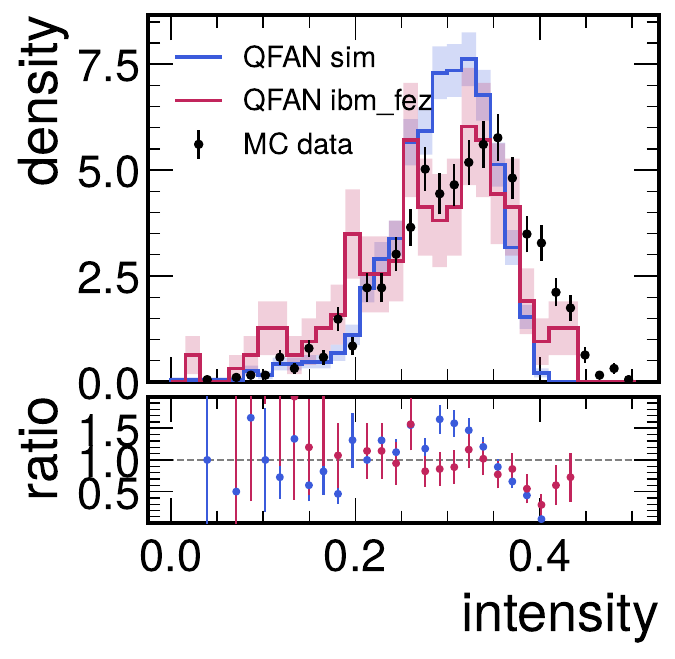}}\hfill
		\subfloat[Pixel 11]{\includegraphics[width=0.19\textwidth]{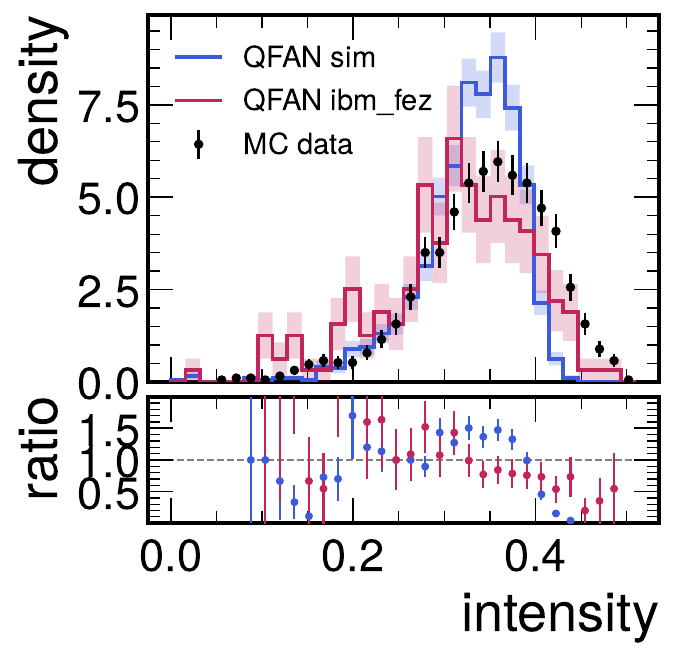}}\hfill
		\subfloat[Pixel 12]{\includegraphics[width=0.19\textwidth]{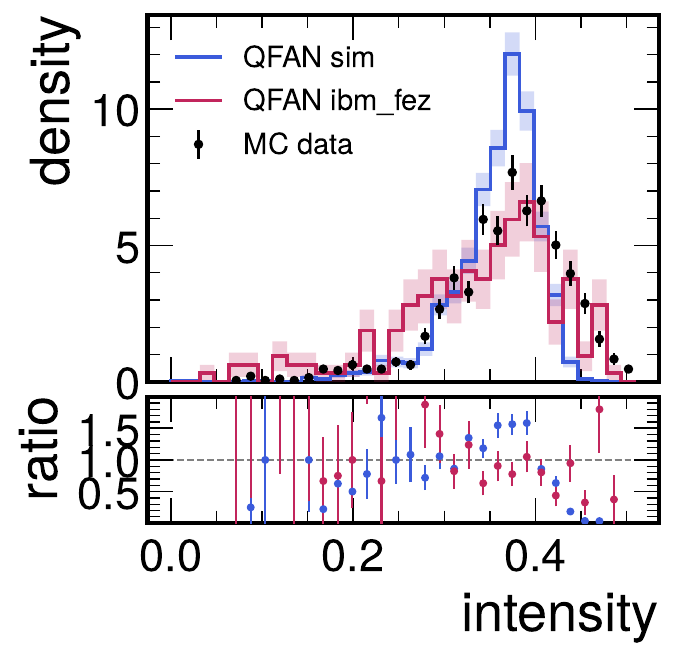}}\hfill
		\subfloat[Pixel 13]{\includegraphics[width=0.19\textwidth]{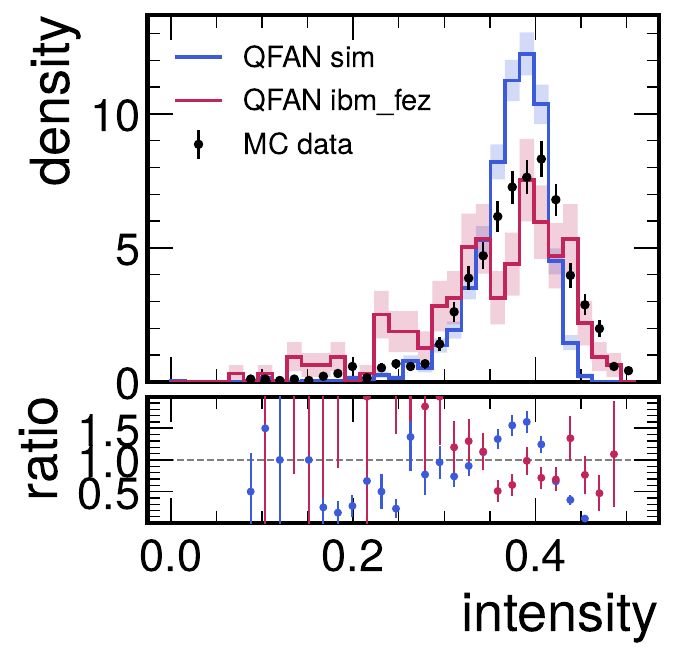}}\hfill
		\subfloat[Pixel 14]{\includegraphics[width=0.19\textwidth]{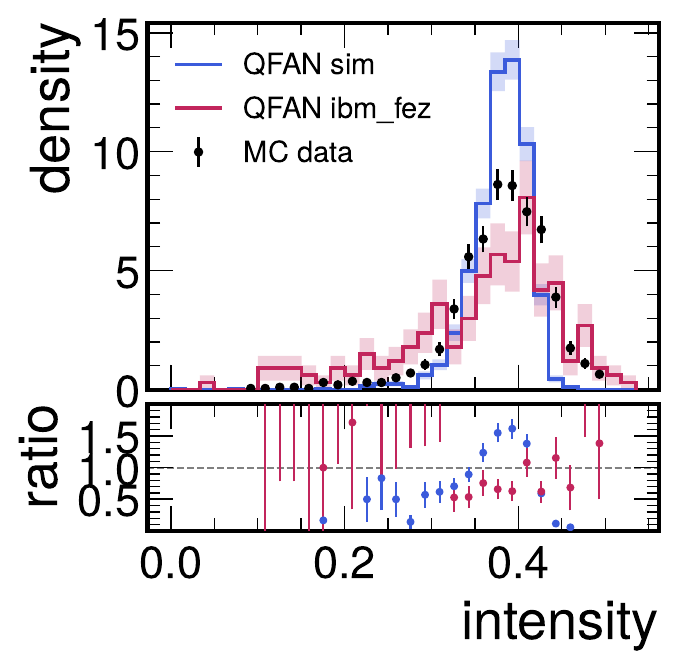}}\\
		\subfloat[Pixel 15]{\includegraphics[width=0.19\textwidth]{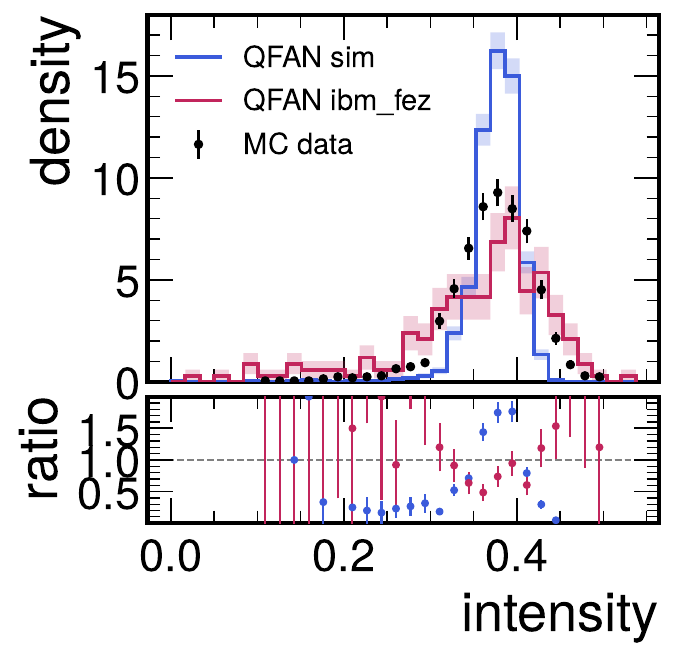}}\hfill
		\subfloat[Pixel 16]{\includegraphics[width=0.19\textwidth]{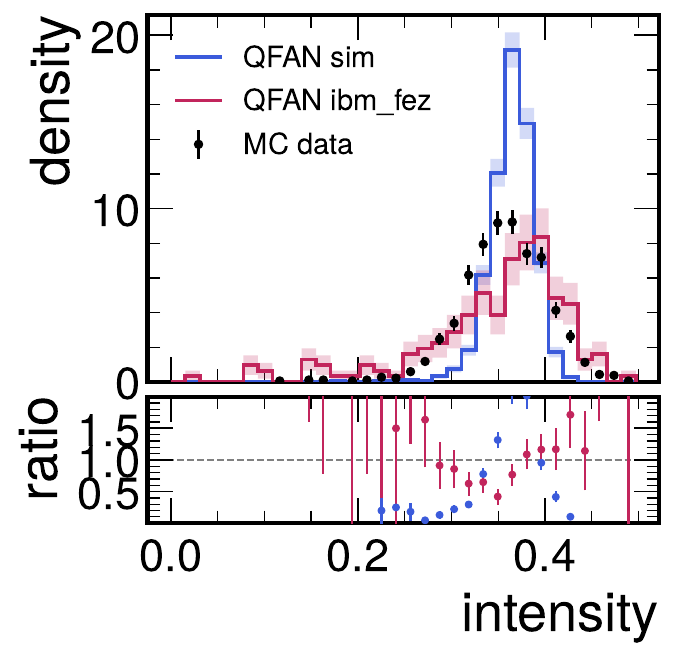}}\hfill
		\subfloat[Pixel 17]{\includegraphics[width=0.19\textwidth]{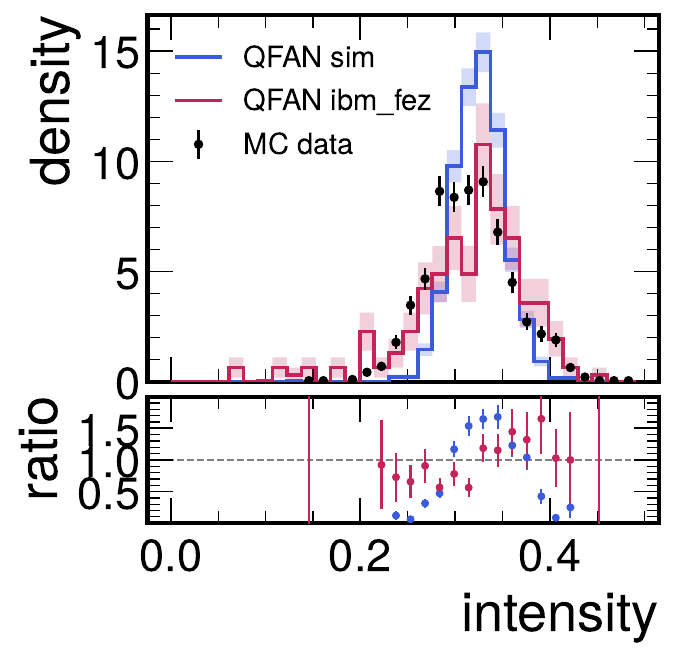}}\hfill
		\subfloat[Pixel 18]{\includegraphics[width=0.19\textwidth]{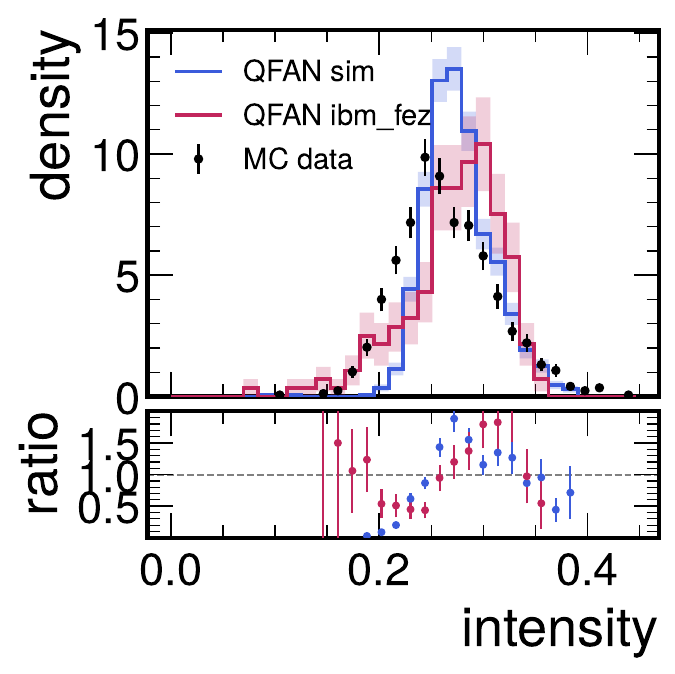}}\hfill
		\subfloat[Pixel 19]{\includegraphics[width=0.19\textwidth]{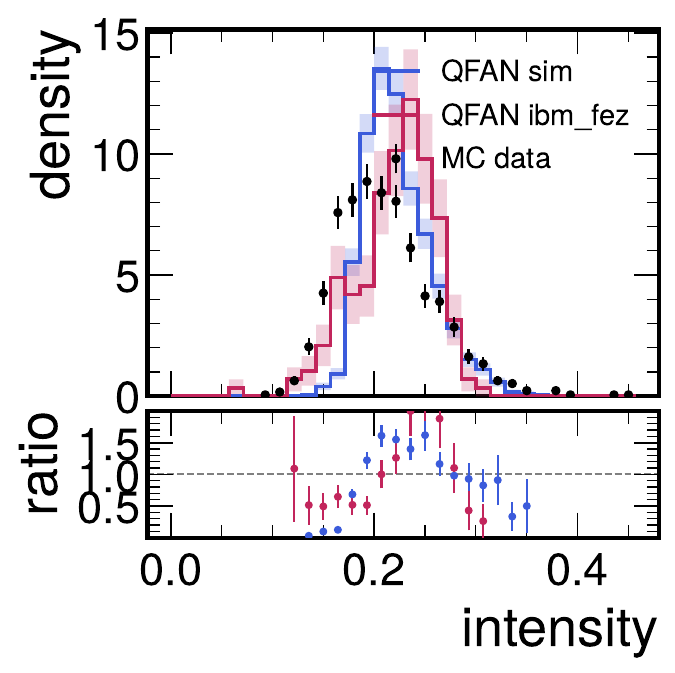}}\\
		\subfloat[Pixel 20]{\includegraphics[width=0.19\textwidth]{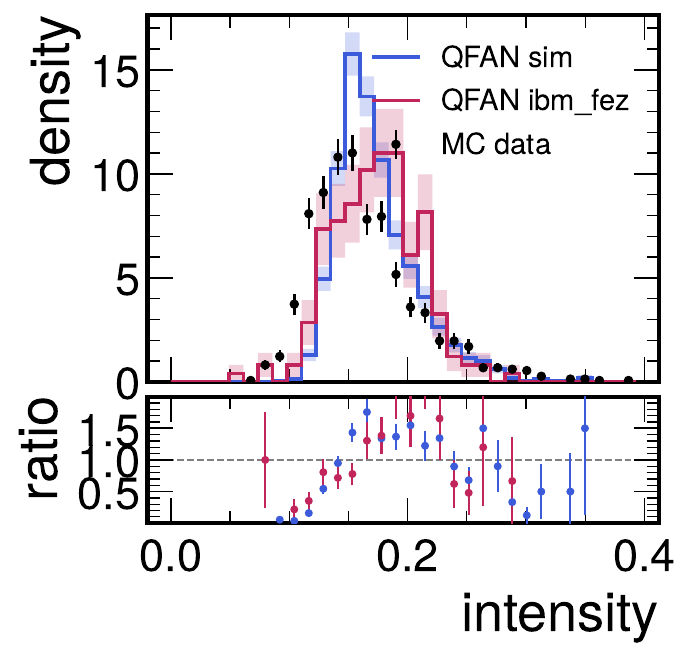}}\hfill
		\subfloat[Pixel 21]{\includegraphics[width=0.19\textwidth]{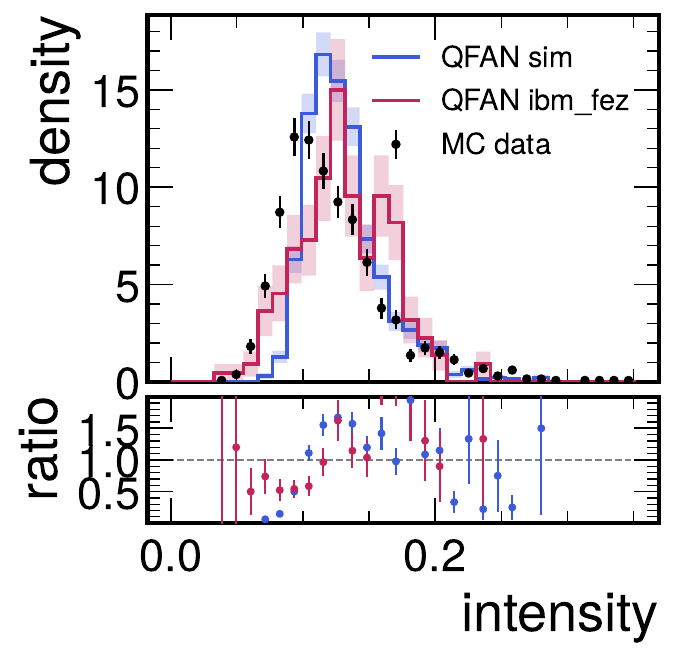}}\hfill
		\subfloat[Pixel 22]{\includegraphics[width=0.19\textwidth]{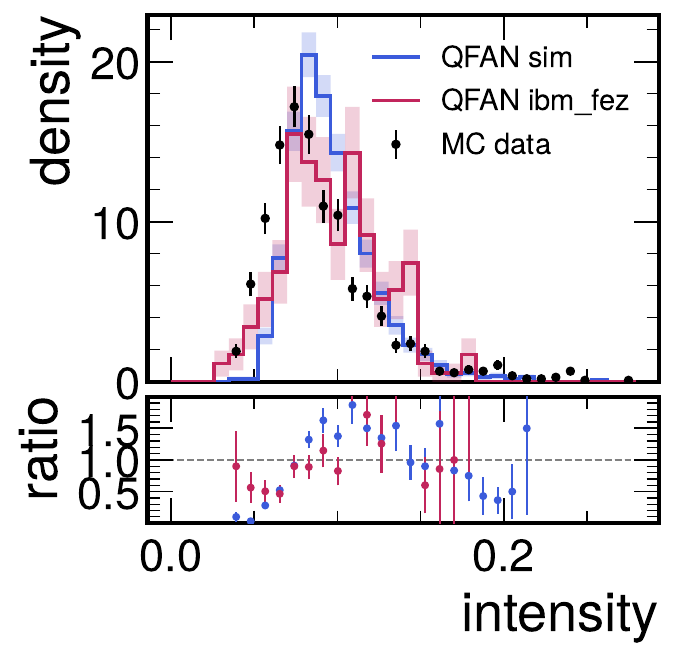}}\hfill
		\subfloat[Pixel 23]{\includegraphics[width=0.19\textwidth]{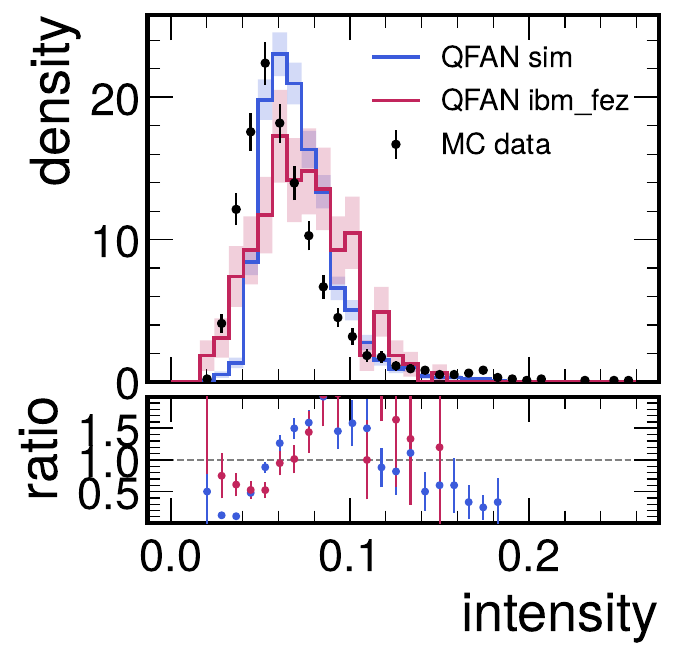}}\hfill
		\subfloat[Pixel 24]{\includegraphics[width=0.19\textwidth]{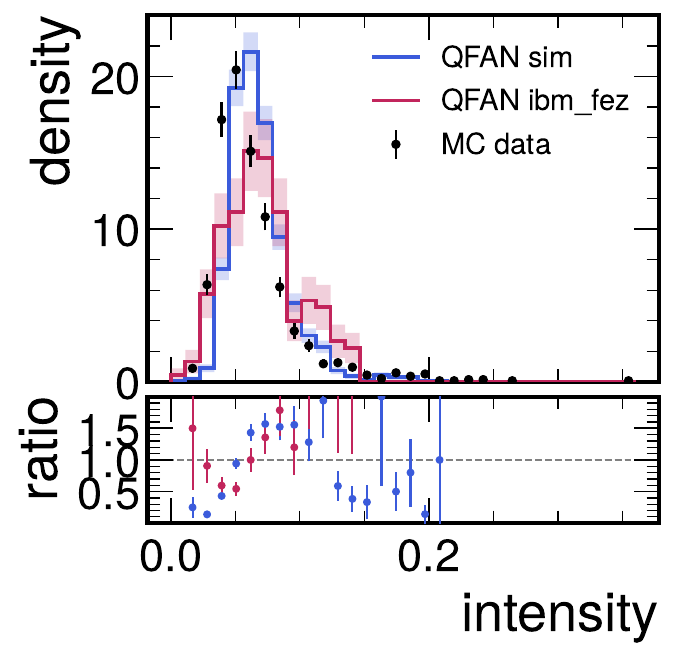}}
		\caption{Per-pixel marginals at $d{=}25$, $B{=}13$ (variable intensity scale). MC data (black points), simulator QFAN (blue) and \texttt{ibm\_fez} QFAN (red), with ratio panels. The ratio panels scatter about unity without systematic drift toward the later blocks.}
		\label{fig:marginals_25px}
	\end{figure*}
	\endgroup
	
	\begin{figure*}[thb]
		\centering
		\subfloat[MC data]{\includegraphics[width=0.32\textwidth]{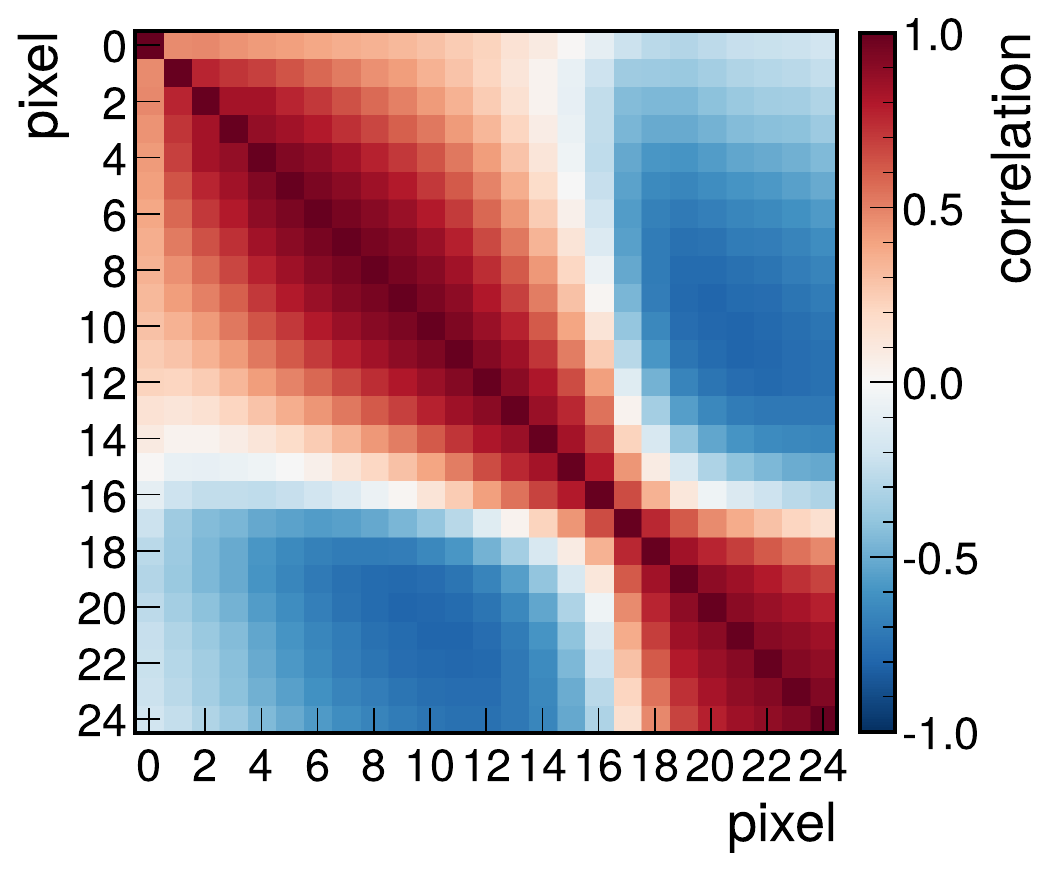}}\hfill
		\subfloat[QFAN (simulator)]{\includegraphics[width=0.32\textwidth]{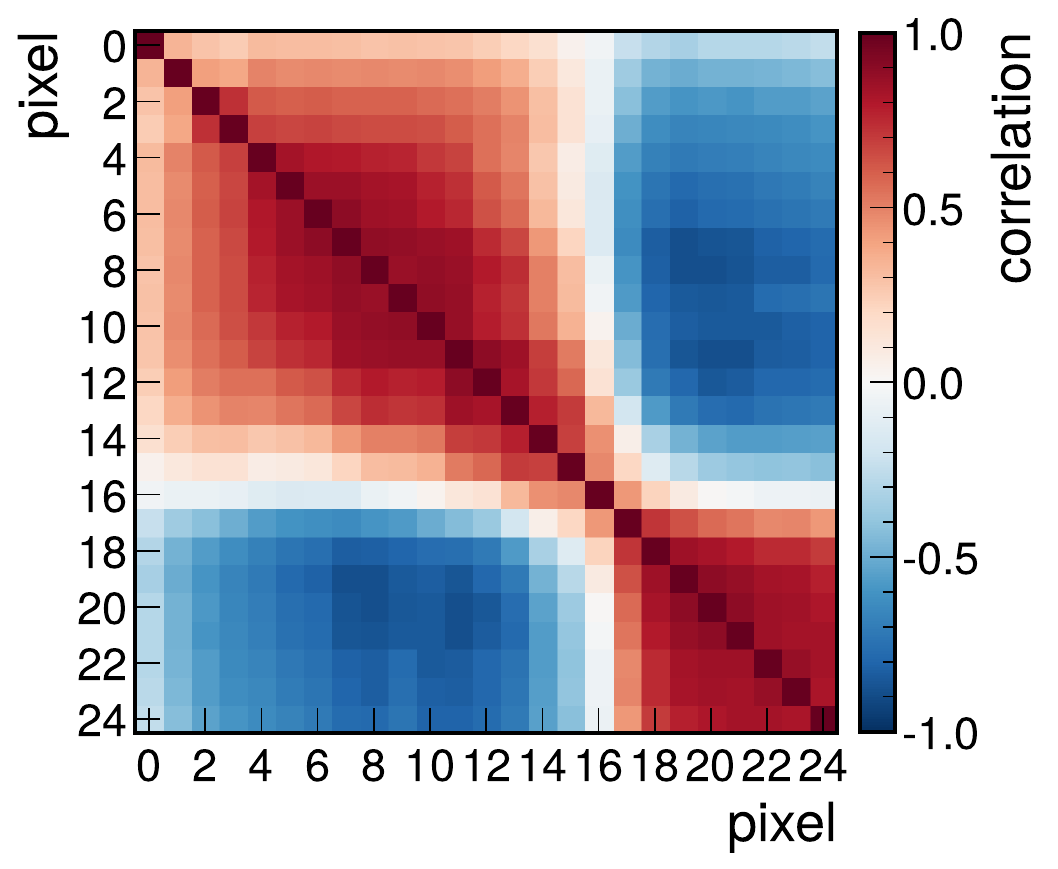}}\hfill
		\subfloat[QFAN (\texttt{ibm\_fez})]{\includegraphics[width=0.32\textwidth]{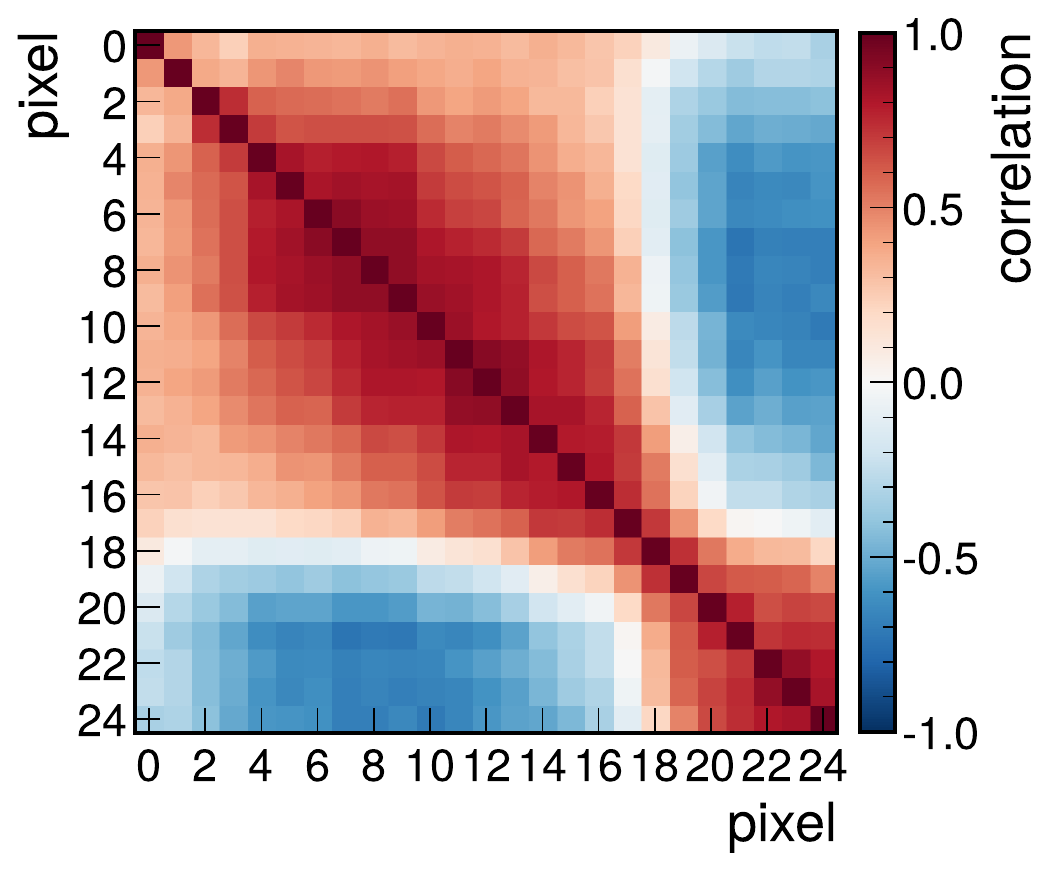}}
		\caption{Pearson correlation matrices at $d{=}25$, $B{=}13$. MC data (a), simulator QFAN (b) and hardware QFAN (c). The block-diagonal structure and the negative cross-block correlations are reproduced on both backends. Residuals concentrate at block boundaries and do not grow from the first boundary to the last.}
		\label{fig:corr_25px}
	\end{figure*}
	
	\section*{Use of AI Tools}
	
	The authors conceived the QFAN architecture, directed all aspects of this work, and made all scientific decisions. A large language model (ChatGPT) assisted with grammar and sentence structure throughout the manuscript, supported the literature review and cross-referencing with prior work, and assisted in writing and commenting on the training and visualization code. The authors independently reproduced all reported results using the released code and take full responsibility for the content of this work.
	
	\section*{Acknowledgments}
	
	This research was supported in part through the Maxwell computational resources operated at Deutsches Elektronen-Synchrotron DESY (Hamburg, Germany), a member of the Helmholtz Association HGF. The authors acknowledge the support with funds from the Helmholtz Association HGF (Germany), Hamburgische Investitions- und F\"orderbank (IFB) (Germany), European Union's HORIZON MSCA Doctoral Networks program project ENGAGE (101034267), and the Ministry of Science, Research and Culture of the State of Brandenburg within the Center for Quantum Technologies and Applications (CQTA) (Germany).
	
	% ═══════════════════════════════════════════════════════
	\appendix
	% ═══════════════════════════════════════════════════════
	
	\section{Sketch distortion}\label{app:sketch}
	
	The count-sketch of Eq.~\eqref{eq:sketch_update} corresponds to a matrix $S\in\R^{m\times d}$ with exactly one nonzero entry per column, equal to a random sign in a uniformly chosen row. The induced inner-product estimator $\widehat{\langle\by,\by'\rangle}\equiv\langle S\by,S\by'\rangle$ is unbiased, and for standard count-sketch randomness satisfies
	\begin{equation}
		\E\!\left[\bigl(\widehat{\langle \by,\by'\rangle}-\langle \by,\by'\rangle\bigr)^2\right]\le \frac{C}{m}\|\by\|_2^2\|\by'\|_2^2
	\end{equation}
	for a universal constant $C$. Expanding the second moment and averaging over the sketch randomness leaves only collision terms, which are bounded as stated. This is the standard guarantee~\cite{charikar2002} and we use it only to justify that the sketch is a controlled summary, not to derive the working size.
	In practice we choose $m$ by the heuristic $m\gtrsim c\sqrt{d}$, motivated by the sparsity and spatial clustering of showers, which have far fewer than $d$ effectively independent entries. At $d{=}12$ we use $m{=}32$, where the sketch is over-provisioned and distortion is negligible. Dataset~1 uses the same $m{=}32$ for prefixes of up to 368 values. This heuristic is not implied by the bound above.
	
	\section{Shot-noise accumulation}\label{app:noise}
	
	We bound the accumulation of measurement noise along the chain under a pessimistic model. Let $\hat{\bF}_\beta$ denote the measured feature matrix at block $\beta$, with $\E\|\hat{\bF}_\beta-\bF_\beta\|^2=\calO(p_f/k)$ from Born sampling of $p_f$ parities with $k$ records. Decoding by $\bW_\beta$ maps this to a pixel error of size at most $\kappa\,\calO(\sqrt{p_f/k})$, with $\kappa$ the decoder-gain bound above, and the sketch update injects it into the conditioning of every later block. Assuming errors add coherently across blocks, the end-to-end sketch perturbation after $B$ steps grows at most linearly in $B$ while falling as $k^{-1/2}$. Holding the end-to-end signal-to-noise ratio fixed as the image grows, with $B$ scaling linearly in $d$ at fixed block size, gives the sufficient budget quoted in Eq.~\eqref{eq:shots}.
	Both assumptions are deliberately conservative. Coherent addition is the worst case, and the ridge gain bound used is the loosest available. The resulting $\calO(d^2p_f)$ should be treated as an envelope, and the observation that $k{=}64$ suffices in our runs is consistent with, but not a test of, the bound.

	\section{Extension to 25-pixel images}\label{app:25px}
	
	We repeat the experiment at $d{=}25$ with the same $n_q{=}3$ circuit, now with $b{=}2$ and hence $B{=}13$ autoregressive steps, on both the simulator and \texttt{ibm\_fez}. The chain is more than twice as long as at $d{=}12$, the shared parameters must serve thirteen blocks, and the finer spatial resolution makes both marginals and correlations more demanding. The $d{=}25$ run uses $L{=}3$, 18 circuit parameters, $k{=}64$ records per block and 300 epochs. Hardware execution uses $k{=}64$ with the decoder refit on device-measured features and $n{=}200$ generated samples.

	\begin{figure}[b]
		\centering
		\includegraphics[width=\columnwidth]{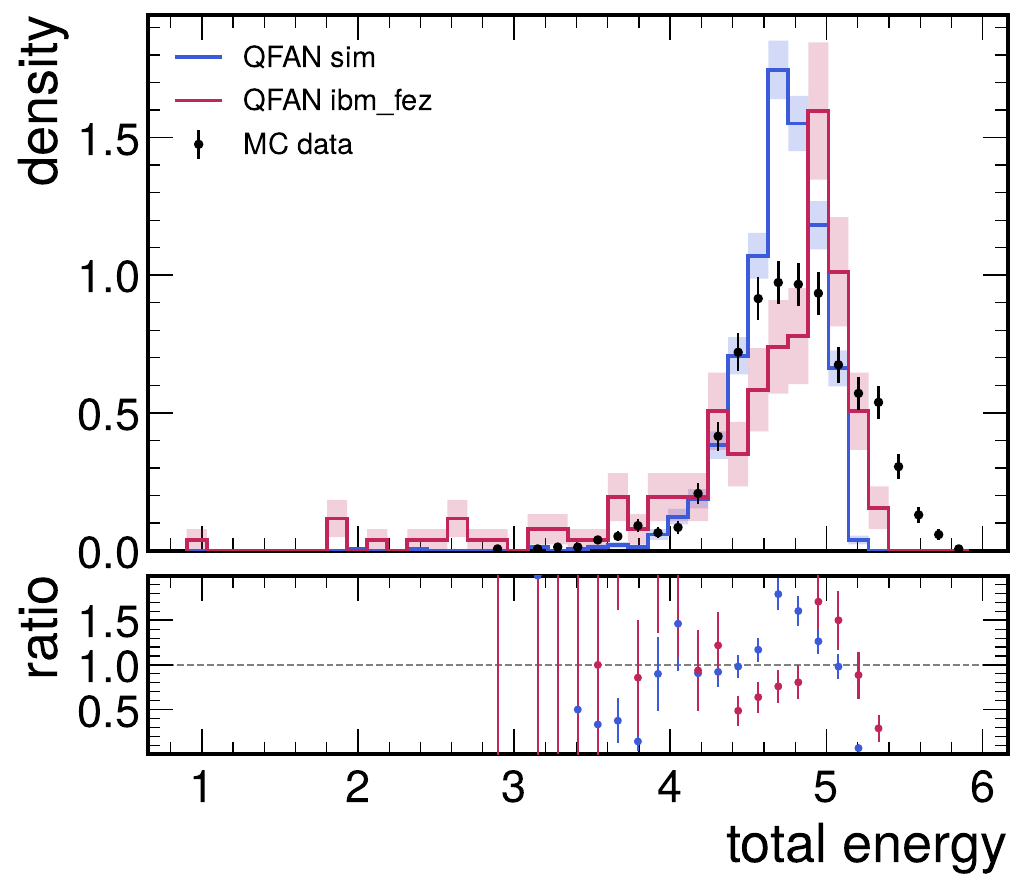}
		\caption{Total energy at $d{=}25$, $B{=}13$. MC data (black points), simulator QFAN (blue), hardware QFAN (red). Agreement across the full five-block chain indicates that the energy budget is propagated by the sketch.}
		\label{fig:energy_25px}
	\end{figure}
	
	\paragraph{Marginals.} Figure~\ref{fig:marginals_25px} shows all 25 marginals for MC data, simulator and hardware. Both backends reproduce the MC shapes, and the ratio panels show no systematic drift from the early to the late blocks, which is the behavior required if the sketch conditioning is to survive longer chains.
	
	\paragraph{Correlations.} Figure~\ref{fig:corr_25px} shows the $25\times25$ correlation matrices. The block-diagonal structure and the cross-block anticorrelations from the longitudinal development of the shower are present on both backends. Residuals remain concentrated at the block boundaries and, importantly for the scaling argument, do not increase from the first boundary to the last. Quantitatively, $\|\Delta\bm{C}\|_F/d$ is $0.152$ on the simulator and $0.299$ on \texttt{ibm\_fez}.
	
	The claim that boundary error does not accumulate along the chain can be made quantitative. Writing $\Delta_\beta$ for the absolute correlation error of the entry coupling the last pixel of block $\beta$ to the first pixel of block $\beta{+}1$, the twelve boundaries of the $d{=}25$ simulator run give
	\[
	\Delta_\beta = (0.260,\,0.103,\,0.063,\,0.027,\,0.020,\,0.001,
	\]
	\[
	\phantom{\Delta_\beta = (}0.033,\,0.300,\,0.086,\,0.016,\,0.079,\,0.117),
	\]
	with mean $0.092$. A linear fit in the boundary index has slope $-0.0033$ per step, i.e.\ mildly negative and consistent with zero. The last boundary ($0.117$) is smaller than the first ($0.260$), and the largest single value occurs in the middle of the chain ($0.300$ at the eighth boundary) rather than at its end. The variation between boundaries is therefore driven by which pixels they happen to couple, not by how far along the chain they sit, which is the behavior required if longer chains are to remain feasible. This is measured at $B{=}13$.
	
	\paragraph{Total energy.} Figure~\ref{fig:energy_25px} shows the total energy over the full chain of 13 blocks. Both paths track the MC peak and width. As at $d{=}12$, the informative feature is the width rather than the peak (Sec.~\ref{sec:results}). Quantitatively, $W_1^{(E)}{=}0.115$ (simulator) and $0.292$ (hardware), with $\Var(E){=}0.076$ and $0.544$ respectively against a data value of $0.172$. The simulator under-disperses the total energy and the hardware over-disperses it.
	
	\paragraph{Reading of this appendix.} The $d{=}25$ run doubles the chain length relative to the 12-pixel study and is the largest hardware execution in this work. It is evidence that the architecture does not break immediately as the chain grows. 
	
	\bibliography{qfan_references}
	
\end{document}